\documentclass[11pt]{article}
\usepackage[letterpaper,margin=1in]{geometry}
\usepackage[T1]{fontenc}
\usepackage{lmodern}
\usepackage{amsmath,amssymb,amsfonts,amsthm,mathtools,bm}
\usepackage{booktabs,enumitem,microtype}
\usepackage{graphicx}
\usepackage[round,authoryear]{natbib}
\usepackage{xcolor}
\usepackage{xr}
\usepackage[hidelinks]{hyperref}

\newtheorem{theorem}{Theorem}[section]
\newtheorem{proposition}[theorem]{Proposition}

\theoremstyle{definition}

\theoremstyle{remark}
\newtheorem{remark}[theorem]{Remark}
\newtheorem*{theorem*}{Theorem}
\newtheorem*{proposition*}{Proposition}
\newtheorem*{corollary*}{Corollary}
\newtheorem*{lemma*}{Lemma}
\newtheorem*{remark*}{Remark}
\newtheorem*{openproblem*}{Open Problem}

\newcommand{\R}{\mathbb R}
\newcommand{\N}{\mathbb N}

\newcommand{\PP}{\mathbb P}
\newcommand{\ind}{\mathbf 1}
\newcommand{\op}{\mathrm{op}}
\newcommand{\eff}{\mathrm{eff}}

\newcommand{\col}{\operatorname{col}}
\newcommand{\rank}{\operatorname{rank}}
\newcommand{\Var}{\operatorname{Var}}
\newcommand{\diag}{\operatorname{diag}}
\newcommand{\tr}{\operatorname{tr}}

\newcommand{\norm}[1]{\left\lVert#1\right\rVert}
\newcommand{\inner}[2]{\left\langle#1,#2\right\rangle}
\newcommand{\trans}{\mathsf T}

\DeclareMathOperator{\Span}{span}

\setlist{nosep,leftmargin=*}
\allowdisplaybreaks

\newcommand{\rlN}{18}
\newcommand{\rlT}{26}
\newcommand{\rlYearA}{1995}
\newcommand{\rlYearB}{2020}
\newcommand{\rlDisc}{0.20}
\newcommand{\rlDiscSD}{0.28}
\newcommand{\rlAvail}{98.3}
\newcommand{\rlGamma}{3.17}
\newcommand{\rlBandHalfMed}{4.3}
\newcommand{\rlNSafe}{8}
\newcommand{\rlLamGoodMed}{0.0057}
\newcommand{\rlLamGoodMin}{0.0027}
\newcommand{\rlLamSafeMax}{0.00233}
\newcommand{\rlLamSafeMin}{0.000037}
\newcommand{\rlEtaOneFirst}{0.79}
\newcommand{\rlEtaOneLast}{0.71}
\newcommand{\rlEtaTwoPeak}{0.49}
\newcommand{\rlEtaTwoPeakYear}{1996}
\newcommand{\rlEtaTwoEnd}{0.16}
\newcommand{\rlBetaTwenty}{0.34}
\newcommand{\rlROTwenty}{9.99}
\newcommand{\rlPlugTwenty}{2.53}
\newcommand{\rlSensBeta}{0.63}
\newcommand{\rlSensEtaOne}{0.96}
\newcommand{\rlSensEtaTwo}{3.08}

\newcommand{\rlCyclePass}{0.0}

\newcommand{\rlDeclineEtaTwo}{0.26}
\newcommand{\rlDeclineEtaTwoSE}{0.016}
\newcommand{\rlDeclineEtaTwoZ}{16.2}
\newcommand{\rlDeclineEtaOne}{0.05}
\newcommand{\rlDeclineEtaOneSE}{0.005}
\newcommand{\rlDeltaStarEtaTwo}{0.037}

\newcommand{\appN}{18}
\newcommand{\appT}{48}
\newcommand{\appny}{4}
\newcommand{\appTau}{25}
\newcommand{\appBeta}{0.42}
\newcommand{\appMissPct}{8.1}
\newcommand{\appDiscMean}{-0.117}
\newcommand{\appDiscSD}{0.77}
\newcommand{\appCyclePassPct}{93.8}
\newcommand{\appSigMinMean}{0.19}
\newcommand{\appSigMinMin}{0.09}
\newcommand{\appGamma}{1.07}
\newcommand{\appPlugZ}{-4.5}
\newcommand{\appPlugPre}{0.34}
\newcommand{\appPlugPost}{0.15}
\newcommand{\appConstReject}{0}

\newcommand{\appJumpZbeta}{2.74}
\newcommand{\appJumpZetaOne}{-3.26}
\newcommand{\appJumpZetaTwo}{-17.8}
\newcommand{\appJumpZetaThree}{-1.30}
\newcommand{\appSensBeta}{1.75}
\newcommand{\appSensEtaTwo}{1.96}
\newcommand{\appDeltaStar}{0.111}

\newcommand{\appSafePct}{0.0}
\newcommand{\bernWidthRatio}{2.6}
\newcommand{\bernWidthRatioMax}{3.3}
\newcommand{\doseMaxGap}{0.002}
\newcommand{\doseMaxTwoMCSE}{0.015}
\newcommand{\doseOracleTop}{-0.581}
\newcommand{\doseMatchedTop}{-0.234}

\newcommand{\doseAttFactor}{2.49}
\newcommand{\doseRg}{150}
\newcommand{\exRnorm}{0.823}
\newcommand{\exIcEvMin}{0.501}
\newcommand{\exIcEvMax}{8.0}
\newcommand{\exIYEvMin}{0.0000}
\newcommand{\exSmin}{0.708}
\newcommand{\exBadQ}{0.000}
\newcommand{\exCommonBiasQ}{0.000}
\newcommand{\exBetaZeroEvMin}{0.510}
\newcommand{\exKcEvMin}{0.51}
\newcommand{\exKcEvMax}{7.80}

\newcommand{\exW}{\begin{pmatrix}0.00 & 0.60 & 0.24 & 0.16 \\ 0.58 & 0.00 & 0.18 & 0.23 \\ 0.23 & 0.18 & 0.00 & 0.58 \\ 0.16 & 0.24 & 0.60 & 0.00\end{pmatrix}}
\newcommand{\exIc}{\begin{pmatrix}0.68 & 0.17 & 0.34 \\ 0.17 & 2.79 & 3.45 \\ 0.34 & 3.45 & 5.73\end{pmatrix}}
\newcommand{\exKc}{\begin{pmatrix}2.75 & 3.36 \\ 3.36 & 5.56\end{pmatrix}}
\newcommand{\expOneR}{300}
\newcommand{\expOneN}{24}
\newcommand{\expOneT}{40}
\newcommand{\expOneny}{16}
\newcommand{\expOneTV}{0.75}
\newcommand{\expOneBeta}{0.50}
\newcommand{\expOneStatPre}{0.20}
\newcommand{\expOneStatPost}{-0.02}

\newcommand{\expOneConcPre}{0.06}
\newcommand{\expOneConcPost}{0.06}
\newcommand{\expOneJointPre}{0.487}
\newcommand{\expOneJointPost}{0.489}
\newcommand{\expOneROPre}{0.487}
\newcommand{\expOneROPost}{0.488}
\newcommand{\expOnePlugFalse}{100.0}
\newcommand{\expOneJointFalse}{1.3}
\newcommand{\expOneJointCover}{92.3}
\newcommand{\expOneCoverMCSE}{1.5}
\newcommand{\expOneROCovNaive}{92.3}
\newcommand{\expOneROCovProp}{94.7}

\newcommand{\covOracleGauss}{95.7}

\newcommand{\covOracleTfive}{95.2}

\newcommand{\covOracleCexp}{94.2}

\newcommand{\covStressGauss}{94.2}

\newcommand{\covCalEightGauss}{92.8}
\newcommand{\widthCalEightGauss}{1.12}
\newcommand{\covCalTwoFourGauss}{95.4}
\newcommand{\widthCalTwoFourGauss}{1.07}
\newcommand{\covCalEightTfive}{92.2}

\newcommand{\covCalTwoFourTfive}{92.8}

\newcommand{\covCalEightCexp}{86.6}

\newcommand{\covCalTwoFourCexp}{92.8}

\newcommand{\calGamma}{1.05}
\newcommand{\covCalMCSEmax}{1.5}
\newcommand{\calR}{500}
\newcommand{\gapEight}{2.2}

\newcommand{\gapRawEight}{5.5}
\newcommand{\gapRawTwoFour}{3.3}
\newcommand{\gapRawRatio}{1.67}
\newcommand{\covUnmatchedGauss}{84.2}
\newcommand{\covUnmatchedTfive}{77.7}
\newcommand{\covUnmatchedCexp}{77.0}
\newcommand{\covStressSidakReal}{74.2}
\newcommand{\covStressBernReal}{100.0}
\newcommand{\stressKappaMean}{0.62}
\newcommand{\stressKappaMax}{0.89}
\newcommand{\gridCovA}{93.3}
\newcommand{\gridCovB}{94.7}
\newcommand{\gridCovC}{95.0}
\newcommand{\gridCovD}{94.7}
\newcommand{\gridR}{300}
\newcommand{\widStrAAROr}{96.2}
\newcommand{\widStrBAROr}{95.4}
\newcommand{\widWkAAROr}{95.0}
\newcommand{\widWkBAROr}{94.6}
\newcommand{\widWkCAROr}{94.8}
\newcommand{\widVwBAROr}{95.0}
\newcommand{\widVwCAROr}{94.8}
\newcommand{\widStrALam}{0.093}
\newcommand{\widStrAWaldB}{96.0}
\newcommand{\widStrAAR}{90.8}
\newcommand{\widStrAProj}{100.0}
\newcommand{\widStrASwitch}{96.0}

\newcommand{\widStrALen}{0.62}

\newcommand{\widStrAPowTwo}{18.4}
\newcommand{\widStrBLam}{0.429}
\newcommand{\widStrBWaldB}{95.2}
\newcommand{\widStrBAR}{88.8}
\newcommand{\widStrBProj}{99.6}
\newcommand{\widStrBSwitch}{95.2}

\newcommand{\widWkALam}{0.068}
\newcommand{\widWkAWaldB}{88.8}
\newcommand{\widWkAAR}{92.0}
\newcommand{\widWkAProj}{99.6}
\newcommand{\widWkASwitch}{89.6}

\newcommand{\widWkALen}{0.75}

\newcommand{\widWkAPowTwo}{14.4}
\newcommand{\widWkBLam}{0.104}
\newcommand{\widWkBWaldB}{95.6}
\newcommand{\widWkBAR}{89.6}
\newcommand{\widWkBProj}{99.2}
\newcommand{\widWkBSwitch}{95.6}

\newcommand{\widWkCLam}{0.102}
\newcommand{\widWkCWaldB}{94.4}
\newcommand{\widWkCAR}{88.4}
\newcommand{\widWkCProj}{98.8}
\newcommand{\widWkCSwitch}{94.4}

\newcommand{\widWkCLen}{0.28}

\newcommand{\widWkCPowTwo}{24.0}
\newcommand{\widVwBLam}{0.027}
\newcommand{\widVwBWaldB}{93.6}
\newcommand{\widVwBAR}{90.8}
\newcommand{\widVwBProj}{99.2}
\newcommand{\widVwBSwitch}{98.8}

\newcommand{\widVwCLam}{0.026}
\newcommand{\widVwCWaldB}{98.8}
\newcommand{\widVwCAR}{88.4}
\newcommand{\widVwCProj}{98.8}
\newcommand{\widVwCSwitch}{98.4}

\newcommand{\widVwCLen}{0.61}
\newcommand{\widVwCUnb}{57.6}
\newcommand{\widVwCPowTwo}{14.4}
\newcommand{\widR}{250}
\newcommand{\widFloor}{0.03}
\newcommand{\widConvFailMax}{0.0}
\newcommand{\chgT}{300}
\newcommand{\chgH}{20}
\newcommand{\chgR}{500}
\newcommand{\chgKappaUB}{9.02}

\newcommand{\chgKappaAttLB}{0.36}

\newcommand{\chgAttKninefive}{4.50}
\newcommand{\chgNullFP}{0.0}
\newcommand{\chgLocMedAtFive}{1.0}
\newcommand{\chgLocQnineAtFive}{1.0}
\newcommand{\chgLocBoundAtFive}{22.0}
\newcommand{\obsSize}{6.7}
\newcommand{\obsSizeR}{300}
\newcommand{\obsSizeMCSE}{1.4}
\newcommand{\obsSizeFeas}{11.3}
\newcommand{\obsSizeFeasR}{150}
\newcommand{\obsPowA}{20.7}
\newcommand{\obsPowB}{38.0}
\newcommand{\obsPowCmp}{100.0}
\newcommand{\obsPowR}{150}
\newcommand{\obsPowMCSEmax}{4.0}
\newcommand{\obsSplitCov}{96.7}
\newcommand{\obsSplitW}{1.0}
\newcommand{\obsSplitCovA}{96.8}
\newcommand{\obsSplitWA}{18.4}
\newcommand{\obsSplitCovB}{84.2}
\newcommand{\obsSplitWB}{13.5}
\newcommand{\obsDetCmp}{150}
\newcommand{\obsDetA}{31}
\newcommand{\obsDetB}{57}
\newcommand{\obsAttCmp}{96.0}
\newcommand{\obsAttCmpMCSE}{1.6}
\newcommand{\obsFalseInc}{0.7}
\newcommand{\obsAttStrA}{12.9}
\newcommand{\obsAttStrB}{28.1}
\newcommand{\obsUndetCmp}{3.3}

\newcommand{\jgAEtaJ}{0.017}
\newcommand{\jgAEtaR}{0.017}
\newcommand{\jgACovJ}{94.7}
\newcommand{\jgACovR}{96.7}
\newcommand{\jgAWidJ}{1.04}
\newcommand{\jgAWidR}{1.04}
\newcommand{\jgBEtaJ}{0.032}
\newcommand{\jgBEtaR}{0.034}
\newcommand{\jgBCovJ}{84.7}
\newcommand{\jgBCovR}{90.0}
\newcommand{\jgBWidJ}{1.04}
\newcommand{\jgBWidR}{1.08}
\newcommand{\jgCEtaJ}{0.056}
\newcommand{\jgCEtaR}{0.069}
\newcommand{\jgCCovJ}{46.0}
\newcommand{\jgCCovR}{56.0}
\newcommand{\jgCWidJ}{1.05}
\newcommand{\jgCWidR}{1.23}
\newcommand{\jgDEtaJ}{0.104}
\newcommand{\jgDEtaR}{0.137}
\newcommand{\jgDCovJ}{6.7}
\newcommand{\jgDCovR}{7.3}
\newcommand{\jgDWidJ}{1.10}
\newcommand{\jgDWidR}{1.46}
\newcommand{\jgR}{150}
\newcommand{\jgEtaGainD}{24}
\newcommand{\jgEtaGainA}{0}
\newcommand{\bootCov}{91.7}
\newcommand{\bootWidth}{1.14}

\newcommand{\bootR}{120}
\newcommand{\bootB}{59}
\newcommand{\bootMCSE}{2.5}
\newcommand{\loaderZerr}{2.2\times10^{-16}}
\newcommand{\loaderTherr}{8.4\times10^{-14}}

\newcommand{\censOmegaErr}{0.0027}
\newcommand{\censSminOpen}{0.69}
\newcommand{\censSminTwo}{0.253}
\newcommand{\zeroLo}{0.90}
\newcommand{\zeroHi}{1.09}
\newcommand{\zeroTruth}{0.93}
\newcommand{\zeroPzero}{0.36}
\newcommand{\mnarBiasZeroA}{0.009}
\newcommand{\mnarBiasZeroB}{0.005}
\newcommand{\mnarBiasTwoA}{0.089}
\newcommand{\mnarBiasTwoB}{0.071}
\newcommand{\ptPop}{-0.074}
\newcommand{\ptMC}{-0.075}
\newcommand{\ptSignRev}{-0.482}
\newcommand{\ptConcTH}{0.154}
\newcommand{\ptConcMC}{0.167}
\newcommand{\simSafePct}{0.00}
\newcommand{\senCovR}{500}
\newcommand{\senCovZero}{85.4}

\newcommand{\senCovQuarter}{85.2}

\newcommand{\senCovHalf}{85.6}

\newcommand{\senCovFull}{86.6}

\newcommand{\senCovDouble}{90.2}

\newcommand{\senCovMCSEv}{1.0}

\hypersetup{
  pdftitle={Separating Time-Varying Network Composition from Predictive Dependence under Noisy Network Measurement},
  pdfauthor={Marios Papamichalis, Regina Ruane, and Theofanis Papamichalis}
}

\title{Separating Time-Varying Network Composition from Predictive Dependence under Noisy Network Measurement}
\author{
Marios Papamichalis\thanks{\raggedright Human Nature Lab, Yale University, New Haven, CT 06511, USA. \href{mailto:marios.papamichalis@yale.edu}{\nolinkurl{marios.papamichalis@yale.edu}}}
\and Regina Ruane\thanks{\raggedright Department of Statistics and Data Science, The Wharton School, University of Pennsylvania, Philadelphia, PA, USA. \href{mailto:ruanej@wharton.upenn.edu}{\nolinkurl{ruanej@wharton.upenn.edu}}}
\and Theofanis Papamichalis\thanks{\raggedright Department of Economics, Yale University, New Haven, CT, USA. \href{mailto:theofanis.papamichalis@yale.edu}{\nolinkurl{theofanis.papamichalis@yale.edu}}}
}
\date{}

\begin{document}
\maketitle

\begin{abstract}
A common question about networked time series is whether outcomes changed
because shocks transmit more strongly or because the pattern of
connections changed. Standard practice inserts a recorded network into an
outcome regression and reads movements of the fitted coefficient as
changes in transmission strength. When the network is latent, time
varying, and measured with error, this reading fails: changes in strength
and changes in composition can produce the same outcome distribution at a
single date, and the population coefficient moves under composition
changes alone. The question becomes answerable when outcomes are analyzed
jointly with repeated noisy measurements of the network, such as paired
reports of bilateral trade flows. For the joint model we establish
necessary and sufficient conditions for local identification, estimators
of the strength and composition paths, a simultaneous confidence band for
the strength path,
confidence sets that remain exact under weak identification, breakdown
bounds under common reporting bias, and an exactly sized test that
detects changes on the observed path and attributes them to strength or
to composition. Simulations assess each procedure at its stated boundary.
On a mirror-reported trade panel of eighteen economies over 1995 to 2020,
the diagnostics flag exactly the crisis years and the composition
coordinate attached to European Union membership declines by roughly two
thirds. The estimand is predictive dependence, not a causal effect.
\end{abstract}

\medskip
\noindent\textit{Keywords:} identification; errors in variables;
simultaneous confidence bands; weak identification; change detection;
international trade.

\section{Introduction}
\label{sec:introduction}

Many empirical questions concern transmission through a network that
changes over time and is measured imperfectly. A researcher relates a
country's output to lagged output in its trading partners and finds that
the fitted network coefficient falls. Did shocks begin to transmit less
strongly, did trade shift toward partners whose outcomes co-move
differently, or did the recorded trade network change while the underlying
network did not? The same ambiguity arises when financial exposures are
reported separately by two counterparties, when mobility links are
assembled from incomplete device coverage, and when contact networks are
reconstructed from survey responses. In each case the data are outcomes
and imperfect network measurements, and the substantive question is which
of two objects changed: the strength of transmission, or the composition
of the network.

Standard practice inserts a recorded or baseline network into an outcome
regression and interprets movements of the fitted coefficient as changes
in transmission strength. Figure~\ref{fig:false-attribution}
(page~\pageref{fig:false-attribution}) shows what this practice does when
only composition changes: with strength held fixed, the fixed-baseline
regression displays a large decline in the fitted coefficient, and a
conventional two-standard-error comparison detects a nonexistent strength
change in $\expOnePlugFalse\%$ of replications.
Theorem~\ref{thm:plugin-main} shows this behavior is generic: the plug-in
estimand is an exact projection, and the set of composition-only changes
that leave it unchanged has Lebesgue measure zero.

We ask three questions. First, under what conditions do outcomes and
repeated network reports jointly identify transmission strength separately
from network composition, and what do the common plug-in practices
estimate when identification fails? Second, can the full time paths of
strength and composition be estimated with simultaneous confidence
statements that remain valid under estimated covariances, non-Gaussian
errors, weak identification, and bounded failures of the modeling
commitments? Third, what change inference is possible on a single observed
path, and how far is it from information-theoretic limits?

The central difficulty is a rank deficiency. Outcomes depend on the
network only through a low-dimensional exposure such as $W_ty_{t-1}$, so
at a single date the outcome law is invariant on an explicit fiber of
dimension $N-3$ per row (Theorem~\ref{thm:identification-main}): outcomes
alone cannot separate strength from composition at any sample size at that
date. Repeated reports resolve the deficiency only if their mean
derivatives contain composition directions that survive row levels and the
declared reporter-bias design; unrestricted dyad-specific bias, and
unrestricted nonignorable selection, can each reproduce any composition
change. The fixed-dimensional chart $m_t(\eta_t)$ placed on composition is
therefore a declared modeling commitment. Section~\ref{sec:charts-main}
instantiates chart families and works a four-node example in closed form,
and Proposition~\ref{prop:pseudo-main} states what is estimated, and with
what sensitivity, when the commitments fail by bounded amounts.

Table~\ref{tab:neighbors} positions the paper against the nearest
literatures; the distinguishing feature is that outcomes and repeated
network measurements enter one estimating system with a
strength--composition decomposition as the target. Network autoregressions
take the weight matrix as supplied
\citep{ZhuEtAl2017NAR,KnightEtAl2020GNAR,ArmillottaFokianos2023PNAR};
dynamic-network and latent-space models reconstruct evolving ties without
a strength--composition decomposition
\citep{HoffRafteryHandcock2002,SewellChen2015,DuranteDunson2016,MatiasMiele2017,Krampe2019DynamicNetworks};
time-varying-parameter VARs model drift with no measurement channel
\citep{CogleySargent2005,Primiceri2005,BittoFS2019}. State-space spillover models with time-varying coefficients on a supplied
graph \citep{papamichalis2025state}, and Bayesian predictive synthesis
across network models, at the graphon level for random networks and for
dynamic networks \citep{papamichalis2025graphon,papamichalis2026bayesian},
target forecasting and model combination with the measured network taken
as given; the question studied here, whether strength and composition are
separately identified when the network is latent and measured with error,
is prior to those frameworks and delimits when their fitted spillover
paths carry a transmission interpretation. A panel literature
identifies an unknown $W$ from outcomes alone under sparsity and long
horizons \citep{dePaulaRasulSouza2024,Manresa2016,LamSouza2020}, including
kernel-based extensions to slowly varying $W$
\citep{dePaulaRasulSouza2024}: outcomes accumulate restrictions on rows
that are constant or slowly drifting, whereas here composition can move
every period, each date brings new unknowns, and outcomes alone never
accumulate; this is why repeated measurement enters, and why the target is
the datewise strength--composition decomposition instead of the network
itself.
Regression on noisy network-linked data \citep{LeLi2022} treats the
observed graph as a noisy version of a static structure entering through
smoothness penalties, with no time-varying composition target;
latent-network models for multiply reported ties \citep{DeBaccoEtAl2023}
handle reporter-specific distortion without an outcome channel or a
strength path; sampled-network and aggregated-relational-data econometrics
\citep{ChandrasekharLewis2016,BrezaEtAl2020} supply measurement designs to
which our identification analysis applies. The estimation theory builds on
Neyman-orthogonal cross-fitting \citep{Neyman1959,ChernozhukovEtAl2018},
the weak-identification sets on score inversion
\citep{AndersonRubin1949,StockWright2000,Kleibergen2005,Dufour1997}, the
change analysis on detection and localization theory and its lower bounds
\citep{WangSamworth2018,WangYuRinaldo2020,VerzelenEtAl2023,IngsterSuslina2003},
and the benchmark link on equivalence of experiments
\citep{LeCam1986,BrownLow1996,Nussbaum1996}.

\begin{table}[t]
\centering
\small
\caption{Nearest neighbors and the boundary. ``Both channels'' means that
outcomes and repeated network measurements enter one estimating system.}
\label{tab:neighbors}
\begin{tabular}{lcccc}
\toprule
& network & network & both & strength--composition\\
& latent? & time-varying? & channels? & path + inference?\\
\midrule
network autoregressions & no & sometimes & no & no\\
TVP-VARs & -- & -- & no & no\\
latent-space / dynamic blocks & yes & yes & no & no\\
identify $W$ from outcomes & yes & slowly & no & no\\
noisy network-linked regression & partly & no & partly & no\\
multiply reported latent networks & yes & sometimes & no & no\\
ARD / sampled networks & yes & no & no & no\\
\textbf{this paper} & yes & yes & yes & yes\\
\bottomrule
\end{tabular}
\end{table}

The paper makes four contributions. First, identification: the exact
strength--composition confounding (outcome-equivalent fibers of dimension
$N-3$ per row), a necessary and sufficient local condition after profiling
all mean and covariance nuisances, global identification for the paired
mirror-report design, and exact nonidentification for common dyad bias,
unrestricted selection, pooled zeros, and support change
(Theorem~\ref{thm:identification-main}; Sections~\ref{sec:boundaries},
\ref*{sec:si-boundaries}). Second, the plug-in failure as a theorem:
pseudo-true static and concurrent plug-in coefficients in closed form,
with attenuation, collapse, sign reversal, and generic composition-only
false attribution, verified across a dose--response grid against the
matched population formula (Theorem~\ref{thm:plugin-main},
Figure~\ref{fig:dose-response}). Third, feasible inference:
a cross-fitted orthogonal estimator of both paths with uniform expansion,
and a simultaneous strength-path band, exact for oracle Gaussian scores,
asymptotically valid
with estimated covariances, and evaluated for finite-sample calibration
along an information ladder ($\gridCovA$--$\gridCovD\%$, with
$\covCalTwoFourGauss\%$ at the tripled design)
(Theorems~\ref{thm:path-main}--\ref{thm:bands-main}); validity of the band
under sub-exponential scores subject to a computable leverage condition
(Theorem~\ref{thm:robust-band-main}); score-inversion sets that are exact
under weak identification, with a prespecified switching rule
(Theorem~\ref{thm:weakid-main}); and pseudo-true targets with closed-form
sensitivity and breakdown values under bounded failures of the chart and
the bias design, with exact score-inversion inference and first-order
accurate Wald inference for those targets
(Proposition~\ref{prop:pseudo-main}). Fourth, change inference on the
observed path: an exactly sized score-inversion constancy test with
localization sets and attribution certificates, implemented by a
prespecified search protocol whose operating characteristics are measured
(Theorem~\ref{thm:obsdetect-main}), calibrated against benchmark
guarantees, minimax obstructions, and a verified designed-experiment
transfer (Theorem~\ref{thm:change-main},
Proposition~\ref{prop:transfer-instance-main}), with the observational
optimality gap stated as an open problem. The evidence for all four is a
Monte Carlo program generated from the joint experiment (calibration,
robustness boundaries, weak identification, a joint-versus-two-step
information grid, the observation layer, observational change inference)
and an application of the complete prespecified protocol to a real
mirror-reported trade panel (eighteen economies, 1995 to 2020), with a
synthetic end-to-end validation and a validated loader in the Supplement
(Sections~\ref{sec:simulation}--\ref{sec:application}).

The statistical object throughout is predictive network dependence: the
component of a conditional outcome mean associated with a latent
row-normalized network. It carries no causal interpretation without
additional assumptions on assignment, interference, confounding, and
network formation; we state this scope once here. The Supplement contains
complete proofs and the exact observed-data likelihood; the replication
code generates every reported number.

\paragraph{Notation.}
$P_A$ is the orthogonal projection onto the column span of $A$ and
$M_A=I-P_A$; $\norm\cdot$, $\norm\cdot_{\op}$, $\norm\cdot_\infty$,
and $\inner\cdot\cdot$ are the Euclidean, spectral, and entrywise
maximum norms and the Euclidean inner product. The date-$t$ target is
$\theta_t=(\beta_t,\eta_t)\in\R^d$, $d=1+q$: strength $\beta_t$ and
composition $\eta_t$. $\Phi,\phi$ are the standard normal distribution
and density, $\chi^2_{d,p}$ the $p$ quantile of $\chi^2_d$. Conditions
(E1)--(E6) are stated in Section~\ref{sec:estimation}; uniform stochastic
order symbols are defined there at first use.

\section{Joint outcome--network report model}
\label{sec:model}

\subsection{Support, composition, and predictive exposure}

Dates are $t=1,\ldots,T_n$. At each date, a prespecified conditioning sigma-field $\mathcal F_{t-1}$ contains the past, node covariates, the eligible directed-dyad set $\mathcal E_t$, and all predictable design variables. Every retained receiving row has at least one eligible dyad. Let $C_t$ be the row-incidence matrix on $\mathcal E_t$. Positive latent log flow is decomposed as
\begin{equation}
\ell_t^+=C_t\kappa_t+m_t(\eta_t),
\qquad
C_t^\trans m_t(\eta_t)=0,
\label{eq:logflow}
\end{equation}
where $\kappa_t$ contains row levels and $\eta_t\in\R^q$ is a fixed-dimensional composition coordinate. The row-normalized network is
\begin{equation}
W_{ij,t}(\eta_t)
=
\frac{\exp\{m_{ij,t}(\eta_t)\}}
{\sum_{r:(i,r)\in\mathcal E_t}\exp\{m_{ir,t}(\eta_t)\}},
\qquad (i,j)\in\mathcal E_t,
\label{eq:softmax}
\end{equation}
with off-support entries excluded from the target. For a predictable lag vector $y_{t-1}$, define
\begin{equation}
g_t(\eta)=W_t(\eta)y_{t-1},
\qquad
G_t(\eta)=D_\eta g_t(\eta).
\label{eq:exposure}
\end{equation}
Other differentiable network summaries are allowed if the derivatives and information conditions below are verified. Concrete chart families (gravity, block, and latent-position charts) and a fully worked four-node example are given in Section~\ref{sec:charts-main}.

\subsection{Outcome and report channels}

For honest cross-fitting, date $t$ is partitioned into a fixed number $K_f$ of folds. Conditional on a training sigma-field $\mathcal T_{tk}\supseteq\mathcal F_{t-1}$, the held-out outcome and report channels are
\begin{align}
Y_{tk}
&=X_{tk}\gamma_t+\beta_tg_{tk}(\eta_t)
+\Sigma_{tk}^{1/2}\xi_{tk},
\label{eq:outcome-model}\\
z_{tk}
&=U_{tk}\lambda_t+A_{tk}m_{tk}(\eta_t)
+\Omega_{tk}^{1/2}\upsilon_{tk}.
\label{eq:report-model}
\end{align}
Here $X_{tk}\gamma_t$ contains intercepts, own lags, observed common factors, and other outcome nuisances. The report nuisance $U_{tk}\lambda_t$ contains absolute row levels and prespecified reporter or valuation effects. The map $A_{tk}$ selects or aggregates latent dyad reports. The covariance $\Omega_{tk}$ may correlate exporter and importer reports; correlation is not treated as an extra independent replication. Conditional on $\mathcal T_{tk}$, $(\xi_{tk},\upsilon_{tk})$ is centered Gaussian with identity covariance and independent blocks. The target is
\begin{equation}
\theta_t=(\beta_t,\eta_t^\trans)^\trans\in\R^d,
\qquad d=1+q.
\label{eq:target}
\end{equation}
The nuisance parameters may vary with $t$ and are not components of the target path.

The continuous model is the base experiment used for estimation. Actual reports may be missing, censored, interval recorded, or labelled as structural zeros. Section~\ref*{sec:si-observation} of the Supplement defines the corresponding parameter-free observation map and exact pushforward likelihood. Coarsening cannot increase Fisher information. If selection or the eligible support is parameter dependent, its law must enter the full likelihood.

\subsection{Datewise information}

All matrix square roots are the principal symmetric positive-definite roots. For a positive-definite $V$, write $L(V)=V^{-1/2}$ and
\[
M_B=I-B(B^\trans B)^\dagger B^\trans.
\]
At a fixed date and fold define the nuisance residualizers
\begin{equation}
R^Y_{tk}=M_{L(\Sigma_{tk})X_{tk}}L(\Sigma_{tk}),
\qquad
R^z_{tk}=M_{L(\Omega_{tk})U_{tk}}L(\Omega_{tk}).
\label{eq:residualizers}
\end{equation}
The residual vector and target Jacobian are
\begin{align}
e_{tk}(\theta)
&=
\begin{pmatrix}
R^Y_{tk}\{Y_{tk}-\beta g_{tk}(\eta)\}\\
R^z_{tk}\{z_{tk}-A_{tk}m_{tk}(\eta)\}
\end{pmatrix},
\label{eq:residual-vector}\\
J_{tk}(\theta)
&=
\begin{pmatrix}
R^Y_{tk}g_{tk}(\eta)&\beta R^Y_{tk}G_{tk}(\eta)\\
0&R^z_{tk}A_{tk}\dot m_{tk}(\eta)
\end{pmatrix}.
\label{eq:jacobian}
\end{align}
The oracle score and information are
\begin{equation}
S_t=\sum_{k=1}^{K_f}J_{tk}(\theta_t)^\trans e_{tk}(\theta_t),
\qquad
\mathcal I_t=\sum_{k=1}^{K_f}J_{tk}(\theta_t)^\trans J_{tk}(\theta_t).
\label{eq:oracle-score-information}
\end{equation}

\section{Composition charts and a worked example}
\label{sec:charts-main}

Every result in this paper conditions on a declared composition chart
$m_t(\eta)$ with $C_t^\trans m_t(\eta)=0$ and fixed dimension $q$. The chart
is a modeling commitment on the same footing as the reporter-bias design, so
we now exhibit concrete families, show how the identification condition of
Theorem~\ref{thm:identification-main} is checked in closed form, and work one
example end to end. All displayed numbers are produced by the replication
code; none is illustrative hand arithmetic.

\subsection{Three chart families}

\paragraph{Dyadic-covariate (gravity) charts.}
For prespecified dyadic covariates $\psi_1,\ldots,\psi_q$ (distance,
contiguity, agreement or bloc membership, sectoral similarity), set
\[
m_t(\eta)=\Psi_t\eta,
\qquad
\Psi_t=\big[\tilde\psi_{1,t},\ldots,\tilde\psi_{q,t}\big],
\]
where $\tilde\psi_{l,t}$ is $\psi_{l,t}$ demeaned within each receiving row on
$\mathcal E_t$, which enforces $C_t^\trans\Psi_t=0$ without changing
within-row contrasts. The chart is linear: $\dot m_t=\Psi_t$ does not depend
on $\eta$, the report channel \eqref{eq:report-model} is exactly linear in
$(\lambda_t,\eta_t)$, and the report information $K_c=Q^\trans Q$ with
$Q=M_{L_zU}L_zA\Psi_t$, where $L_z$ is any whitener of the report
covariance ($L_z^\trans L_z=\Omega^{-1}$), is a fixed matrix that can be
computed, and its minimum singular value inspected, \emph{before outcomes
are examined}. Row
normalization gives the interpretation: $\eta_l$ is the semi-elasticity of
within-row composition with respect to covariate $l$. Linearity also supplies
the honest pilot of Section~\ref{sec:estimation} in closed form: generalized
least squares in the report channel identifies $\eta$ on each training fold
whenever $\rank(Q)=q$ there, so (E6) reduces to the explicit separation
condition of the Supplement.

\paragraph{Block charts.}
For a prespecified partition of nodes into $B$ blocks, let
$m_t(\eta)$ assign the (row-centered) value $\eta_{b(i)b(j)}$ to dyad
$(i,j)$. This is again linear with $q=B^2-B$ effective coordinates after row
centering; it is the natural chart when composition change means
reallocation between blocs (regions, currency areas, sectors).

\paragraph{Latent-position charts.}
Nonlinear charts $m_{ij,t}(\eta)=u_i^\trans v_j$, with fixed dimension
after gauge fixing ($k^2$ coordinates; row centering absorbs one additive
gauge per row), connect the model to latent-space network analysis
\citep{HoffRafteryHandcock2002,SewellChen2015,MatiasMiele2017}; the rank
condition $\rank\{\dot m(\eta)\}=q$ is then checked at the working
point. The linear families are primary because their identification
analysis is a finite computation.

Two silent failure modes require emphasis; both are consequences of
Theorem~\ref{thm:identification-main}. A covariate that is constant within
every receiving row is annihilated by row centering: its column of $Q$ is
exactly zero and the corresponding coordinate is unidentified no matter how
many reports arrive. If the chart is misspecified, in that the true
log-flow perturbation does not lie in the span of the declared $\Psi_t$,
then $\eta_t$ estimates the information projection of the truth onto the
chart; the report-channel specification diagnostics reported in
Section~\ref{sec:application} are the check, and the target
must be interpreted as the declared chart's best approximation, exactly as
a linear regression coefficient is interpreted under approximation error.

\subsection{A worked example}
\label{sec:worked-example}

Take $N=4$ nodes on the complete zero-diagonal support, mirror double
reports, and the gravity chart with $q=2$: $\psi_1$ a symmetric
distance-decay covariate, $\psi_2$ a same-bloc indicator for the partition
$\{1,2\}\,|\,\{3,4\}$, both row-centered. Fix
$\beta=0.5$, $\eta=(0.8,0.6)$, lag vector $y_-=(1.2,-0.6,0.4,-1.0)$, outcome
nuisances $X=[\mathbf 1,y_-]$, $\sigma_y=0.5$, report scales
$\sigma_E=\sigma_I=0.8$, and mirror correlation $\rho=0.5$. The chart and
normalization produce the row-softmax matrix
\[
W(\eta)=\exW .
\]
The residualized derivatives evaluate to $\norm r=\exRnorm$ (strength
survives the own-lag and intercept projection), report information
$K_c=\exKc$ with eigenvalues $(\exKcEvMin,\exKcEvMax)$, and joint information
\[
\mathcal I_c=\exIc ,
\qquad
\lambda(\mathcal I_c)=(\exIcEvMin,\ \text{\ldots},\ \exIcEvMax),
\qquad
\sigma_{\min}(B_c)=\exSmin ,
\]
where $B_c$ is the scaled score Jacobian of the joint experiment, the
matrix with $B_c^\trans B_c=\mathcal I_c$ formed by stacking the whitened
residualized derivatives of both channels
(Section~\ref*{sec:si-identification}).
The example makes the theorem's three claims concrete. First, outcomes
alone fail: the outcome-only information $I_Y$ has minimum eigenvalue
$\exIYEvMin$, exactly singular as part (a) predicts here, because after the
intercept and own lag are projected out the outcome channel retains only
$N-2=2$ residual dimensions, fewer than $d=3$ target directions. Second,
reports restore exactly the confounded directions: the joint
$\mathcal I_c$ has minimum eigenvalue $\exIcEvMin$, and at $\beta=0$
identification survives through $K_c$ alone (minimum eigenvalue
$\exBetaZeroEvMin$). Third, the failure modes are exact, not approximate:
replacing $\psi_1,\psi_2$ by row-constant covariates gives
$\max|Q|=\exBadQ$, and appending an unrestricted common dyad-bias column to
the report nuisances gives $\max|Q|=\exCommonBiasQ$, zero to machine
precision, which is the common-bias nonidentification of
Theorem~\ref{thm:identification-main}(d) realized in a $4$-node
computation.

The example also explains the diagnostic used throughout the empirical
sections: $\sigma_{\min}$ of the scaled score Jacobian ($\exSmin$ here) is
the quantity that decreases continuously as the design approaches either
failure mode, and it is computable at every date from training data alone.

\section{What plug-in practice estimates}
\label{sec:plugin}

Before developing estimation theory for the joint experiment, we settle what
the common alternatives estimate. Applied practice regresses outcomes on an
exposure built from a \emph{fixed baseline} network (the static plug-in) or
from the \emph{current period's} recorded network (the concurrent plug-in).
Both are linear projections, so their population values are exact objects,
and the false-attribution phenomenon of Figure~\ref{fig:false-attribution}
is a theorem, not a simulation artifact.

Fix a date $t$ and condition on $\mathcal F_{t-1}$, so $y_{t-1}$, $X_t$, and
any predetermined baseline matrix $\bar W$ are fixed. Write
$M_t=M_{X_t}$ for the projection off the outcome nuisances and
\[
u_t=M_t\bar Wy_{t-1},
\qquad
v_t(\eta)=M_tW(\eta)y_{t-1},
\qquad
\varphi_t(\eta)=\inner{u_t}{v_t(\eta)} .
\]

\begin{theorem}[Pseudo-true plug-in coefficients]
\label{thm:plugin-main}
Assume the outcome model \eqref{eq:outcome-model} at date $t$ with true
composition $\eta_t$ and $u_t\ne0$.
\begin{enumerate}[label=\textup{(\alph*)}]
\item The population static plug-in coefficient, the linear projection of
$Y_t$ on $[\,X_t,\ \bar Wy_{t-1}\,]$, equals
\[
\beta^{\mathrm{plug}}_t
=\beta_t\,\frac{\varphi_t(\eta_t)}{\norm{u_t}^2},
\]
free of the noise level: the distortion is pure projection geometry. It
equals $\beta_t$ if and only if $v_t(\eta_t)-u_t\perp u_t$; it is
attenuated, zero, or sign-reversed according to whether
$0<\varphi_t(\eta_t)<\norm{u_t}^2$, $\varphi_t(\eta_t)=0$, or
$\varphi_t(\eta_t)<0$, and all three regimes occur already for $N=4$.
\item Under a composition-only change $\eta^{(0)}\to\eta^{(1)}$ with
$\beta_t\equiv\beta$ fixed, the plug-in path shifts by
\[
\Delta^{\mathrm{plug}}
=\beta\,\frac{\varphi_t(\eta^{(1)})-\varphi_t(\eta^{(0)})}{\norm{u_t}^2}.
\]
If $\varphi_t$ is nonconstant, equivalently if
$u_t^\trans M_tG_t(\eta)\ne0$ for some $\eta$, then
$\varphi_t$ is a nonconstant real-analytic function of $\eta$, and the set
of composition-only changes with $\Delta^{\mathrm{plug}}=0$ is
Lebesgue-null in $\R^{2q}$: for $\beta\ne0$, false attribution is
generic. Off that null set $\Delta^{\mathrm{plug}}\ne0$, so any
procedure whose power against fixed nonzero shifts of the projection
coefficient tends to one will report a strength change that did not
occur.
\item (Concurrent plug-in.) If the exposure is rebuilt each date from noisy
reports, $\widetilde W_t=W(m_t+\varepsilon_t)$ with
$\tilde u_t=M_t\widetilde W_ty_{t-1}$, the population projection coefficient
is
\[
\beta^{\mathrm{conc}}_t
=\beta_t\,
\frac{\inner{E\tilde u_t}{v_t(\eta_t)}}
{\norm{E\tilde u_t}^2+E\norm{\tilde u_t-E\tilde u_t}^2},
\]
the errors-in-variables form: report noise enters the denominator but not
the numerator, so $|\beta^{\mathrm{conc}}_t|$ is strictly smaller than the
corresponding noiseless-baseline coefficient
$|\beta_t\inner{E\tilde u_t}{v_t(\eta_t)}|/\norm{E\tilde u_t}^2$
whenever the noise is nondegenerate, and this attenuation does not vanish
with the number of dates.
\end{enumerate}
\end{theorem}

The proof is elementary and given in Section~\ref*{sec:si-plugin}; the
measure-zero statement uses only that zero sets of nonconstant real-analytic
functions are Lebesgue-null \citep{KrantzParks2002}. The value of the
theorem is exactness: parts (a)--(c) give
Figure~\ref{fig:false-attribution} a closed-form population counterpart. In a twelve-node
instance of the same design family, the population formula evaluates to
$\ptPop$ post-change against a Monte Carlo static plug-in mean of $\ptMC$
(agreement within one MCSE); a sign-reversing composition direction exists
with $\beta^{\mathrm{plug}}=\ptSignRev$ against a true $\beta=0.5$; and
the errors-in-variables formula puts the concurrent plug-in at $\ptConcTH$
against a Monte Carlo value of $\ptConcMC$ (the residual difference is the
finite-$N$ second-order term the formula's expectation ratio omits).
Figure~\ref{fig:dose-response} traces the theorem across a dose--response
grid of composition-change magnitudes. The theorem's value is conditional
on the baseline actually used, so the population curve is evaluated at each
replication's estimated baseline, constructed from the first eight report
waves exactly as the plug-in constructs it; the Monte Carlo shifts track
this matched population curve across the grid (largest absolute deviation
$\doseMaxGap$, with two-MCSE margins up to $\doseMaxTwoMCSE$;
$R=\doseRg$ per dose). The oracle-baseline curve is displayed for
reference: at the top dose it reaches $\doseOracleTop$ against
$\doseMatchedTop$ for the estimated baseline, an attenuation by the factor
$\doseAttFactor$ that is the errors-in-variables geometry of part (c)
acting on the baseline itself. The joint estimator of
Section~\ref{sec:estimation} stays at zero at every dose.

\begin{figure}[t]
\centering
\includegraphics[width=0.52\textwidth]{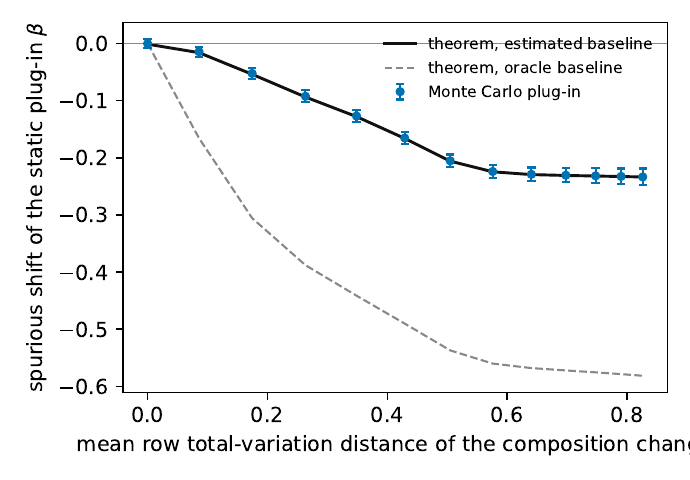}
\caption{Dose--response of the spurious strength shift under
composition-only changes of increasing size. Solid line:
Theorem~\ref{thm:plugin-main}(b)'s population formula evaluated at each
replication's estimated baseline, matched to the Monte Carlo procedure;
dashed line: the same formula at the oracle baseline $W(\eta^{(0)})$;
points: Monte Carlo ($\pm2$ MCSE, $R=\doseRg$ per dose). The joint
estimator is unaffected at every dose.}
\label{fig:dose-response}
\end{figure}

Two remarks connect the theorem to the rest of the paper. First, the
theorem delimits what the plug-in failure is not: it is not small-sample
noise, not report error (part (a) has none), and not a property of one
severe design; it is the projection geometry of
Theorem~\ref{thm:identification-main}(a) expressed in the estimand. Second,
no choice of baseline resolves the problem: any predetermined baseline
faces part (b), and any concurrent rebuild faces part (c). Separating
strength from composition requires the report channel to enter the
estimating equations jointly, which is what the remainder of the paper
constructs.

\section{Orthogonal path estimator}
\label{sec:estimation}

Training observations in fold complement $-k$ produce $\mathcal T_{tk}$-measurable covariance estimates $\widehat\Sigma_{tk}^{(-k)}$, $\widehat\Omega_{tk}^{(-k)}$ and a pilot $\theta^0_{t,-k}$. Estimated covariances are symmetrized and spectrally clipped to a fixed interval $[\underline v,\overline v]\subset(0,\infty)$ before principal inverse square roots are computed. Replacing the oracle covariances in \eqref{eq:residualizers}--\eqref{eq:jacobian} gives $\widehat R^Y_{tk}$, $\widehat R^z_{tk}$, $\widehat e_{tk}$, and $\widehat J_{tk}$. Put
\begin{equation}
\widehat S_{tk}(\theta)=\widehat J_{tk}(\theta)^\trans\widehat e_{tk}(\theta),
\qquad
\widehat{\mathcal I}_{tk}(\theta)
=\widehat J_{tk}(\theta)^\trans\widehat J_{tk}(\theta),
\label{eq:feasible-score-information}
\end{equation}
and
\begin{equation}
\widehat{\mathcal I}_t
=\sum_{k=1}^{K_f}\widehat{\mathcal I}_{tk}(\theta^0_{t,-k}).
\label{eq:feasible-information}
\end{equation}
For a prespecified $c_I>0$, define the everywhere-defined safe inverse
\begin{equation}
\widehat{\mathcal I}_t^{\ominus}
=
\begin{cases}
\widehat{\mathcal I}_t^{-1},
&\lambda_{\min}(\widehat{\mathcal I}_t)\ge c_In/2,\\
2(c_In)^{-1}I_d,&\text{otherwise}.
\end{cases}
\label{eq:safe-inverse}
\end{equation}
The fold-aggregated one-step estimator is
\begin{equation}
\widehat\theta_t
=\widehat{\mathcal I}_t^{\ominus}
\sum_{k=1}^{K_f}
\left\{
\widehat{\mathcal I}_{tk}(\theta^0_{t,-k})\theta^0_{t,-k}
+\widehat S_{tk}(\theta^0_{t,-k})
\right\}.
\label{eq:estimator}
\end{equation}
Residualization makes the realized score exactly invariant to the linear mean-nuisance values. Its conditional expectation is also orthogonal to admissible covariance perturbations. The complete statement and proof are in Section~\ref*{sec:si-estimation}.

\subsection{Uniform conditions}

All limits are uniform over a model class $\mathcal P_n$. Constants below do not depend on $P\in\mathcal P_n$, $t$, or $k$.

\begin{description}[style=nextline]
\item[(E1) Conditional experiment and folds.]
The target dimension, nuisance dimensions, and $K_f$ are fixed;
$c_fn\le n_{tk}\le C_fn$; $\ell_n=\log(2T_nK_f)=o(n)$.
The truth and designs are $\mathcal F_{t-1}$-measurable. Held-out innovations satisfy the conditional Gaussian model \eqref{eq:outcome-model}--\eqref{eq:report-model}, are independent of $\mathcal T_{tk}$, and are conditionally independent over folds. All training estimates and residualizers are $\mathcal T_{tk}$-measurable. Covariance nuisances are variation-independent of the target; target derivatives hold the training quantities and covariance nuisances fixed.

\item[(E2) Covariance and linear-design regularity.]
The spectra of $\Sigma_{tk}$ and $\Omega_{tk}$ lie in a fixed compact subset of $(0,\infty)$. The nonzero eigenvalues of $n_{tk}^{-1}X_{tk}^\trans\Sigma_{tk}^{-1}X_{tk}$ and $n_{tk}^{-1}U_{tk}^\trans\Omega_{tk}^{-1}U_{tk}$ are bounded away from zero and infinity; redundant nuisance columns are removed or handled with Moore--Penrose projections of locally constant rank.

\item[(E3) Local smoothness.]
For a fixed $r_0>0$, the $r_0$-ball about every $\theta_t$ lies in the parameter space. Uniformly on that ball,
\[
\sup_{0<\norm{\theta-\theta_t}\le r_0}
\frac{\norm{e_{tk}(\theta)-e_{tk}(\theta_t)}}{\norm{\theta-\theta_t}}
\le C\sqrt{n_{tk}},
\qquad
\sup_{\norm{\theta-\theta_t}\le r_0}
\norm{D^jJ_{tk}(\theta)}_{\op}
\le C\sqrt{n_{tk}},
\]
for $j=0,1,2$, with analogous bounds for feasible quantities on their covariance-good events.

\item[(E4) Joint information.]
For fixed $0<c_I<C_I<\infty$,
\[
c_InI_d\preceq\mathcal I_t\preceq C_InI_d
\qquad\text{simultaneously for }t\le T_n.
\]

\item[(E5) Feasible covariance and design rate.]
For a deterministic $\delta_n=o(1)$, with probability tending to one uniformly over $\mathcal P_n$,
\[
\max_{t,k}
\left\{
\norm{\widehat\Sigma_{tk}^{(-k)}-\Sigma_{tk}}_{\op}
+\norm{\widehat\Omega_{tk}^{(-k)}-\Omega_{tk}}_{\op}
\right\}\le\delta_n.
\]
Any estimated nonlinear design, reporting map, or sampling weight obeys the same Jacobian perturbation rate required in the Supplement.

\item[(E6) Honest pilot.]
For a deterministic $a_n=o(1)$, the pilots are $\mathcal T_{tk}$-measurable and
\[
\sup_{P\in\mathcal P_n}
P\left(\max_{t,k}\norm{\theta^0_{t,-k}-\theta_t}>a_n\right)\to0.
\]
\end{description}

Condition (E6) is substantive. Full-date information does not imply that a leave-fold training sample identifies the target. Section~\ref*{sec:si-estimation} gives both a counterexample and sufficient profiled-pilot conditions. For the linear charts of Section~\ref{sec:charts-main} the pilot is available in closed form: generalized least squares in the report channel identifies $\eta$ on each training fold, and one weighted regression then recovers $\beta$, so (E6) reduces to an explicit and checkable separation condition.

For deterministic $r_n>0$, $X_n=O_{\PP}^{\mathrm{unif}}(r_n)$ means that for every $\varepsilon>0$ there are finite $M$ and $n_0$ such that $\sup_{P\in\mathcal P_n}P\{\norm{X_n}>Mr_n\}<\varepsilon$ for every $n\ge n_0$.

\paragraph{Finite-sample studentization.}
On the event that each fold information matrix is nonsingular (probability
tending to one under (E4)), the estimator \eqref{eq:estimator} is the
information-weighted average of the per-fold one-steps $\widehat\theta^{(k)}_t=\theta^0_{t,-k}+\widehat{\mathcal
I}_{tk}^{-1}\widehat S_{tk}(\theta^0_{t,-k})$, so the folds' disagreement
measures the estimator's actual sampling noise, pilot included. The
implementation therefore (i) polishes each training pilot by two
Gauss--Newton steps on the training-fold joint criterion (training
measurable, so (E6) is unaffected), (ii) re-evaluates
$\widehat{\mathcal I}_t$ at $\widehat\theta_t$ for studentization, and
(iii) inflates the variance by
$\widehat\gamma^2=\max\{1,\ \mathrm{mean}_t\,D_t^2/(4\widehat
v_{\beta,t})\}$ with $D_t=\widehat\beta^{(1)}_t-\widehat\beta^{(2)}_t$; the
implementation fixes $K_f=2$, so the two fold one-steps are the complete
set and $4v_{\beta,t}$ is the null variance of $D_t$. By
the expansion of Theorem~\ref{thm:path-main}, $\widehat\gamma^2\to1$
uniformly, so every asymptotic statement is unchanged, while in finite
samples the factor absorbs the pilot-induced variance the leading term
omits. At small per-date information a fixed-design parametric bootstrap
critical value replaces the \v{S}id\'{a}k constant; both are evaluated in
Section~\ref{sec:sim-band}.

\section{Identification and inference}
\label{sec:theory}

\subsection{Sharp one-date condition}

Suppress date and fold indices. Let $L_y^\trans L_y=\Sigma^{-1}$ and $L_z^\trans L_z=\Omega^{-1}$, and define
\begin{equation}
\begin{aligned}
r&=M_{L_yX}L_yg(\eta),
&H&=M_{L_yX}L_yG(\eta),\\
U&=[AC,B],
&Q&=M_{L_zU}L_zA\dot m(\eta),
\qquad K_c=Q^\trans Q.
\end{aligned}
\label{eq:one-date-residuals}
\end{equation}
Here $C\kappa$ contains row levels and $B\alpha$ contains the declared report biases. All ranks are evaluated at an interior differentiable-in-quadratic-mean point, after removal of redundant composition coordinates.

\begin{theorem}[Outcome-only failure and joint identification]
\label{thm:identification-main}
Assume the conditional Gaussian model in Section~\ref{sec:model}, variation-independent mean and covariance nuisances, positive-definite covariances, a nonempty eligible support in every retained row, and a minimal differentiable composition chart with $C^\trans m(\eta)=C^\trans\dot m(\eta)=0$ and $\rank\{\dot m(\eta)\}=q$.

\begin{enumerate}[label=\textup{(\alph*)}]
\item After profiling outcome nuisances, the outcome-only information for $(\beta,\eta)$ is
\[
\mathcal I_Y=[\,r\ \ \beta H\,]^\trans[\,r\ \ \beta H\,].
\]
Thus outcomes alone identify the target locally if and only if $\rank[\,r\ \ \beta H\,]=q+1$. In particular, composition is not identified from outcomes when $\beta=0$ and $q>0$. In the unrestricted zero-diagonal row-stochastic model, a nonconstant lag row has an exact outcome-equivalent fiber of dimension $N-3$ when $N\ge4$.

\item After profiling $(\kappa,\alpha)$ and covariance nuisances, the continuous-report information for composition is $K_c=Q^\trans Q$. Reports identify composition locally if and only if
\[
\rank(Q)=q
\quad\Longleftrightarrow\quad
A\dot m(\eta)h\notin\col([AC,B])\ \text{for every }h\ne0.
\]
Correlated mirror-report errors are allowed through $\Omega$.

\item The joint efficient information is
\begin{equation}
\mathcal I_c=
\begin{pmatrix}
r^\trans r&\beta r^\trans H\\
\beta H^\trans r&\beta^2H^\trans H+K_c
\end{pmatrix}.
\label{eq:joint-information-main}
\end{equation}
The target $(\beta,\eta)$ is regularly locally identified if and only if
\begin{equation}
r\ne0,
\qquad
K_c+\beta^2H^\trans P_r^\perp H\succ0,
\qquad
P_r^\perp=I-\frac{rr^\trans}{r^\trans r}.
\label{eq:joint-identification-main}
\end{equation}
For $\beta\ne0$, this is equivalently
\[
\ker(Q)\cap\{h:Hh\in\Span(r)\}=\{0\}.
\]
For $\beta=0$, it reduces to $r\ne0$ and $K_c\succ0$.

\item Under paired exporter/importer means containing only row levels and receiver/supplier effects, the report law identifies the row-softmax network on the eligible support. An unrestricted common dyad effect invalidates this conclusion by absorbing every composition derivative.
\end{enumerate}
\end{theorem}

The proof in Section~\ref*{sec:si-identification} also gives the exact row fibers, scale ambiguity, component-wise gauge count, and weak-information geometry. The condition is first-order and local; it is not a claim of global identification for a noninjective chart.

\subsection{Uniform estimation of the path}

\begin{theorem}[Orthogonal expansion and uniform path rate]
\label{thm:path-main}
Suppose (E1)--(E6) hold. Then, uniformly over $P\in\mathcal P_n$,
\[
P\left\{
\min_{t\le T_n}\lambda_{\min}(\widehat{\mathcal I}_t)
\ge c_In/2
\right\}\to1,
\]
and
\begin{align}
\max_{t\le T_n}
\norm{\widehat\theta_t-\theta_t-\mathcal I_t^{-1}S_t}
&=
O_{\PP}^{\mathrm{unif}}\!\left[
a_n^2+(a_n+\delta_n)\sqrt{\frac{\ell_n}{n}}
\right],
\label{eq:path-expansion-main}\\
\max_{t\le T_n}\norm{\widehat\theta_t-\theta_t}
&=
O_{\PP}^{\mathrm{unif}}\!\left[
\sqrt{\frac{\ell_n}{n}}+a_n^2
\right],
\label{eq:path-rate-main}\\
\max_{t\le T_n}n^{-1}
\norm{\widehat{\mathcal I}_t-\mathcal I_t}_{\op}
&=O_{\PP}^{\mathrm{unif}}(a_n+\delta_n).
\label{eq:information-rate-main}
\end{align}
If $a_n^2=O\{\sqrt{\ell_n/n}\}$, the feasible estimator has the oracle maximum-over-date order $\sqrt{\ell_n/n}$. This is an upper-rate result conditional on an honest pilot, not a minimax claim.
\end{theorem}

\subsection{A simultaneous band for the coefficient path}

Let $e_1=(1,0,\ldots,0)^\trans$ and
\[
v_{\beta,t}=e_1^\trans\mathcal I_t^{-1}e_1,
\qquad
\widehat v_{\beta,t}=e_1^\trans\widehat{\mathcal I}_t^{\ominus}e_1.
\]

\begin{theorem}[Simultaneous Gaussian band]
\label{thm:bands-main}
In addition to Theorem~\ref{thm:path-main}, suppose there is a filtration $\mathcal F_{n,0}\subseteq\cdots\subseteq\mathcal F_{n,T_n}$ such that $S_t$ is $\mathcal F_{n,t}$-measurable, $\mathcal I_t$ is $\mathcal F_{n,t-1}$-measurable, and
\begin{equation}
\mathcal L(S_t\mid\mathcal F_{n,t-1})=N_d(0,\mathcal I_t)
\quad\text{almost surely}.
\label{eq:sequential-score-main}
\end{equation}
Assume $c_InI_d\preceq\mathcal I_t\preceq C_InI_d$ (the same constants as in (E4)) and
\[
b_n:=\sqrt{n\ell_n}\,a_n^2+(a_n+\delta_n)\ell_n\to0.
\]
For a fixed level $\alpha\in(0,1)$, let
\[
c_{1-\alpha,T_n}
=\Phi^{-1}\!\left(
\frac{1+(1-\alpha)^{1/T_n}}2
\right)
\]
and
\[
C_t=
\left[
\widehat\beta_t-c_{1-\alpha,T_n}\sqrt{\widehat v_{\beta,t}},
\widehat\beta_t+c_{1-\alpha,T_n}\sqrt{\widehat v_{\beta,t}}
\right].
\]
Then
\[
\sup_{P\in\mathcal P_n}
\left|
P\{\beta_t\in C_t\text{ for all }t\le T_n\}-(1-\alpha)
\right|\to0.
\]
\end{theorem}

Condition \eqref{eq:sequential-score-main}, not cross-fitting alone, yields independent standardized leading scores and the exact \v{S}id\'{a}k calibration. Section~\ref*{sec:si-bands} proves that estimation and studentization errors are negligible uniformly over the full path.

\subsection{Change attribution in an exact Gaussian benchmark}

The change theory uses a separate experiment. At each date, one observes $n_{\mathrm c}$ independent copies $X_{a,t}=x_t+\nu\xi_{a,t}$, $a\le n_{\mathrm c}$, with iid $N_d(0,I_d)$ noise. Strength and composition correspond to a fixed known orthogonal decomposition $\R^d=\mathcal V_{\mathrm{str}}\oplus\mathcal V_{\mathrm{cmp}}$.

\begin{theorem}[Gaussian change guarantees and obstructions]
\label{thm:change-main}
Consider a piecewise-constant path with $K\ge1$ changes, minimum spacing $\Delta$, minimum jump $\varsigma_{\min}$, and three independent copies. Let the screening window be $h$, $T\ge4h$, $M_h=T-4h+1$, and define
\[
\lambda_\alpha
=\nu\left\{\sqrt d+\sqrt{2\log(3M_h/\alpha)}\right\},
\quad
\varepsilon_{\alpha,K}
=\frac{\nu}{\sqrt h}
\left\{\sqrt d+\sqrt{2\log(6K/\alpha)}\right\}.
\]
If
\[
\Delta\ge5h,
\qquad
\sqrt{h/2}\,\varsigma_{\min}>3\lambda_\alpha,
\qquad
\varsigma_{\min}\ge8\varepsilon_{\alpha,K},
\]
then the three-copy procedure defined in Section~\ref*{sec:si-change} detects all $K$ changes, localizes each preliminary estimate within $h$, and refines each location to error less than
\[
r_\alpha=
\frac{800\nu^2}{9\varsigma_{\min}^2}
\log(12K/\alpha)
\]
simultaneously with probability at least $1-\alpha$. It correctly labels a strength or composition block whenever that projected jump is either zero or at least $6\varepsilon_{\alpha,K}$.

Conversely, in the $n_{\mathrm c}$-copy Gaussian experiment the following necessary conditions hold in the respective hard submodels: the detection condition assumes both type-I and worst-case type-II errors at most $\alpha_0<1/2$; the attribution condition assumes complete-vector error at most $\alpha_0$; and the localization condition assumes Hausdorff error probability at most $\alpha_0$.
\begin{align*}
\frac{n_{\mathrm c}\Delta\varsigma^2}{\nu^2}
&\ge\log\{1+4M(1-2\alpha_0)^2\}
&&\text{for detection among $M$ blocks},\\
\frac{n_{\mathrm c}h\varsigma^2}{\nu^2}
&\ge\mathfrak q_{J,\alpha_0}
&&\text{for simultaneous attribution of $J$ changes},\\
\frac{n_{\mathrm c}r\varsigma^2}{\nu^2}
&\ge\mathfrak q_{J,\alpha_0}
&&\text{for localization error below $r/2$},
\end{align*}
where
\[
\mathfrak q_{J,\alpha_0}
=2\left[\Phi^{-1}\{(1-\alpha_0)^{1/J}\}\right]^2
=4\log J+O(\log\log J).
\]
\end{theorem}

Theorem~\ref{thm:change-main} is exact for the Gaussian benchmark.
Section~\ref*{sec:si-transfer} gives the Le Cam transfer conditions (common
information, a shrinking local neighborhood, and a parameter-free
reconstruction kernel), and Proposition~\ref{prop:transfer-instance-main}
below verifies all of them, with explicit constants, for a fully specified
pinned-design replication experiment.

\subsection{Inference beyond Gaussian scores and strong identification}
\label{sec:robust-inference}

Two assumptions in Theorems~\ref{thm:path-main}--\ref{thm:bands-main} are
restrictive in applications: the exact sequential Gaussian score law
\eqref{eq:sequential-score-main}, and the information floor (E4) with its
honest pilot (E6). This subsection removes each in turn; proofs are in
Sections~\ref*{sec:si-robust} and \ref*{sec:si-weakid} of the Supplement.

For the band, the Gaussian law can be replaced by moment conditions plus a
computable leverage diagnostic. Write the standardized oracle score as
$Z_t=a_t^\trans\varepsilon_t$ with predictable unit-norm weights $a_t$, and
define the leverage $\mathrm{lev}_n=\max_{t\le T_n}\norm{a_t}_\infty$.

\begin{theorem}[The band without Gaussianity]
\label{thm:robust-band-main}
Replace the conditional Gaussian specification by: conditionally on
$\mathcal F_{n,t-1}$, the held-out innovation coordinates are independent,
centered, unit variance, uniformly sub-exponential, and independent of the
past. Assume (E1)--(E6), $\mathrm{lev}_n\{1+(\log T_n)^{3/2}\}\to0$, and
$\{\sqrt na_n^2+(a_n+\delta_n)\sqrt{\ell_n}\}\sqrt{\ell_n}\to0$. Then the
band of Theorem~\ref{thm:bands-main}, with the same \v{S}id\'{a}k critical
value, has simultaneous coverage converging to $1-\alpha$ uniformly over
$\mathcal P_n$.
\end{theorem}

The leverage $\mathrm{lev}_n$ certifies that no single held-out coordinate
dominates any date's score, so per-date central limit behavior holds far
enough into the tails to control a union over $T_n$ dates. It is computable
from the design and reported in every experiment below. The Supplement also
gives a finite-sample conservative alternative (a leverage-adjusted
Bernstein band) whose critical value uses the computed
$\norm{a_t}_\infty$ directly, and Section~\ref{sec:sim-band} reports both
the robustness and its empirical boundary.

For weak identification, the appropriate tool is score inversion: at a
hypothesized $\theta^0$ the residualized score is exactly centered and
exactly standardized whatever the information level, so validity requires
neither an estimator, nor a pilot, nor (E4). Let
$\mathrm{AR}_t(\theta^0)=S_t(\theta^0)^\trans\mathcal I_t(\theta^0)^\dagger
S_t(\theta^0)$.

\begin{theorem}[Identification-robust confidence sets]
\label{thm:weakid-main}
Under the conditional Gaussian model with oracle covariances and
(E1)--(E3): (a) $P\{\mathrm{AR}_t(\theta_t)>\chi^2_{d,1-\alpha}\}\le\alpha$
exactly, at every $n$ and every information level, with equality under full
rank; (b) under the sequential condition, the sets with per-date level
$(1-\alpha)^{1/T_n}$ are exactly simultaneously valid; (c) projections
deliver valid sets for $\beta_t$ alone; (d) with estimated covariances and a
ridge dominating the covariance-estimation error (in the rate sense made
precise in the Supplement), validity holds asymptotically uniformly,
without (E4) or (E6), and with exact asymptotic size under an information
floor.
\end{theorem}

The division of labor is prespecified and used in both empirical sections:
Wald bands where the information-floor diagnostic
$\lambda_{\min}(\widehat{\mathcal I}_t)/n$ clears its threshold, score sets
where it does not, the diagnostic itself reported always, and the realized
performance of this switching rule evaluated in
Section~\ref{sec:sim-weakid}.

\subsection{Misspecified charts and bounded reporter bias}
\label{sec:pseudo-main}

The chart and the reporter-bias design are declared commitments, so the
theory must state what happens when they fail by bounded amounts. Both
failures enter the report-channel mean and propagate through one fixed
linear map; Section~\ref*{sec:si-pseudo} proves the following and the
replication code verifies each display numerically.

\begin{proposition}[Pseudo-true targets, inference for them, sensitivity]
\label{prop:pseudo-main}
(a) Under off-chart composition $m^\dagger=\Psi\eta^\dagger+w$, the
population score has a root
$\theta^*=\theta^\dagger+\mathcal I_c^{-1}(0;\,Q^\trans R_zAw)+O(\norm w^2)$,
which for the report channel is exactly the information projection of the
truth onto the chart. With $\theta^*$ as the target, the score-inversion
statements of Theorem~\ref{thm:weakid-main} remain exact, and the
Wald-type path and band statements hold with an additional coverage error
of order $O(\norm w)$, vanishing as the misspecification shrinks and
measured directly in the replication suite: the declared-chart analysis
estimates its chart's best approximation, with uncertainty statements for
that object that are exact for score inversion and first-order accurate
for Wald inference, and a report-channel residual statistic with
computable noncentrality diagnoses the misspecification. (b) Under a
common dyad bias $\norm c_\infty\le\delta$, the target shifts by
$\Lambda c$ with a design-computable passthrough matrix $\Lambda$; each
coordinate moves by at most $\delta\norm{\Lambda_{l\cdot}}_1$, contrasts
$a^\trans\theta$ by at most $\delta\norm{a^\trans\Lambda}_1$ (doubled
across change windows), and the breakdown value $\delta^*$, the smallest
common bias that could overturn a stated conclusion, is a closed-form
number reported alongside the conclusion. As $\delta\to\infty$ the
sensitivity regions recover the impossibility result of
Theorem~\ref{thm:identification-main}(d).
\end{proposition}

\subsection{Change inference: benchmark limits and an observational test}
\label{sec:change-main-restructured}

The three-copy procedure and its finite-sample guarantee
(Theorem~\ref{thm:change-main}, upper part) are benchmark results: they
calibrate what is achievable, and Section~\ref*{sec:si-simchange} of the
Supplement measures how conservative their constants are. The minimax obstructions (lower part)
bind every procedure, present or future, and are the part of the change
theory that belongs with the identification analysis. For a fully specified
designed experiment, Proposition~\ref{prop:transfer-instance-main}
transfers both directions with explicit constants. None of this provides
inference on the path the motivating applications observe. The following
does, by inverting Theorem~\ref{thm:weakid-main}; the proof and the
computational protocol are in Section~\ref*{sec:si-obsdetect}.

\begin{theorem}[Score-inversion change inference, observational experiment]
\label{thm:obsdetect-main}
Under the conditions of Theorem~\ref{thm:weakid-main} with per-date level
$(1-\alpha)^{1/T_n}$: (a) the constancy test that rejects when
$\min_{\theta^0}\max_{t}\mathrm{AR}_t(\theta^0)$ exceeds
$\chi^2_{d,(1-\alpha)^{1/T_n}}$ has exact size at most $\alpha$ on the
observational path, with no pilot and no information floor; (b) under a
single change, the accepted split dates form a confidence set containing the
true change date, and the accepted parameter pairs induce a confidence
region for the jump, each with probability at least $1-\alpha$; a change is
certified \emph{inconsistent with composition-only} when no accepted split
admits a pair with $\Delta_\beta=0$, an error of probability at most
$\alpha$; (c) all statements hold asymptotically with estimated
covariances, without (E4) or (E6).
\end{theorem}

Theorem~\ref{thm:obsdetect-main} is exact for the idealized minimax test
but makes no optimality claim. Its implementation is search-based:
acceptance is certified by exhibiting a feasible value, while rejection
records exhaustion of a prespecified multistart search, a numerical
statement and not a global optimality certificate; search failure can
move the constancy test only toward rejection, whose error rate is
measured under the null, and attribution returns a three-state verdict
(consistent, inconsistent by search, undetermined) so search limits are
reported, never converted into certificates. Realized size, power, and
verdict rates are in Section~\ref{sec:sim-obsdetect}, bounded above by
the benchmark obstructions.

\begin{openproblem*}
Exhibit a procedure for the observational outcome--report experiment
(realized lags, estimated dependent covariances, one path) whose
detection, attribution, and localization thresholds match the benchmark
lower bounds of Theorem~\ref{thm:change-main} up to constants, or prove
that a gap is unavoidable.
\end{openproblem*}

Finally, the change benchmark is connected to a buildable experiment.

\begin{proposition}[A verified transfer instance]
\label{prop:transfer-instance-main}
In the pinned-design replication experiment (fixed designs across dates
and replicates, linear chart, known covariances, prespecified exposure
basket $\bar y$ chosen so that $H^\trans r=0$, and a local path of radius
$\varpi_n$ with $A_nT_n\bar n\varpi_n^4\to0$, where $A_n$ is the number of
independent replicate panels and $\bar n=N+2|\mathcal E|$ the per-replicate
dimension; the Supplement writes the radius as $b_n$), conditions
(C1)--(C3) hold with explicit
constants ($\nu_n=1$, $\mathcal I_0$ block diagonal and computable, and
quadratic constant $L_\star$ given in the Supplement), and an exact
parameter-free reconstruction kernel exists. Consequently the three-copy
guarantees and the minimax obstructions of Theorem~\ref{thm:change-main}
apply to that experiment with deficiency error at most
$\rho_n=(L_\star/2)(A_nT_n\bar n\varpi_n^4)^{1/2}\to0$, with jumps measured in
the $\mathcal I_0$ metric.
\end{proposition}

The instance is a designed replication, with repeated independent panels
measured against a fixed reference basket, and not the observational
autoregression; its role is that every hypothesis of the benchmark is
checkable before data collection for at least one fully specified
mechanism, with the orthogonality $H^\trans r=0$ a solvable design equation
verified to machine precision in the replication code.

\section{Observation boundaries}
\label{sec:boundaries}

The information conditions in Section~\ref{sec:theory} do not license arbitrary preprocessing of network reports. Four boundaries determine whether the target remains meaningful.

\paragraph{Censoring and structural zeros.}
With independent Gaussian reports, known positive scales, and known finite thresholds, censoring multiplies each mean-score contribution by a strictly positive information weight. It therefore preserves exact first-order rank, although severe censoring can make indispensable composition directions arbitrarily weak. With unknown scales or correlated censored reports, the full observed-data nuisance projection is required instead. A recorded zero that pools true zeros with positive flows below a detection limit is different: without a model for the below-limit positive distribution, the structural-zero probability and latent mean are only partially identified. The sharp sets are given in Section~\ref*{sec:si-boundaries}.

\paragraph{Report selection.}
If selection may depend without restriction on the unobserved report, two distinct latent-report laws can induce the same observed distribution. Identification then fails even in a Gaussian location family. Under a specified differentiable selection likelihood, vector composition is locally identified precisely when every nonzero linear combination of its observed-data score has a nonzero residual after projection onto the complete nuisance tangent space, that is, precisely when the profiled score covariance is positive definite. Selection probabilities cannot be estimated in a preliminary step and then ignored unless that operation is justified by the joint score calculation.

\paragraph{Node turnover and support change.}
Datewise inference is valid under an exogenous, known eligible support if the date-specific information condition is checked. Comparing row-normalized networks across changing supports requires an additional estimand. For any declared comparison or change-analysis horizon, the default in this paper is one support fixed in advance as the within-row intersection across all dates in that horizon; composition is renormalized on that fixed support, with entry and exit reported separately as an extensive-margin quantity. Off-support counterfactual weights are not parameters of the observed model.

\paragraph{Reporter-bias diagnostics.}
Paired discrepancies imply cycle restrictions under additive receiver/supplier bias. With known Gaussian covariance, the whitened residual cycle norm has an exact chi-square law. This is a specification test, not a validation theorem. With the error law held fixed, an unrestricted dyad-specific bias common to both reports cancels from every discrepancy and can absorb any admissible within-row composition perturbation. External validation, a third measurement channel, or a scientifically justified restriction is then required.

Repeated measurement is necessary but not sufficient: the data must contain variation that survives the declared nuisance design, censoring mechanism, selection law, and support definition. Table~\ref{tab:boundary} collects the complete boundary; Section~\ref{sec:worked-example} shows each failure as an exact computation, Section~\ref{sec:sim-obslayer} as numbers.

\begin{table}[t]
\centering
\small
\caption{The measurement boundary at a glance. ``Sensitivity'' = the
bounded-departure intervals of Proposition~\ref{prop:pseudo-main}.}
\label{tab:boundary}
\begin{tabular}{lll}
\toprule
measurement situation & status of $(\beta_t,\eta_t)$ & remedy / object\\
\midrule
outcomes only, one date & not identified (exact fiber) & repeated reports\\
paired reports, additive reporter bias & identified (globally) & --\\
\quad + unrestricted common dyad bias & not identified & restriction, validation, or sensitivity\\
\quad + common bias $\norm{c}_\infty\le\delta$ & set-identified & closed-form intervals, breakdown $\delta^*$\\
censoring, known thresholds and scales & identified; information decays & weights $\omega_m$; monitor $\sigma_{\min}$\\
recorded zeros pooling true zeros & partially identified & sharp sets for mass and mean\\
selection, ignorable or specified law & identified iff profiled score p.d. & check the observed-data score\\
selection, unrestricted nonignorable & not identified & design, not analysis\\
support change across a horizon & different estimand & common support + extensive margin\\
off-chart composition & pseudo-true target identified & information projection + diagnostic\\
\bottomrule
\end{tabular}
\end{table}

\section{Monte Carlo evidence}
\label{sec:simulation}

All experiments generate both outcomes and primitive mirror reports from
the joint experiment of Section~\ref{sec:model}; no simulated network path
is treated as observed. The base design is the gravity chart of
Section~\ref{sec:charts-main} scaled up: complete zero-diagonal support,
mirror double reports with pair correlation $\rho=0.5$, row-centered dyadic
covariates, outcome nuisances containing an intercept, the own lag, and a
persistent node covariate, $n_y$ within-date outcome replications and $n_z$
report waves, so per-date information grows in both channels as the
theory's $n$ requires. Every number below is produced by the replication code; each table and
experiment states its own replication counts and per-cell Monte Carlo
standard errors, which differ across experiments (up to
$\bootMCSE$ points in Table~\ref{tab:coverage}, up to
$\obsPowMCSEmax$ points in the observational change experiment, whose
conditional cells also state effective counts). The safe-inverse fallback \eqref{eq:safe-inverse} triggered in
$\simSafePct\%$ of all fitted dates.

\subsection{False attribution: the theorem in simulation}
\label{sec:sim-false-attr}

Figure~\ref{fig:false-attribution} illustrates
Theorem~\ref{thm:plugin-main}. Strength is fixed at
$\beta_t=\expOneBeta$ for all $T=\expOneT$ dates ($N=\expOneN$,
$n_y=\expOneny$) and only composition moves at $t=21$ (mean row
total-variation distance $\expOneTV$; Figure~\ref{fig:dose-response} covers
the full dose range, with the population formula evaluated at the same
estimated baseline the plug-in uses). The static plug-in averages
$\expOneStatPre$ before the change and $\expOneStatPost$ after, and the
conventional two-standard-error comparison detects the nonexistent
strength change in $\expOnePlugFalse\%$ of $R=\expOneR$ replications. The
concurrent plug-in attenuates to $\expOneConcPre$/$\expOneConcPost$, the
errors-in-variables denominator of part (c). Two estimators that use the
reports are unaffected: the report-only two-step tracks $\expOneROPre$
before and $\expOneROPost$ after, and the joint estimator
$\expOneJointPre$ and $\expOneJointPost$. For the joint estimator the
constancy verdict rejects only if no horizontal line fits inside the
simultaneous band; false detection then occurs in $\expOneJointFalse\%$ of
replications, and the calibrated band covers the constant path in
$\expOneJointCover\%$ (nominal $95\%$, MCSE $\expOneCoverMCSE$). Ignoring
the sampling error of $\widehat\eta$ reduces the report-only band to
$\expOneROCovNaive\%$ coverage against $\expOneROCovProp\%$ with
delta-method propagation, which the joint information matrix supplies
automatically; the controlled comparison of the two estimators across
report-information levels is in Section~\ref{sec:sim-jointgain}.

\begin{figure}[t]
\centering
\includegraphics[width=\textwidth]{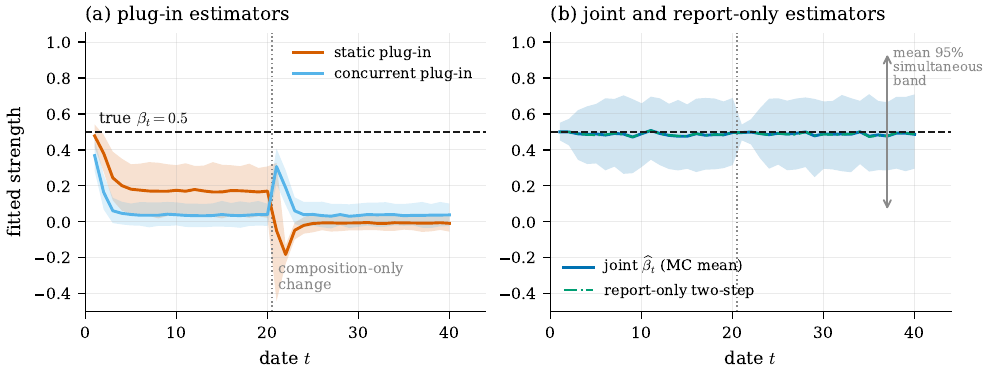}
\caption{Composition-only change at the dotted line; true strength constant
(dashed). (a) Plug-in paths (MC mean, 10--90\% bands over $\expOneR$
replications): the static plug-in shows a spurious strength decline
(Theorem~\ref{thm:plugin-main}b), the concurrent plug-in attenuates
(Theorem~\ref{thm:plugin-main}c). (b) The joint estimator and the
report-only two-step on the same scale (MC mean, 10--90\% band), with the
average calibrated $95\%$ simultaneous band width marked.}
\label{fig:false-attribution}
\end{figure}

\subsection{Calibration of the simultaneous band}
\label{sec:sim-band}

Table~\ref{tab:coverage} reports simultaneous coverage of the $\beta$ path
($T=25$, $N=18$) for the calibrated feasible estimator of
Section~\ref{sec:estimation}: polished training pilots, variance
re-evaluated at the estimate, and the fold-disagreement studentization
factor (average $\widehat\gamma=\calGamma$). The oracle-score rows isolate
the \v{S}id\'{a}k calibration: $\covOracleGauss\%$ under Gaussian
innovations and $\covOracleCexp\%$ under the centered-exponential family
covered by Theorem~\ref{thm:robust-band-main}; the $t_5$ column lies
outside that theorem's sub-exponential hypothesis and is reported as an
out-of-class check ($\covOracleTfive\%$ at this information level). The
calibrated feasible rows attain $\covCalEightGauss\%$ at the base design
and $\covCalTwoFourGauss\%$ when information triples in both channels
under Gaussian innovations; under centered-exponential innovations the
base design undercovers at $\covCalEightCexp\%$ and reaches
$\covCalTwoFourCexp\%$ at the tripled design, so the asymptotic claim, not
a finite-sample one, is what the non-Gaussian feasible rows support. The
mechanism decomposes into two parts: without the studentization factor,
the coverage gap decreases at the expansion's $n^{-1/2}$ rate
($\gapRawEight$ to $\gapRawTwoFour$ points, ratio
$\gapRawRatio\approx\sqrt3$); the studentization absorbs most of the
remainder, leaving $\gapEight$ points at the base design and slight
overcoverage at the tripled design. A fixed-design parametric bootstrap at
the base design gives $\bootCov\%$ ($B=\bootB$), no improvement over the
studentized band; this attributes the residual base-design gap to
cross-replication design variation and not to the critical value.
Consistently, Figure~\ref{fig:coverage-curve} shows the gap closing as
information grows, with the ladder reaching the nominal level at the
second rung ($\gridCovA$, $\gridCovB$, $\gridCovC$, $\gridCovD\%$). The
requirement that information grow in both channels is verified directly:
growing the outcome channel alone (the $n_y{=}24$, $n_z{=}1$ row) worsens
coverage to $\covUnmatchedGauss$--$\covUnmatchedCexp\%$, because the bands
narrow while report-channel pilot error is unchanged; this is the joint
information floor of condition (E4). The stress rows shrink per-date information to $n_y=1$, $N=10$, $T=60$:
Gaussian coverage is unchanged while skewed innovations undercover
($\covStressSidakReal\%$), locating the empirical boundary of the
leverage condition of Theorem~\ref{thm:robust-band-main}; the
leverage-adjusted Bernstein band, computed with the per-date design
leverage (mean $\norm{a_t}_\infty=\stressKappaMean$, maximum
$\stressKappaMax$), restores validity conservatively at
$\covStressBernReal\%$, at a width cost of the factor $\bernWidthRatio$
on average ($\bernWidthRatioMax$ at the maximum-leverage date).

\begin{table}[t]
\centering
\caption{Simultaneous coverage of the $\beta$ path (nominal $95\%$;
$R=\calR$ for calibrated rows and $R=\bootR$ for the bootstrap row;
per-cell MCSEs at most $\covCalMCSEmax$ points in the calibrated rows and
at most $\bootMCSE$ points elsewhere; width in parentheses where
informative). The $t_5$ column is outside
Theorem~\ref{thm:robust-band-main}'s hypotheses (out-of-class check).}
\label{tab:coverage}
\small
\begin{tabular}{lccc}
\toprule
& Gaussian & $t_5$ & centered exp.\\
\midrule
oracle score, $n_y=8$ & $\covOracleGauss$ & $\covOracleTfive$ & $\covOracleCexp$\\
calibrated feasible, $n_y=8$, $n_z=1$ & $\covCalEightGauss$ ($\widthCalEightGauss$) & $\covCalEightTfive$ & $\covCalEightCexp$\\
calibrated feasible, $n_y=24$, $n_z=3$ & $\covCalTwoFourGauss$ ($\widthCalTwoFourGauss$) & $\covCalTwoFourTfive$ & $\covCalTwoFourCexp$\\
\quad outcome channel only, $n_y=24$, $n_z=1$ & $\covUnmatchedGauss$ & $\covUnmatchedTfive$ & $\covUnmatchedCexp$\\
+ parametric bootstrap, $n_y=8$ & $\bootCov$ ($\bootWidth$) & -- & --\\
\midrule
stress, oracle score ($n_y{=}1$, $N{=}10$, $T{=}60$) & $\covStressGauss$ & -- & $\covStressSidakReal$\\
\quad + Bernstein band, computed leverage & -- & -- & $\covStressBernReal$\\
\bottomrule
\end{tabular}
\end{table}

\begin{figure}[t]
\centering
\includegraphics[width=0.5\textwidth]{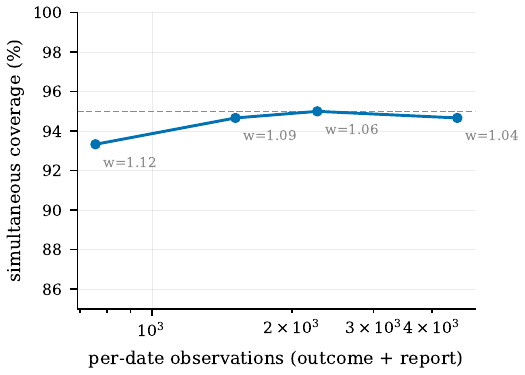}
\caption{Coverage--width ladder for the calibrated feasible band as
per-date information grows in both channels (Gaussian innovations,
$R=\gridR$ per point; $w$ = mean band width).}
\label{fig:coverage-curve}
\end{figure}

\subsection{Weak identification: validity and cost}
\label{sec:sim-weakid}

Table~\ref{tab:weakid} evaluates Theorem~\ref{thm:weakid-main} under both
covariance regimes. The computation of every profiled quantity was
repaired for this revision: folds and training-covariance inputs are drawn
once per replication and frozen, so the profiled objective is a
deterministic function of the hypothesized value; the inner minimization
is Gauss--Newton with backtracking and multistart, run to a verified
tolerance; and grid scans warm-start from neighboring solutions. Across
all configurations the profile failed to converge in at most
$\widConvFailMax\%$ of evaluations. (The previous implementation redrew
folds inside the objective, recorded convergence failure in every
replication, and its projected-set columns were accordingly withdrawn and
recomputed; the single-evaluation score columns were unaffected.)
Validity: the score statistic with oracle covariances holds its level at
every information level ($\widWkBAROr$--$\widStrAAROr\%$ across all seven
configurations), as Theorem~\ref{thm:weakid-main}(a) requires, while Wald
coverage degrades where the information diagnostic is lowest, to
$\widWkAWaldB\%$ for $\beta$ at the weak $\beta{=}0.4$ design. (With the
repaired estimator the Wald degradation at these designs is moderate, not
the collapse reported by the earlier implementation, whose weak-design
fits predated the pilot polish; the repaired numbers supersede them.)
Feasibility has a measurable cost: with covariances estimated on training
folds at the hypothesized value, the joint score rows run
$\widWkCAR$--$\widWkAAR\%$ at these fold sizes, the covariance-estimation
error that the ridge of Theorem~\ref{thm:weakid-main}(d) absorbs at a cost
in power. Projected sets are conservative everywhere
($\widWkCProj$--$\widStrAProj\%$). The switching rule (Wald when
$\lambda_{\min}(\widehat{\mathcal I}_t)/n$ clears $\widFloor$, projected
score otherwise) behaves as designed below the floor, where it uses the
projected set and covers ($\widVwBSwitch\%$ at very weak reports), and
inherits Wald's dip when the diagnostic clears the floor at a design where
Wald nevertheless undercovers ($\widWkASwitch\%$ at weak $\beta{=}0.4$,
diagnostic $\widWkALam$ against floor $\widFloor$): the floor level is a
tuning choice, and the table makes its consequence visible. The costs of
the protection are length and power: median
projected length grows from $\widStrALen$ (strong reports) to $\widWkALen$
(weak), $\widVwCUnb\%$ of very-weak-report scans reach the scan boundary
and yield near-vacuous sets, and power against a $2$-unit local
alternative is $\widStrAPowTwo\%$ even under strong identification.

\begin{table}[t]
\centering
\caption{One-date coverage under weakening report information
($R=\widR$; nominal $95\%$; MCSE at most $2.1$ points per coverage cell
and $2.7$ per power cell). Score
columns: oracle covariances (exact theory) vs covariances estimated at the
hypothesized value (feasible). ``switch'' = prespecified Wald/score rule
(feasible); ``len'' = median projected length; ``pow$_2$'' =
projected-test rejection of a $2$-unit local alternative. All profiled
columns use the frozen deterministic objective with verified convergence.}
\label{tab:weakid}
\small
\begin{tabular}{lcccccccc}
\toprule
config & $\lambda_{\min}/n$ & Wald $\beta$ & \multicolumn{2}{c}{score joint} & score proj. & switch & len & pow$_2$\\
& & & oracle & feasible & & & &\\
\midrule
strong, $\beta{=}0.4$ & $\widStrALam$ & $\widStrAWaldB$ & $\widStrAAROr$ & $\widStrAAR$ & $\widStrAProj$ & $\widStrASwitch$ & $\widStrALen$ & $\widStrAPowTwo$\\
strong, $\beta{=}0.1$ & $\widStrBLam$ & $\widStrBWaldB$ & $\widStrBAROr$ & $\widStrBAR$ & $\widStrBProj$ & $\widStrBSwitch$ & -- & --\\
weak, $\beta{=}0.4$ & $\widWkALam$ & $\widWkAWaldB$ & $\widWkAAROr$ & $\widWkAAR$ & $\widWkAProj$ & $\widWkASwitch$ & $\widWkALen$ & $\widWkAPowTwo$\\
weak, $\beta{=}0.1$ & $\widWkBLam$ & $\widWkBWaldB$ & $\widWkBAROr$ & $\widWkBAR$ & $\widWkBProj$ & $\widWkBSwitch$ & -- & --\\
weak, $\beta{=}0$ & $\widWkCLam$ & $\widWkCWaldB$ & $\widWkCAROr$ & $\widWkCAR$ & $\widWkCProj$ & $\widWkCSwitch$ & $\widWkCLen$ & $\widWkCPowTwo$\\
very weak, $\beta{=}0.1$ & $\widVwBLam$ & $\widVwBWaldB$ & $\widVwBAROr$ & $\widVwBAR$ & $\widVwBProj$ & $\widVwBSwitch$ & -- & --\\
very weak, $\beta{=}0$ & $\widVwCLam$ & $\widVwCWaldB$ & $\widVwCAROr$ & $\widVwCAR$ & $\widVwCProj$ & $\widVwCSwitch$ & $\widVwCLen$ & $\widVwCPowTwo$\\
\bottomrule
\end{tabular}
\end{table}

\subsection{Joint versus two-step: where the outcome channel matters}
\label{sec:sim-jointgain}

The report-only two-step with propagated uncertainty is the strongest
simple alternative to the joint estimator, and at the base design it is
nearly as accurate, so the comparison must be made where the theory says
the estimators separate: Theorem~\ref{thm:identification-main}(c) gives
the outcome channel's composition information as
$\beta^2H^\trans P_r^\perp H$, which is negligible when report information
is strong and becomes the binding source when it weakens.
Table~\ref{tab:jointgain} varies the report noise scale at the base design
with $\beta=0.5$ and reports composition RMSE, band coverage, and band
width for both estimators. At the strong-report end the two estimators are
equivalent in practice (composition RMSE $\jgAEtaJ$ against $\jgAEtaR$)
and the simpler two-step is adequate. As report noise grows the joint
estimator's advantage appears where predicted and grows monotonically: the
composition RMSE reduction is $\jgEtaGainA\%$ at the strongest reports and
$\jgEtaGainD\%$ at the weakest ($\jgDEtaJ$ against $\jgDEtaR$), and the
joint $\beta$ band stays narrow ($\jgAWidJ$ to $\jgDWidJ$) while the
two-step band widens ($\jgAWidR$ to $\jgDWidR$). The coverage columns
document a second, structural point: the two weakest rows are the weak and
very-weak regimes of Table~\ref{tab:weakid} (the same noise scales), where
Wald bands fail for \emph{both} estimators; the prespecified protocol
switches to score-inversion sets there, and the score construction is
built on the joint model's score, which the two-step does not supply. The
recommendation is design-based and prespecified: compute the
report-channel information $\sigma_{\min}(Q)$ before outcomes are
examined; when reports are strong the two-step with propagated uncertainty
is a legitimate simplification, and when they are weak the joint
estimator both improves composition recovery and is the only one of the
two with valid (score-inversion) inference.

\begin{table}[t]
\centering
\caption{Joint estimator versus report-only two-step (propagated
uncertainty) as report information weakens ($\beta=0.5$, base design,
$R=\jgR$ per row; nominal $95\%$ simultaneous bands). RMSE is for the
composition path; width is the median simultaneous band width for
$\beta$.}
\label{tab:jointgain}
\small
\begin{tabular}{lcccccc}
\toprule
report noise & \multicolumn{2}{c}{$\eta$ RMSE} & \multicolumn{2}{c}{$\beta$ band cov.} & \multicolumn{2}{c}{$\beta$ band width}\\
$\sigma_{\mathrm{rep}}$ & joint & two-step & joint & two-step & joint & two-step\\
\midrule
$0.8$ & $\jgAEtaJ$ & $\jgAEtaR$ & $\jgACovJ$ & $\jgACovR$ & $\jgAWidJ$ & $\jgAWidR$\\
$1.6$ & $\jgBEtaJ$ & $\jgBEtaR$ & $\jgBCovJ$ & $\jgBCovR$ & $\jgBWidJ$ & $\jgBWidR$\\
$3.2$ & $\jgCEtaJ$ & $\jgCEtaR$ & $\jgCCovJ$ & $\jgCCovR$ & $\jgCWidJ$ & $\jgCWidR$\\
$6.4$ & $\jgDEtaJ$ & $\jgDEtaR$ & $\jgDCovJ$ & $\jgDCovR$ & $\jgDWidJ$ & $\jgDWidR$\\
\bottomrule
\end{tabular}
\end{table}

\subsection{Observational change inference}
\label{sec:sim-obsdetect}

Theorem~\ref{thm:obsdetect-main} is evaluated as a deployable procedure on
the observational experiment ($N=18$, $n_y=24$, $n_z=3$, $T=25$),
repaired for this revision in three ways: the split scan covers every
eligible split and the true change date enters only the scoring of
verdicts; attribution is evaluated over the entire accepted-split set with
a multistart constrained search; and verdicts take three states
(consistent with composition-only, inconsistent by search, undetermined).
Monotonicity of the minimax value in the date set implies the
accepted-split set is an interval, and the implementation certifies its
two endpoints (Section~\ref*{sec:si-obsdetect}).

Size. With oracle covariances the realized size of the constancy test is
$\obsSize\%$ ($R=\obsSizeR$, MCSE $\obsSizeMCSE$); rejection is a
certified-search statement, so this measures statistical plus search
error under the null, and it is statistically consistent with the nominal
$5\%$. The plug-in feasible variant (covariances estimated once on
training folds at a pooled pilot value) over-rejects at $\obsSizeFeas\%$
($R=\obsSizeFeasR$): estimated covariances without the ridge of
Theorem~\ref{thm:weakid-main}(d) cost real size, consistent with the
feasible score rows of Table~\ref{tab:weakid}; the ridge is the
prescribed remedy.

Power and localization. Detection rates are $\obsPowA\%$ at
$\Delta\beta=0.25$, $\obsPowB\%$ at $\Delta\beta=0.4$, and $\obsPowCmp\%$
for the composition reallocation ($R=\obsPowR$ each; per-cell MCSE at most
$\obsPowMCSEmax$ points). For the composition change ($\obsDetCmp$
detections) the accepted interval contains the true date in
$\obsSplitCov\%$ of detections with mean cardinality $\obsSplitW$ dates.
Small strength jumps are detected but poorly localized: the accepted
interval averages $\obsSplitWA$ dates at $\Delta\beta=0.25$ (covering the
true date in $\obsSplitCovA\%$ of $\obsDetA$ detections) and
$\obsSplitWB$ dates at $\Delta\beta=0.4$ ($\obsSplitCovB\%$ of
$\obsDetB$). Conditional-on-detection coverage may fall below $95\%$
without contradicting the theorem, whose guarantee is unconditional;
unconditionally the coverage events hold in at least $94\%$ of
replications in every case here.

Attribution. The composition-only change is correctly certified
consistent with composition-only in $\obsAttCmp\%$ of its detections
(MCSE $\obsAttCmpMCSE$ points on $\obsDetCmp$ detections); the false
certificate ``inconsistent with composition-only'' occurs in
$\obsFalseInc\%$, well within the theorem's $5\%$ bound, and
$\obsUndetCmp\%$ return undetermined. The multistart search over the
full accepted set repaired the single-start asymmetry of the earlier
implementation. Certification power against strength jumps is low at this
design ($\obsAttStrA\%$ at $\Delta\beta=0.25$, $\obsAttStrB\%$ at
$\Delta\beta=0.4$): the union quantifier that makes the composition
certificate valid also lets a small strength jump hide at an off-center
accepted split where a common $\beta$ fits both segments. This is the
single-path attribution difficulty the benchmark obstructions quantify,
now visible in a deployable procedure; the benchmark evaluation in
Section~\ref*{sec:si-simchange} bounds the possible improvement.

\subsection{The observation layer in numbers}
\label{sec:sim-obslayer}

Three short experiments give the Supplement's observation-model results
numerical content. Censoring: the information weight $\omega(a)$ of
Section~\ref*{sec:si-boundaries} matches the Monte Carlo score variance
to $\censOmegaErr$ uniformly over thresholds, and in the worked four-node
design $\sigma_{\min}(Q_{\mathrm{cen}})$ decays from $\censSminOpen$ to
$\censSminTwo$ at a two-standard-deviation threshold: rank preserved,
information declining, as the theory states. Pooled zeros: with a
detection limit pooling true zeros and small flows, the realized sharp
identified set for the latent mean is $[\zeroLo,\ \zeroHi]$ around a
truth of $\zeroTruth$, and the structural-zero mass is identified only up
to $[0,\ \zeroPzero]$. Nonignorable selection: when reporting probability
depends on the latent report, composition estimates acquire bias growing
from $\mnarBiasZeroA$/$\mnarBiasZeroB$ (none) to
$\mnarBiasTwoA$/$\mnarBiasTwoB$ at strong selection, the finite-sample
counterpart of the nonidentification result and the reason the selection
model is a declared commitment.

\section{An application to mirror-reported trade}
\label{sec:application}

This section applies the complete prespecified protocol to a real
mirror-reported panel: annual bilateral goods trade among $\rlN$ economies
over \rlYearA--\rlYearB\ ($T=\rlT$), with the exporter-reported
free-on-board value and the importer-reported cost-insurance-freight value
of the same flow supplying the paper's mirror double report. Flows are the
IMF Direction of Trade Statistics pair as distributed in the CEPII Gravity
database \citep{ConteCotterlazMayer2022}; outcomes are World Bank real GDP
log growth \citep{WorldBankWDI}; the static node covariate is log initial
GDP. The chart is the gravity chart of Section~\ref{sec:charts-main} with
$q=2$: the negative log great-circle distance between capitals and a
same-EU bloc indicator fixed at membership as of 2000, both row-centered.
A dyad-year enters only when both mirror reports are present and positive
($\rlAvail\%$ of dyad-years). The outcome channel has one observation per
economy-year ($n_y=1$), the thinnest configuration the theory covers; the
consequences are reported below exactly as the estimation theory predicts
them. All design choices were fixed before outcomes were examined; the
deviation log and the loader validation are in the replication package.

Diagnostics come first, and the first one delivers a substantive verdict
on the measurement design. Mirror discrepancies average $\rlDisc$ log
points (dispersion $\rlDiscSD$), the magnitude the measurement model
assumes for the CIF/FOB wedge plus reporter noise. The reporter-cycle
specification test of Theorem~\ref{thm:identification-main}(d), run
yearly with the extract's pilot covariance, rejects the purely additive
receiver/supplier design in every year ($\rlCyclePass\%$ pass rate): real
CIF/FOB discrepancies contain a dyad-specific component, as expected when
freight and insurance costs depend on the route, so a full deployment
should add wedge columns (distance interacted with the importer report,
at minimum) to the declared bias design, and the deviation log records
this. Two consequences must be separated. A route-specific wedge that is
approximately constant over time shifts the level of the composition
coordinates through its projection on the chart but cancels from
across-time contrasts, so the decline results below are robust to it. A
time-varying component common to both reports is exactly the direction
that no test can detect (Theorem~\ref{thm:identification-main}(d)) and
that the sensitivity analysis below prices. The information-floor
diagnostic
then partitions the sample exactly as Theorem~\ref{thm:weakid-main}
anticipates: at $\rlNSafe$ of $\rlT$ dates the scaled minimum eigenvalue
falls to between $\rlLamSafeMin$ and $\rlLamSafeMax$, against a median of
$\rlLamGoodMed$ (minimum $\rlLamGoodMin$) at the remaining dates, and the
flagged dates are precisely the crisis years 1998, 2001--2002, 2009--2010,
2015, and 2019--2020, where the single-observation outcome cross-section
is most volatile. At those dates the joint Wald fit is not licensed and
the safe-inverse fallback engages; composition is then read from the
closed-form report-channel generalized least squares fit, which is exact
for the linear chart, uses the $2|\mathcal E_t|$ report observations
alone, and is stable at every date. Figure~\ref{fig:real-application}
shades the flagged years. This is the paper's prespecified switching
discipline operating on data no one tuned.

\begin{figure}[t]
\centering
\includegraphics[width=\textwidth]{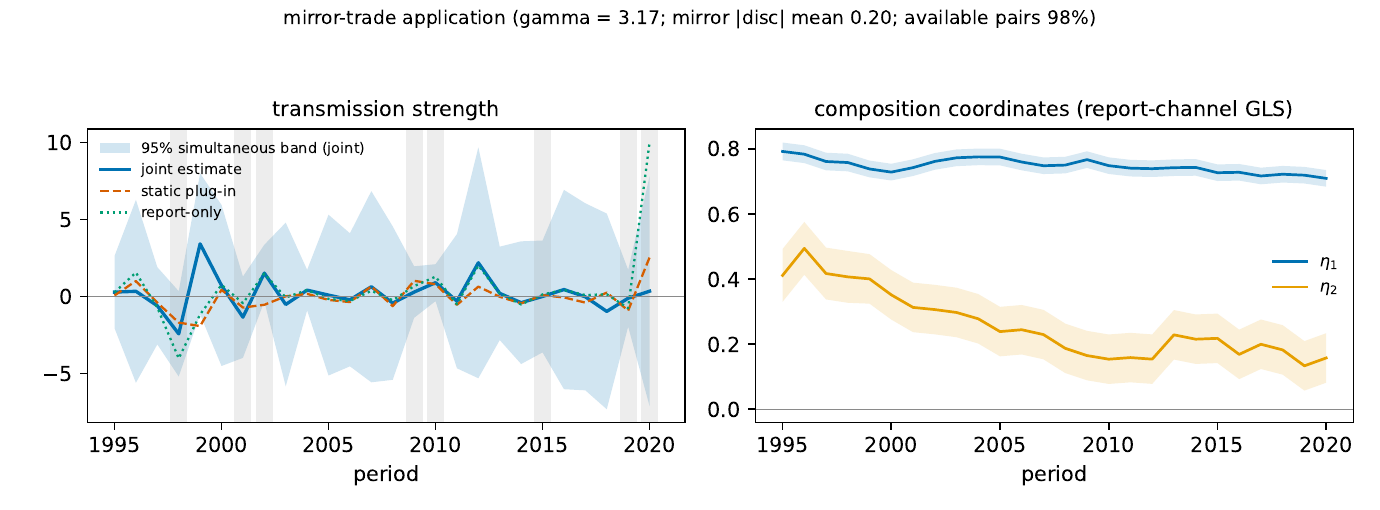}
\caption{Mirror-trade application, $\rlN$ economies, \rlYearA--\rlYearB.
Left: joint strength estimate with $95\%$ simultaneous band, static
plug-in, and report-only two-step; shaded years fail the information
floor. Right: composition coordinates from the exact report-channel GLS
with \v{S}id\'{a}k simultaneous bands built from the GLS standard errors
($\eta_1$ distance decay, $\eta_2$ EU bloc).}
\label{fig:real-application}
\end{figure}

The composition path is the substantive finding. Distance decay is high
and stable, drifting from $\rlEtaOneFirst$ in \rlYearA\ to
$\rlEtaOneLast$ in \rlYearB. The EU-bloc coordinate declines from a peak
of $\rlEtaTwoPeak$ in \rlEtaTwoPeakYear\ to $\rlEtaTwoEnd$ by \rlYearB:
the within-row trade-share premium attached to common EU membership fell
by roughly two thirds over the sample, a gradual reallocation consistent
with the rise of extra-bloc trade over this period. The prespecified
contrast (mean of the first five years minus mean of the last five) is
$\rlDeclineEtaTwo$ with standard error $\rlDeclineEtaTwoSE$
($\rlDeclineEtaTwoZ$ standard errors), against $\rlDeclineEtaOne$ (SE
$\rlDeclineEtaOneSE$) for distance decay. The sensitivity discipline of
Proposition~\ref{prop:pseudo-main} attaches the breakdown value: with
passthrough $\ell_1$ rows $(\rlSensBeta,\ \rlSensEtaOne,\ \rlSensEtaTwo)$
computed from the design, a common dyad bias would need to reach
$\delta^*=\rlDeltaStarEtaTwo$ on the log-report scale, sustained with
opposite signs across the two windows, before the bloc-decline conclusion
could be an artifact of correlated misreporting; the bound concerns the
component common to both reports, which mirror discrepancies cannot
reveal and only substantive knowledge of reporting practice can exclude.

The strength path illustrates the paper's honesty requirements; it does
not carry a substantive conclusion. With $n_y=1$ the calibrated simultaneous
band has median half-width $\rlBandHalfMed$ growth-percentage units per
unit of exposure: the constancy verdict correctly does not reject, but
the band contains every economically relevant value of $\beta_t$, so the
annual outcome channel cannot bound transmission strength usefully, and
the paper reports exactly that. The fold-disagreement studentization is
$\widehat\gamma=\rlGamma$, the finite-sample honesty factor at work. Two comparators make the
joint estimator's stability visible at the hardest date: in 2020 the
joint estimate is $\rlBetaTwenty$ while the report-only two-step returns
$\rlROTwenty$ and the static plug-in $\rlPlugTwenty$. A sharper strength
analysis requires more outcome information per date; quarterly national
accounts (four within-year replications) are the natural extension and
are listed in the deployment requirements below.

\subsection{Deployment requirements and deviations}
\label{sec:protocol}

The prespecified requirements of a deployment are: the eligible support
rule and common-support horizon (Section~\ref{sec:boundaries}); the chart
covariates and their provenance; the reporter-bias columns; the selection
model for missing returns; censoring thresholds; the information floor
licensing Wald against score inference; the common-bias sensitivity bound
reported with every attribution; and a held-out forecast block. This
application fixes all of them and records two deviations from the ideal:
the extract is the CEPII redistribution of the DOTS mirror pair, which
ends in 2020 (the IMF's own portal was mid-migration at the time of the
run, as documented in the replication log), and the outcome channel is
annual with $n_y=1$, which the strength analysis above reports as its
binding limitation. The loader is validated end to end on the documented
schema (reconstruction to $\loaderZerr$, estimates matching the direct
pipeline to $\loaderTherr$), so extending to quarterly outcomes or a
Comtrade extract requires only the data files. A synthetic end-to-end
validation of the full pipeline under known truth, including the
attribution and forecasting steps that this application's information
level does not support, is in Section~\ref*{sec:si-vignette} of the
Supplement.

\section{Conclusion}
\label{sec:conclusion}

This paper asked when transmission strength can be separated from network
composition when the network is latent, time varying, and measured with
error. The answer is an exact local criterion: strength is identified when
its exposure survives the outcome nuisances, and composition is identified
when every composition direction confounded with strength in the outcome
channel is recovered by the report channel. Both requirements are
summarized by positive definiteness of a profiled joint information
matrix that is computable date by date from training data, with its
report block available before outcomes are examined.

Around this criterion the paper provides an inferential chain whose exact
and asymptotic parts are separated explicitly. Closed-form pseudo-true
coefficients establish that plug-in practice generically converts
composition changes into apparent strength changes, and the population
formula is verified against Monte Carlo at the matched baseline across a
full dose grid. A cross-fitted orthogonal estimator recovers both paths;
the strength-path simultaneous band is exact for oracle Gaussian scores,
asymptotically valid with estimated covariances, valid under
sub-exponential scores subject to a computable leverage condition, and
measured for finite-sample calibration along an information ladder.
Score-inversion sets are exact under weak identification and replace the
bands under a prespecified switching rule. Pseudo-true targets,
sensitivity intervals, and breakdown values quantify bounded failures of
the chart and the bias design, with score-inversion inference exact for
the pseudo-true target and Wald inference accurate to first order in the
misspecification. Change inference on the observed path is exactly sized
for the idealized test; its search-based implementation is prespecified,
its verdicts include an undetermined state, and its operating
characteristics are measured against benchmark guarantees and minimax
obstructions. The complete protocol runs on a real mirror-reported trade
panel of eighteen economies over 1995 to 2020: the information
diagnostics select exactly the crisis years, and the report channel
identifies a stable distance-decay coefficient alongside a two-thirds
decline in the EU-bloc composition coordinate.

The limitations are structural and are stated where they bind. The
composition chart, the reporter-bias design, the selection model, and the
support rule are declared commitments: unrestricted common dyad bias and
unrestricted nonignorable selection destroy composition identification
exactly, pooled zeros are only partially identified, and off-support
counterfactuals are not parameters of the observed model. The estimand is
predictive dependence, not a causal effect. The application's outcome
channel is annual with one observation per economy-year, so its strength
path is reported as uninformative, exactly as the theory predicts at that
information level; the composition findings and the diagnostics carry the
section. Wald inference for pseudo-true targets
carries a first-order error in the misspecification size, measured in the
verification suite. The observational change test carries no optimality
guarantee, and rejection by its implementation is a certified-search
statement whose error rate is measured, not a global optimality
certificate.

Three directions follow. First, sharpening the application: quarterly
outcome data (four within-year replications) to make the strength path
informative, and a Comtrade extract extending the mirror panel beyond
2020; both are drop-in under the validated loader and the deviation log
of Section~\ref{sec:protocol}. Second, procedures for the observational
experiment whose detection and attribution thresholds match the benchmark
lower bounds, or a proof that a gap is unavoidable; this is the main open
theoretical problem the paper defines. Third, data-driven chart selection
with valid post-selection inference, and sensitivity regions for selection
departures analogous to the common-bias analysis.

\bibliographystyle{plainnat}
{\footnotesize
\bibliography{references}}

\clearpage
\setcounter{section}{0}
\renewcommand{\thesection}{S\arabic{section}}
\setcounter{equation}{0}
\numberwithin{equation}{section}
\setcounter{table}{0}
\renewcommand{\thetable}{S\arabic{table}}
\setcounter{figure}{0}
\renewcommand{\thefigure}{S\arabic{figure}}
\begin{center}
{\LARGE Supplementary Material}\\[10pt]
{\normalsize Sections S1--S16 contain the complete observed-data
experiment, all proofs, and additional results. Section, theorem,
equation, table, and figure numbers in this part carry the prefix S.}
\end{center}
\medskip

\section{Notation and the complete observed-data experiment}
\label{sec:si-observation}
\subsection{Prespecified observation map and likelihood}

Fix date $t$ and condition on $\mathcal H_t$, which contains the past, prespecified covariates, master node registry, active and lag-eligible node sets, eligible dyads $\mathcal E_t$, complete report opportunities $\mathcal J_t$, the report-to-dyad map, and every known measurement protocol and bound. These objects are parameter free conditional on $\mathcal H_t$, or their law is included in the full likelihood.

For $e\in\mathcal E_t$, an optional label $S_{e,t}\in\{0,1\}$ records structural positivity and has differentiable mass $p_{\xi_t}(s\mid\mathcal H_t)$. The main text conditions on a prespecified eligible support and targets the positive-flow composition coordinate $\eta_t$; $\xi_t$ is therefore profiled as a nuisance whenever this optional hurdle layer is present. An analyst who instead normalizes $\rho_{e,t}(\xi_t)\exp\{\ell^+_{e,t}(\kappa_t,\eta_t)\}$ and targets the resulting network must enlarge the target to $(\eta_t,\xi_t)$ and rederive both the outcome derivative and the joint information. No such enlarged-target claim is made here.

Let $M_t=|\mathcal J_t|$ and write
\begin{equation}
Z_t^\star\mid\mathcal H_t
\sim N_{M_t}\!\left(
\mathsf A_t\{C_t\kappa_t+m_t(\eta_t)\}+B_t\alpha_t,
\Omega_t(\lambda_t)\right),
\qquad \Omega_t(\lambda_t)\succ0.
\label{eq:si-latent-report}
\end{equation}
Assume $Z_t^\star\perp S_t\mid\mathcal H_t$ unless a conditional report law given $S_t$ is specified instead. The receipt vector $R_t$ has prespecified law
\begin{equation}
P_{\chi_t}(R_t=r\mid S_t=s,Z_t^\star=z,\mathcal H_t)
=g_{\chi_t}(r\mid s,z,\mathcal H_t),
\label{eq:si-selection}
\end{equation}
which is invariant to counterfactual report coordinates attached to structurally absent dyads.

Each opportunity is exact, left-censored at a known threshold, or recorded in one cell of a prespecified finite, disjoint, exhaustive Borel partition. Nonreturn and structural absence have separate labels. Imputed, reconciled, revised, or otherwise transformed series are either excluded by a prespecified rule or assigned separate report modes with their own declared observation maps. The declared protocol defines a measurable, parameter-free map
\[
\Gamma_t:(s,z,r)\longmapsto(D_t,Y_t^o).
\]
Let $D_t$ collect the discrete pattern, let $\mathcal P(D_t)$ be the exactly returned coordinates, and let $Y_t^o$ collect their values. The observed space is the disjoint union
\[
\mathfrak O_t=\bigsqcup_{d\in\mathfrak D_t}
\bigl(\{d\}\times\R^{|\mathcal P(d)|}\bigr),
\qquad
\mu_t^o=\sum_{d\in\mathfrak D_t}\delta_d\otimes
\operatorname{Leb}^{|\mathcal P(d)|}.
\]
For $o=(d,y)$, let $\mathcal S(d)$ be compatible structural states and let $\mathcal D(d,y,s)$ contain exactly those unreturned coordinates $z_{-\mathcal P(d)}$ for which inserting $y$ in the returned coordinates makes $\Gamma_t\{s,z,r(d)\}=(d,y)$.

\begin{proposition*}[Exact pushforward likelihood]
Let $\nu_t=(\kappa_t,\alpha_t,\lambda_t,\chi_t,\xi_t)$ and
\[
\mu_t(\kappa_t,\eta_t,\alpha_t)
=\mathsf A_t\{C_t\kappa_t+m_t(\eta_t)\}+B_t\alpha_t.
\]
With respect to counting measure on patterns and Lebesgue measure on exact coordinates, the observed-report density is
\begin{align}
L_t^o(\eta_t,\nu_t;o)
&=\sum_{s\in\mathcal S(d)}p_{\xi_t}(s\mid\mathcal H_t)
\int_{\mathcal D(d,y,s)}
\phi_{\Omega_t(\lambda_t)}\{z-\mu_t(\kappa_t,\eta_t,\alpha_t)\}
\notag\\[-2pt]
&\hspace{4em}\times
g_{\chi_t}\{r(d)\mid s,z,\mathcal H_t\}
\,dz_{-\mathcal P(d)},
\label{eq:si-observed-likelihood}
\end{align}
evaluated at $z_{\mathcal P(d)}=y$. If the risk set, node process, or eligible support is not ancillary conditional on $\mathcal H_t$, this is only one factor of the full joint likelihood.
\end{proposition*}

\begin{proof}
With respect to counting measure in $(s,r)$ and Lebesgue measure in $z$, the full-data density is
\[
f_t(s,z,r)=p_{\xi_t}(s\mid\mathcal H_t)
\phi_{\Omega_t(\lambda_t)}\{z-\mu_t\}
g_{\chi_t}(r\mid s,z,\mathcal H_t).
\]
For every nonnegative measurable $\psi$, Tonelli's theorem gives
\[
E[\psi\{\Gamma_t(S_t,Z_t^\star,R_t)\}\mid\mathcal H_t]
=\sum_s\sum_r\int\psi\{\Gamma_t(s,z,r)\}f_t(s,z,r)\,dz.
\]
The declared categories form a measurable partition. On pattern $d$, split $z$ into its exact coordinates $y$ and the remaining coordinates in $\mathcal D(d,y,s)$. This transformation has unit Jacobian. A second use of Tonelli yields the integral of $\psi(o)$ against \eqref{eq:si-observed-likelihood}. Hence that display is the Radon--Nikodym density; setting $\psi=1$ verifies normalization. A nonancillary support or risk-set process contributes its own likelihood factor before the same pushforward operation.
\end{proof}

\begin{lemma*}[DQM under a parameter-free observation map]
Let $\{P_\theta:\theta\in\Theta\subset\R^p\}$ be differentiable in quadratic mean at $\theta_0$ with score $\dot\ell\in L_2^0(P_{\theta_0})$, and suppose $P_{\theta_0+h}\ll P_{\theta_0}$ for every sufficiently small $h$. Let $K$ be a parameter-free Markov kernel. Then $Q_\theta=P_\theta K$ is differentiable in quadratic mean with score
\[
\dot\ell^o(O)=E_{\theta_0}(\dot\ell\mid O),
\]
and its information satisfies
\[
I^o=E(\dot\ell^o\dot\ell^{o\trans})\preceq E(\dot\ell\dot\ell^\trans)=I.
\]
\end{lemma*}

\begin{proof}
Write $P_0=P_{\theta_0}$ and $L_h=dP_{\theta_0+h}/dP_0$. DQM gives
\[
\sqrt{L_h}=1+\tfrac12h^\trans\dot\ell+\Delta_h,
\qquad \norm{\Delta_h}_{L_2(P_0)}=o(\norm h).
\]
Under the joint law induced by $P_0$ and $K$, the observed likelihood ratio is $\bar L_h=E_0(L_h\mid O)$. Put
\[
u_h=\sqrt{L_h}-1,
\qquad m_h=E_0(u_h\mid O),
\qquad v_h=E_0\{(u_h-m_h)^2\mid O\}.
\]
Since $1+m_h=E_0(\sqrt{L_h}\mid O)\ge0$,
\[
\bar L_h=(1+m_h)^2+v_h,
\qquad
d_h:=\sqrt{\bar L_h}-(1+m_h)\ge0.
\]
We show $\norm{d_h}_{L_2(Q_0)}=o(\norm h)$. Always $d_h\le\sqrt{v_h}$; on $\{1+m_h\ge1/2\}$,
\[
d_h=\frac{v_h}{\sqrt{(1+m_h)^2+v_h}+1+m_h}\le v_h.
\]
DQM implies that $u_h^2/\norm h^2$ is uniformly integrable along every sequence $h\to0$. Conditional expectation preserves uniform integrability, so the same is true of $v_h/\norm h^2$. For fixed $\delta>0$,
\[
\frac{E(d_h^2\ind\{1+m_h\ge1/2\})}{\norm h^2}
\le
\delta\frac{Ev_h}{\norm h^2}
+\frac{E\{v_h\ind(v_h>\delta)\}}{\norm h^2}
=O(\delta)+o(1).
\]
Moreover, conditional Jensen and DQM give $P(1+m_h<1/2)=O(\norm h^2)$. Uniform integrability then yields
\[
\frac{E(d_h^2\ind\{1+m_h<1/2\})}{\norm h^2}
\le
\frac{E(v_h\ind\{1+m_h<1/2\})}{\norm h^2}=o(1).
\]
Letting $\delta\downarrow0$ proves the claim. Finally,
\[
m_h=\tfrac12h^\trans E_0(\dot\ell\mid O)+E_0(\Delta_h\mid O),
\qquad
\norm{E_0(\Delta_h\mid O)}_2=o(\norm h),
\]
so
\[
\sqrt{\bar L_h}
=1+\tfrac12h^\trans E_0(\dot\ell\mid O)+o_{L_2(Q_0)}(\norm h).
\]
This proves DQM and the score formula. For every $a\in\R^p$,
\[
a^\trans(I-I^o)a=E\{\Var(a^\trans\dot\ell\mid O)\}\ge0,
\]
which proves the information inequality. If the observation map depends on the parameter, an additional kernel score or moving-boundary term is generally present.
\end{proof}

\subsection{Coarsened-report identification}

Let $S_{\eta,t}^o$ be the observed-data report score for positive-flow composition, let $\mathcal T_{\nu,t}^o$ be the closed nuisance tangent space (including the structural-margin score when present), and let $S_{\eta,t}^{o,\eff}$ denote the residual after orthogonal projection. Put
\[
K_{\eta,t}^o=E(S_{\eta,t}^{o,\eff}S_{\eta,t}^{o,\eff\trans}\mid\mathcal H_t).
\]
Combining it with the residualized outcome derivatives $r_t,H_t$ gives
\[
\mathcal I_t^o=
\begin{pmatrix}
r_t^\trans r_t&\beta_tr_t^\trans H_t\\
\beta_tH_t^\trans r_t&\beta_t^2H_t^\trans H_t+K_{\eta,t}^o
\end{pmatrix}.
\]
Under the conditional-independence, disjoint-nuisance, DQM, and locally constant-rank conditions of the main identification theorem, the same quadratic-form argument shows that $(\beta_t,\eta_t)$ is regularly first-order locally identified exactly when $r_t\ne0$ and
\[
K_{\eta,t}^o+\beta_t^2H_t^\trans P_{r_t}^\perp H_t\succ0.
\]
This is the general coarsened-report extension. It does not state that an arbitrary selection or censoring model has positive information; that must be checked from its observed score.

\subsection{Common-support target}

Datewise identification does not create counterfactual weights outside an observed support. For a comparison or change-analysis horizon $\mathcal T$, a prespecified alternative is to choose once, before inspecting outcomes, the within-row common support $\mathcal E_\cap=\bigcap_{t\in\mathcal T}\mathcal E_t$, require it to be nonempty in every retained row, renormalize every date on $\mathcal E_\cap$, and report entry and exit as a separate extensive-margin target. Holding this support fixed across the entire horizon defines a common composition path but changes the target and discards off-intersection information. Recomputing pairwise intersections after each candidate change would instead make the estimand candidate-dependent and is not the path target used here.

\section{Proof of the one-date identification theorem}
\label{sec:si-identification}
\subsection{Full statement and proof of the main identification theorem}

\begin{theorem*}[Theorem~\ref*{thm:identification-main}, restated]
Fix a date and condition on a sigma-field containing the prespecified outcome and report designs, the lag vector $y_-$, and the eligible support $\mathcal E$. Assume the model in Section~\ref*{sec:model}, differentiability in quadratic mean, locally constant nuisance-score ranks, variation-independent positive-definite covariance nuisances, a nonempty eligible support in each receiving row, and a minimal chart satisfying $C^\trans m=C^\trans\dot m=0$ and $\rank(\dot m)=q$. Assume the outcome and report channels are conditionally independent apart from the common target $\eta$.

Let $X=[D,y_-]$ and define
\[
P_X=L_yX\{(L_yX)^\trans L_yX\}^\dagger(L_yX)^\trans,
\quad M_X=I-P_X,
\quad r=M_XL_yg(\eta),
\quad H=M_XL_yG(\eta),
\]
and, with $U=[AC,B]$,
\[
P_U=L_zU\{(L_zU)^\trans L_zU\}^\dagger(L_zU)^\trans,
\quad M_U=I-P_U,
\quad Q=M_UL_zA\dot m(\eta),
\quad K_c=Q^\trans Q.
\]

\begin{enumerate}[label=\textup{(\alph*)}]
\item The outcome-only efficient information is
\[
I_Y=
\begin{pmatrix}
r^\trans r&\beta r^\trans H\\
\beta H^\trans r&\beta^2H^\trans H
\end{pmatrix}
=[\,r\ \ \beta H\,]^\trans[\,r\ \ \beta H\,].
\]
Thus outcome-only regular local identification holds exactly when $\rank[\,r\ \ \beta H\,]=q+1$.

For the unrestricted complete zero-diagonal row-stochastic family, fix $\beta\ne0$, $\gamma$, the covariance, and all rows except row $i$. Put $x_i=y_{-,-i}$ and
\[
\mathcal F_i(W;y_-)=
\{w\in\R_{++}^{N-1}:\mathbf1^\trans w=1,\ x_i^\trans w=x_i^\trans W_{i,-i}^\trans\}.
\]
This is the exact $W$-only outcome-equivalent fiber in row $i$, of dimension
\[
(N-1)-\rank[\,\mathbf1_{N-1}\ \ y_{-,-i}\,].
\]
It is $N-3$ for a nonconstant lag row and $N-2$ for a constant lag row. At $\beta=0$, the whole row simplex is outcome equivalent. If row normalization is absent and $W(s,\zeta)=s\bar W(\zeta)$, the outcome law is invariant under $(\beta,s)\mapsto(\beta/c,cs)$ whenever admissible.

\item The report efficient information is $K_c$. Reports locally identify composition if and only if $\rank(Q)=q$, equivalently
\[
A\dot m(\eta)h\notin\col(U)\quad\text{for every }h\ne0.
\]
Under the paired mean design
\[
z^E_{ij}=\kappa_i+m_{ij}(\eta)+a^E_j+u_{ij},
\qquad
z^I_{ij}=\kappa_i+m_{ij}(\eta)+a^I_i+v_{ij},
\]
with no further unknown mean effects and positive-definite joint error covariance, the report law globally identifies the row-softmax matrix on the eligible support. If the associated receiver--supplier bipartite graph has $c_{\mathcal E}$ connected components, exactly one location normalization per component is necessary and sufficient for the active linear nuisances.

\item The joint information is
\[
I_c=
\begin{pmatrix}
r^\trans r&\beta r^\trans H\\
\beta H^\trans r&\beta^2H^\trans H+K_c
\end{pmatrix}.
\]
The following are equivalent: $I_c\succ0$; no $(b,h)\ne0$ satisfies $rb+\beta Hh=0$ and $Qh=0$; and
\[
r\ne0,
\qquad
S_{\eta,c}:=K_c+\beta^2H^\trans P_r^\perp H\succ0.
\]
For $\beta\ne0$, this is equivalent to
\[
\ker(Q)\cap\{h:Hh\in\Span(r)\}=\{0\};
\]
for $\beta=0$, it is equivalent to $r\ne0$ and $K_c\succ0$.

\item Let
\[
B_c=\begin{pmatrix}r&\beta H\\0&Q\end{pmatrix}.
\]
Then $I_c=B_c^\trans B_c$ and, for every direction $(b,h)$,
\[
(b,h^\trans)I_c(b,h^\trans)^\trans
=\norm{rb+\beta Hh}^2+\norm{Qh}^2.
\]
Hence the smallest singular value of a prespecified coordinate scaling of $B_c$ is a finite-experiment sensitivity diagnostic and vanishes exactly at first-order nonidentification. If $D$ excludes the own lag, define
\[
u_D=M_{L_yD}L_yy_-,
\qquad
v_D=M_{L_yD}L_yg(\eta).
\]
When both are nonzero,
\[
\norm r^2=\norm{v_D}^2\{1-\cos^2(u_D,v_D)\},
\]
which isolates loss of strength information caused by own-lag collinearity.
\end{enumerate}
\end{theorem*}

\begin{proof}
All scores and projections are evaluated at the truth.

\paragraph{Step 1: the row gauge and softmax derivative.}
Adding $Cc$ to a dyad log-flow vector adds the constant $c_i$ to every eligible entry in receiving row $i$, so it cancels from row normalization. Conversely, equality of two positive row-softmax vectors implies equality of every within-row log ratio; their log-flow difference is therefore row constant. The constraint $C^\trans m=0$ removes this gauge. Differentiation gives, for eligible $(i,j)$,
\[
\frac{\partial W_{ij}}{\partial\eta^\trans}
=W_{ij}\left\{
\dot m_{ij}^\trans-
\sum_{k:(i,k)\in\mathcal E}W_{ik}\dot m_{ik}^\trans
\right\}.
\]
Thus $G(\eta)=D_\eta\{W(\eta)y_-\}$ is well defined on the minimal chart.

\paragraph{Step 2: outcome score and information.}
Whitening the outcome and projecting off $L_yX$ leaves mean derivatives $r$ for $\beta$ and $\beta H$ for $\eta$. Gaussian mean scores are linear in the whitened innovation. Gaussian covariance scores are centered quadratic forms, and their covariance with a centered linear form is zero by vanishing third moments. Hence covariance profiling removes no additional mean information. The efficient information is the Gram matrix of $[r,\beta H]$, proving part (a)'s rank statement.

For the exact fiber, the conditional mean in row $i$ depends on $w$ only through $x_i^\trans w$ once the row-sum constraint is imposed. Intersecting the positive simplex with these linear equalities gives exactly $\mathcal F_i$. At an interior point its dimension is the ambient dimension $N-1$ minus $\rank[\mathbf1,x_i]$. The remaining special cases follow by the stated ranks, by setting $\beta=0$, and by observing that $\beta s$ is invariant under reciprocal rescaling.

\paragraph{Step 3: report score and paired identification.}
After whitening the report channel, the composition derivative is $L_zA\dot m$ and the linear nuisance derivative is $L_zU$. Orthogonal projection off the nuisance columns gives $Q$, so the efficient information is $Q^\trans Q$. For $h\ne0$,
\[
h^\trans K_ch=\norm{Qh}^2,
\]
which is positive exactly under the displayed rank condition.

Under the paired design, equality of two parameterized mean arrays implies
\[
\Delta\ell_{ij}+\Delta a_j^E=0,
\qquad
\Delta\ell_{ij}+\Delta a_i^I=0.
\]
Thus the log-flow difference is constant along every connected receiver--supplier component and, in particular, constant within each receiving row. Row normalization removes these constants, proving global identification of $W_{\mathcal E}$. The null space of an incidence-type nuisance design contains one additive location per connected component and no others, proving the normalization count. Appending an unrestricted common dyad column makes every composition derivative a nuisance direction and destroys the conclusion.

\paragraph{Step 4: joint information and Schur complement.}
Conditional independence of the two channels adds their efficient information matrices, yielding $I_c$. Its quadratic form is
\[
\norm{rb+\beta Hh}^2+\norm{Qh}^2.
\]
Therefore $I_c$ is positive definite exactly when the two expressions cannot vanish jointly outside the zero direction. If $r\ne0$, minimizing the first term over $b$ leaves
\[
\beta^2\norm{P_r^\perp Hh}^2+\norm{Qh}^2
=h^\trans S_{\eta,c}h.
\]
This proves the Schur-complement condition and the kernel formulation. If $r=0$, the coefficient direction is null; if $\beta=0$, the composition information is exactly $K_c$.

\paragraph{Step 5: geometry and own-lag collinearity.}
The factorization $I_c=B_c^\trans B_c$ proves the singular-value claims. To partial out the own lag, first residualize both $y_-$ and $g(\eta)$ against $D$. Residualizing $v_D$ once more against the span of $u_D$ gives
\[
r=\left(I-\frac{u_Du_D^\trans}{u_D^\trans u_D}\right)v_D.
\]
Pythagoras yields the angular identity. This completes the proof.
\end{proof}

\begin{remark}[Scope]
The theorem is local and first order except for the explicitly stated exact fibers, scale invariance, and paired-design identification of the row-softmax matrix. Small positive singular values diagnose weak finite-experiment separation only after a coordinate scaling is fixed; an asymptotic weak-identification theorem requires a specified sequence of experiments.
\end{remark}

\section{Proof of the plug-in theorem}
\label{sec:si-plugin}
\subsection{Proof of the plug-in theorem}

\begin{proof}[Proof of Theorem~\ref*{thm:plugin-main}]
Throughout, condition on $\mathcal F_{t-1}$ and suppress $t$.

\paragraph{(a).}
The population linear projection of $Y$ on $[X,\ \bar Wy_-]$ solves the
normal equations; partialling $X$ out by $M=M_X$, the coefficient on
$\bar Wy_-$ is
\[
\beta^{\mathrm{plug}}
=\frac{\inner{M\bar Wy_-}{M\,E(Y\mid\mathcal F_{t-1})}}{\norm{M\bar Wy_-}^2}
=\frac{\inner{u}{M\{X\gamma+\beta W(\eta)y_-\}}}{\norm u^2}
=\beta\,\frac{\inner{u}{v(\eta)}}{\norm u^2},
\]
because $MX=0$ and the innovation is conditionally centered. The noise level
cancels. Equality with $\beta$ holds iff $\inner{u}{v-u}=0$; the attenuation,
zero, and sign-reversal regimes are immediate from the sign and size of
$\varphi(\eta)=\inner uv$. The worked $N=4$ design of
Section~\ref*{sec:charts-main} realizes all three regimes (replication code;
the sign-reversing direction is reported there).

\paragraph{(b).}
The shift formula is (a) applied at $\eta^{(1)}$ and $\eta^{(0)}$. For the
genericity claim: on the complete support with the linear chart,
$\eta\mapsto W(\eta)$ has entries
$W_{ij}(\eta)=\exp(\Psi_{ij}^\trans\eta)/\sum_k\exp(\Psi_{ik}^\trans\eta)$,
a ratio of finite sums of exponentials of linear functions, hence real
analytic on $\R^q$; so is
$\varphi(\eta)=u^\trans MW(\eta)y_-$, and
$\psi(\eta^{(0)},\eta^{(1)})=\varphi(\eta^{(1)})-\varphi(\eta^{(0)})$ is real
analytic on $\R^{2q}$. Since
$D_\eta\varphi(\eta)=u^\trans MG(\eta)$, the stated condition makes
$\varphi$ nonconstant, hence $\psi$ is not identically zero, and the zero
set of a nonconstant real-analytic function on a connected open set has
Lebesgue measure zero (see, e.g., \citealp{KrantzParks2002}). Off that null set $\Delta^{\mathrm{plug}}\ne0$, and
any test consistent against fixed mean shifts of the projection coefficient
rejects with probability tending to one as information accumulates.

\paragraph{(c).}
With stochastic regressor $\tilde u=M\widetilde Wy_-$ independent of the
date-$t$ outcome innovation, the population projection coefficient on the
partialled exposure is
\[
\beta^{\mathrm{conc}}
=\frac{E\inner{\tilde u}{MY}}{E\norm{\tilde u}^2}
=\beta\,\frac{E\inner{\tilde u}{v(\eta)}}{E\norm{\tilde u}^2}
=\beta\,\frac{\inner{E\tilde u}{v(\eta)}}
{\norm{E\tilde u}^2+E\norm{\tilde u-E\tilde u}^2},
\]
using $E\inner{\tilde u}{v}=\inner{E\tilde u}{v}$ (the target vector $v$ is
$\mathcal F_{t-1}$-measurable) and the variance decomposition of
$E\norm{\tilde u}^2$. Nondegenerate report noise makes the second
denominator term strictly positive; softmax and projection are measurable
bounded maps, so all expectations exist. Attenuation is strict and does not
average out across dates because each date contributes its own noise
denominator.
\end{proof}

\begin{remark}[Scope]
The theorem is a statement about linear projections, so it applies verbatim
to weighted or heteroskedasticity-robust variants of plug-in practice (the
projection geometry changes with the weighting, not the structure of the
result). It does not require the chart to be correctly specified: parts
(a)--(c) hold with $W(\eta_t)$ replaced by the true network path, with
$\varphi$ evaluated along it.
\end{remark}

\section{Orthogonality, path expansion, and pilot conditions}
\label{sec:si-estimation}
\subsection{Exact nuisance orthogonality}

\begin{proposition*}[Exact nuisance orthogonality]
Fix $(t,k)$ and positive-definite working covariances $V_Y,V_z$ measurable with respect to the conditioning/training sigma-field. Define
\[
R^Y=M_{V_Y^{-1/2}X}V_Y^{-1/2},
\qquad
R^z=M_{V_z^{-1/2}U}V_z^{-1/2}.
\]
For arbitrary candidate linear nuisances $\bar\gamma,\bar\lambda$, let
\[
\widetilde e(\theta,\bar\gamma,\bar\lambda)=
\begin{pmatrix}
R^Y\{Y-X\bar\gamma-\beta g(\eta)\}\\
R^z\{z-U\bar\lambda-Am(\eta)\}
\end{pmatrix},
\]
let $J(\theta;V_Y,V_z)$ be the Jacobian in (\ref*{eq:jacobian}) with these residualizers, and put $\widetilde S=J^\trans\widetilde e$.

\begin{enumerate}[label=\textup{(\alph*)}]
\item With the working covariances fixed, $D_{\bar\gamma}\widetilde S=D_{\bar\lambda}\widetilde S=0$ sample by sample.
\item At the truth, $E(\widetilde S\mid\mathcal C)=0$ for every admissible $(\bar\gamma,\bar\lambda,V_Y,V_z)$. Consequently, derivatives of the conditional expected score with respect to admissible covariance directions are zero.
\item Under the conditional Gaussian model, every score coordinate is $L_2$-orthogonal to every allowed variation-independent Gaussian covariance score. At the true covariances, if $X$ and $U$ span the complete linear mean-nuisance tangent spaces, $\widetilde S$ is the nuisance-projected Gaussian likelihood score.
\end{enumerate}
\end{proposition*}

\begin{proof}
Because $R^YX=R^zU=0$, changing $\bar\gamma$ or $\bar\lambda$ leaves $\widetilde e$, and hence $\widetilde S$, unchanged. At the truth, the same annihilation reduces $\widetilde e$ to fixed residualizer matrices multiplying centered held-out innovations. Conditional centering therefore gives $E(\widetilde S\mid\mathcal C)=0$ for every working covariance; differentiation of this pointwise identity proves the expected-score orthogonality.

For a covariance direction $H$, a Gaussian covariance score has the form
\[
s_H(\varepsilon)=\frac12\left\{
\varepsilon^\trans\mathcal V^{-1}H\mathcal V^{-1}\varepsilon
-\tr(\mathcal V^{-1}H)
\right\},
\]
whereas each target-score coordinate is $a^\trans\varepsilon$. Their covariance vanishes by centered first and third Gaussian moments. At the true covariances, generalized least-squares projection off the complete linear mean-nuisance columns is precisely likelihood-score projection.
\end{proof}

\subsection{Uniform expansion}

\begin{proof}[Proof of Theorem~\ref*{thm:path-main}]
Write $\ell_n=\log(2T_nK_f)$. Define the simultaneous good events
\[
\mathcal E^V_{tk}=
\left\{
\norm{\widehat\Sigma_{tk}^{(-k)}-\Sigma_{tk}}_{\op}
+\norm{\widehat\Omega_{tk}^{(-k)}-\Omega_{tk}}_{\op}
\le\delta_n
\right\},
\]
and $\mathcal E^\theta_{tk}=\{\norm{\theta^0_{t,-k}-\theta_t}\le a_n\}$. Conditions (E5)--(E6) make their global intersection have probability tending to one uniformly. Conditional score calculations below condition only on $\mathcal T_{tk}$ and use the fold-local event; they do not condition on the global intersection.

\paragraph{1. Residualizer perturbation.}
On the clipped spectral interval, the inverse-square-root map is operator Lipschitz. This follows from
\[
V^{-1/2}=\frac1\pi\int_0^\infty s^{-1/2}(V+sI)^{-1}\,ds
\]
and the resolvent identity. The nonzero singular-value bounds in (E2), the formula $P_B=B(B^\trans B)^\dagger B^\trans$, and standard inverse perturbation therefore yield
\begin{equation}
\max_{t,k}\left(
\norm{\widehat R^Y_{tk}-R^Y_{tk}}_{\op}
+\norm{\widehat R^z_{tk}-R^z_{tk}}_{\op}
\right)\lesssim\delta_n.
\label{eq:si-residualizer-perturbation}
\end{equation}
Together with the columnwise norm bounds in (E3), the operator bound
\eqref{eq:si-residualizer-perturbation} gives
\begin{equation}
\max_{t,k}\sup_{\norm{\theta-\theta_t}\le r_0/2}
\norm{\widehat J_{tk}(\theta)-J_{tk}(\theta)}_{\op}
\lesssim\delta_n\sqrt{n_{tk}}.
\label{eq:si-jacobian-perturbation}
\end{equation}
Consequently, if $\widehat{\mathcal I}_t^\star=\sum_k\widehat{\mathcal I}_{tk}(\theta_t)$,
\begin{equation}
\max_t n^{-1}\norm{\widehat{\mathcal I}_t^\star-\mathcal I_t}_{\op}
=O_{\PP}^{\mathrm{unif}}(\delta_n).
\label{eq:si-information-truth}
\end{equation}

\paragraph{2. Gaussian maximal bounds.}
At the truth the whitened residual covariance is an orthogonal projection, so
\[
S_t\mid\mathcal F_{t-1}\sim N_d(0,\mathcal I_t).
\]
Condition (E4) and a Gaussian union bound imply
\begin{equation}
\max_{t\le T_n}\norm{S_t}
=O_{\PP}^{\mathrm{unif}}(\sqrt{n\ell_n}).
\label{eq:si-score-max}
\end{equation}
Conditionally on $\mathcal T_{tk}$, the feasible-oracle score difference is a centered Gaussian linear form. Equations \eqref{eq:si-residualizer-perturbation}--\eqref{eq:si-jacobian-perturbation} and a union bound give
\begin{equation}
\max_t\norm{\sum_k\{\widehat S_{tk}(\theta_t)-S_{tk}(\theta_t)\}}
=O_{\PP}^{\mathrm{unif}}(\delta_n\sqrt{n\ell_n}).
\label{eq:si-feasible-score-max}
\end{equation}

\paragraph{3. Derivatives and Taylor expansion.}
Holding training quantities fixed,
\[
D\widehat e_{tk}(\theta)[h]=-\widehat J_{tk}(\theta)h,
\]
and hence
\[
D\widehat S_{tk}(\theta_t)[h]
=-\widehat{\mathcal I}_{tk}(\theta_t)h+Z_{tk}h,
\quad
Z_{tk}h=\{D\widehat J_{tk}(\theta_t)[h]\}^\trans\widehat e_{tk}(\theta_t).
\]
Conditional Gaussian concentration and (E3) give
\[
\max_{t,k}\norm{Z_{tk}}_{\op}
=O_{\PP}^{\mathrm{unif}}(\sqrt{n\ell_n}),
\qquad
\max_{t,k}\sup_{\norm{\theta-\theta_t}\le r_0}
\norm{D^2\widehat S_{tk}(\theta)}_{\op}
=O_{\PP}^{\mathrm{unif}}(n).
\]
With $d_{tk}=\theta^0_{t,-k}-\theta_t$, Taylor's theorem gives
\begin{equation}
\widehat S_{tk}(\theta_t+d_{tk})
=\widehat S_{tk}(\theta_t)
-\widehat{\mathcal I}_{tk}(\theta_t)d_{tk}
+Z_{tk}d_{tk}+r_{tk},
\quad
\max_{t,k}\norm{r_{tk}}=O_{\PP}^{\mathrm{unif}}(na_n^2).
\label{eq:si-score-taylor}
\end{equation}
Similarly,
\[
\max_{t,k}
\norm{\widehat{\mathcal I}_{tk}(\theta^0_{t,-k})
-\widehat{\mathcal I}_{tk}(\theta_t)}_{\op}
\lesssim n_{tk}a_n.
\]
Combining this with \eqref{eq:si-information-truth} proves (\ref*{eq:information-rate-main}). Weyl's inequality and (E4) then imply simultaneous nonsingularity and equality of the safe and ordinary inverses with probability tending to one.

\paragraph{4. Cancellation and oracle replacement.}
Substituting \eqref{eq:si-score-taylor} into (\ref*{eq:estimator}) cancels the first-order pilot term exactly. The inverse bound, the bounds on $Z_{tk}$ and $r_{tk}$, and fixed $K_f$ yield
\[
\max_t
\norm{\widehat\theta_t-\theta_t
-\widehat{\mathcal I}_t^{-1}\sum_k\widehat S_{tk}(\theta_t)}
=O_{\PP}^{\mathrm{unif}}\!\left(
a_n^2+a_n\sqrt{\ell_n/n}
\right).
\]
Finally, using
\[
\widehat{\mathcal I}_t^{-1}-\mathcal I_t^{-1}
=\widehat{\mathcal I}_t^{-1}
(\mathcal I_t-\widehat{\mathcal I}_t)\mathcal I_t^{-1}
\]
with \eqref{eq:si-score-max}, \eqref{eq:si-feasible-score-max}, and (\ref*{eq:information-rate-main}) proves (\ref*{eq:path-expansion-main}). Adding the oracle term $\mathcal I_t^{-1}S_t=O_{\PP}^{\mathrm{unif}}(\sqrt{\ell_n/n})$ proves (\ref*{eq:path-rate-main}).
\end{proof}

\subsection{Why the leave-fold pilot requires its own condition}

\begin{proposition*}[Full-date information does not imply pilot consistency]
There is a Gaussian triangular array satisfying (E1)--(E5) for which no pilot measurable with respect to one leave-fold training sample can satisfy (E6) uniformly over a two-point model class.
\end{proposition*}

\begin{proof}
Let $T_n=1$, $K_f=2$, $m=n/2$, and consider a scalar target. Put
\[
Y_1=\theta\mathbf1_m+\varepsilon_1,
\qquad
Y_2=\varepsilon_2,
\qquad
\varepsilon_1,\varepsilon_2\stackrel{\mathrm{ind}}\sim N_m(0,I_m).
\]
The full-data information is $m\asymp n$, and all other regularity conditions are immediate. If fold 1 is held out, the training law is the same under two distinct interior values $\theta_0,\theta_1$. For any training-measurable pilot $\widetilde\theta$ and any $a<|\theta_1-\theta_0|/2$, the two events $\{|\widetilde\theta-\theta_j|\le a\}$ are disjoint under their common training law. Hence at least one parameter has pilot-error probability at least $1/2$. No deterministic $a_n\to0$ can satisfy (E6).
\end{proof}

\subsection{A sufficient profiled pilot}

Stack the training folds $j\ne k$ as
\[
O_{t,-k}=V_{t,-k}c_t+\mu_{t,-k}(\theta_t)+\varepsilon_{t,-k},
\]
where $V_{t,-k}$ contains all prespecified constrained linear outcome and report nuisances. Let $R_{t,-k}$ be the oracle whitened residualizer and $\widehat R_{t,-k}$ its training estimate. On a fixed compact convex $\Theta$ containing $B(\theta_t,r_0)$ for every $t$, with $\mu$ twice differentiable on a neighborhood of $\Theta$, define the measurable lexicographically selected minimizer
\[
\theta^0_{t,-k}=\operatorname{lexmin}\arg\min_{\theta\in\Theta}
\frac12\norm{\widehat R_{t,-k}\{O_{t,-k}-\mu_{t,-k}(\theta)\}}^2.
\]

\begin{theorem*}[Uniform rate of the stacked profiled pilot]
Let $L_n=\log(2T_nK_f)+x_n$, where $x_n\to\infty$ and $L_n=o(n)$. Assume: conditional Gaussian training errors with $O(n)$ dimension and uniformly bounded covariance; a fixed compact target with twice differentiable $\mu$ and derivatives $O(\!(\sqrt n))$; $O(n)$ residual degrees of freedom and uniformly stable whitening; the leave-fold separation
\[
\norm{R_{t,-k}\{\mu_{t,-k}(\theta)-\mu_{t,-k}(\theta_t)\}}^2
\ge cn\min\{\norm{\theta-\theta_t}^2,r_0^2\};
\]
and
\[
P\left\{\max_{t,k}\norm{\widehat R_{t,-k}-R_{t,-k}}_{\op}>C\rho_n^{\mathrm{tr}}\right\}
\le\epsilon_n,
\qquad \rho_n^{\mathrm{tr}}=o(1).
\]
Then, for constants $C_1,c_1>0$,
\[
\sup_{P\in\mathcal P_n}
P\left\{
\max_{t,k}\norm{\theta^0_{t,-k}-\theta_t}
>C_1\left(\sqrt{L_n/n}+\rho_n^{\mathrm{tr}}\right)
\right\}
\le\epsilon_n+C_1e^{-c_1x_n}.
\]
Thus (E6) holds when $a_n=C_1\{\sqrt{L_n/n}+\rho_n^{\mathrm{tr}}\}=o(1)$.
\end{theorem*}

\begin{proof}
Write $\Delta(\theta)=\mu(\theta)-\mu(\theta_t)$. Exact annihilation of $V$ gives
\[
Q(\theta)-Q(\theta_t)
=\tfrac12\norm{\widehat R\Delta(\theta)}^2
-\inner{\widehat R\varepsilon}{\widehat R\Delta(\theta)}.
\]
The fixed-dimensional, twice differentiable image classes formed by $R\Delta(\theta)$ and $R\Delta(\theta)/\norm{\theta-\theta_t}$ have polynomial covering numbers and radius $O(\!(\sqrt n))$. Dudley's inequality and Borell concentration, followed by a union bound over $(t,k)$, give simultaneous stochastic-process bounds $O(\!(\sqrt{nL_n}))$, locally per unit $\norm{\theta-\theta_t}$ and globally.

Residualizer perturbation changes the inner product by at most $Cn\rho_n^{\mathrm{tr}}\norm{\theta-\theta_t}$ locally. The separation assumption and $\rho_n^{\mathrm{tr}}=o(1)$ give curvature $cn\norm{\theta-\theta_t}^2$ inside the $r_0$ ball and $cnr_0^2$ outside it. The global curvature dominates the global stochastic term, so every minimizer lies in the local ball. Comparing $Q(\theta^0)$ with $Q(\theta_t)$ then gives
\[
cn\norm{\theta^0-\theta_t}^2
\le C\{\sqrt{nL_n}+n\rho_n^{\mathrm{tr}}\}\norm{\theta^0-\theta_t},
\]
which proves the rate after division.
\end{proof}

\begin{remark}[Nuisance extensions]
Prespecified reporter, partner, and report-mode coefficients are profiled exactly, although the separation condition depends on their design. Fixed-dimensional covariance or conditionally exogenous sampling parameters contribute to $\rho_n^{\mathrm{tr}}$ through their residualizer perturbation rate. Nonignorable selection requires the joint observed-data likelihood and is not covered by conditioning or inverse-probability weighting unless its own overlap, score, and separation assumptions are verified.
\end{remark}

\section{Proof of the simultaneous coefficient-path band}
\label{sec:si-bands}
\begin{proof}[Proof of Theorem~\ref*{thm:bands-main}]
Write $T=T_n$ and $q_n=a_n+\delta_n$. If $q_n=0$ along a subsequence, the information approximation is exact there and the following bounds hold with zero right side; otherwise restrict to the eventually positive subsequence. Put
\[
Z_t=
\frac{e_1^\trans\mathcal I_t^{-1}S_t}
{\sqrt{v_{\beta,t}}},
\qquad
v_{\beta,t}=e_1^\trans\mathcal I_t^{-1}e_1.
\]

\paragraph{1. Exact leading law.}
Conditionally on $\mathcal F_{n,t-1}$, the numerator is Gaussian with variance $v_{\beta,t}$, so $Z_t\mid\mathcal F_{n,t-1}\sim N(0,1)$. This conditional distribution is nonrandom, hence $Z_t$ is independent of $\mathcal F_{n,t-1}$. Since earlier $Z_s$ are measurable with respect to that sigma-field, iteration gives
\begin{equation}
Z_1,\ldots,Z_T\stackrel{\mathrm{iid}}\sim N(0,1).
\label{eq:si-band-iid}
\end{equation}

\paragraph{2. Safe inverse and variance consistency.}
The information bounds imply
\[
(C_In)^{-1}\le v_{\beta,t}\le(c_In)^{-1}.
\]
Let $D_n=\max_tn^{-1}\norm{\widehat{\mathcal I}_t-\mathcal I_t}_{\op}$. By Theorem~\ref*{thm:path-main}, $D_n=O_{\PP}^{\mathrm{unif}}(q_n)=o_{\PP}^{\mathrm{unif}}(1)$. On $\{D_n\le c_I/2\}$, Weyl's inequality makes the safe inverse equal the ordinary inverse, and the inverse identity yields
\[
\max_t\norm{\widehat{\mathcal I}_t^\ominus-\mathcal I_t^{-1}}_{\op}
=O_{\PP}^{\mathrm{unif}}(q_n/n).
\]
The safe inverse is $O(n^{-1})$ on every sample, so the same stochastic order holds without an exceptional undefined event. Consequently,
\begin{equation}
\max_t\left|\frac{\widehat v_{\beta,t}}{v_{\beta,t}}-1\right|
=O_{\PP}^{\mathrm{unif}}(q_n)
=o_{\PP}^{\mathrm{unif}}(1).
\label{eq:si-band-var}
\end{equation}

\paragraph{3. Studentized approximation.}
Let
\[
r_t=\widehat\theta_t-\theta_t-\mathcal I_t^{-1}S_t,
\qquad
\widehat Z_t=\frac{\widehat\beta_t-\beta_t}{\sqrt{\widehat v_{\beta,t}}}.
\]
Then
\[
\widehat Z_t-Z_t
=\frac{e_1^\trans r_t}{\sqrt{\widehat v_{\beta,t}}}
+Z_t\left(\sqrt{\frac{v_{\beta,t}}{\widehat v_{\beta,t}}}-1\right).
\]
The path expansion, \eqref{eq:si-band-var}, and a Gaussian union bound give
\begin{align*}
\Delta_n
&:=\max_t|\widehat Z_t-Z_t|\\
&=O_{\PP}^{\mathrm{unif}}\left(
\sqrt n\,a_n^2+q_n\sqrt{\ell_n}
\right),
\end{align*}
and therefore
\begin{equation}
\Delta_n\sqrt{\ell_n}
=O_{\PP}^{\mathrm{unif}}\left(
\sqrt{n\ell_n}\,a_n^2+q_n\ell_n
\right)
=o_{\PP}^{\mathrm{unif}}(1).
\label{eq:si-band-small}
\end{equation}

\paragraph{4. Anti-concentration and transfer.}
Let $M_T=\max_t|Z_t|$. Its distribution function and density on $x>0$ are
\[
F_T(x)=\{2\Phi(x)-1\}^T,
\qquad
f_T(x)=2T\phi(x)\{2\Phi(x)-1\}^{T-1}.
\]
Using $T(1-u)u^{T-1}\le1$ and standard Mills bounds gives the uniform density bound
\[
\sup_{x>0}f_T(x)\le C\sqrt{\log(2T)}.
\]
Hence
\[
\sup_xP(|M_T-x|\le\varepsilon)
\le2C\varepsilon\sqrt{\log(2T)}.
\]
From \eqref{eq:si-band-small}, choose a deterministic $\varepsilon_n=o(\ell_n^{-1/2})$ such that $\sup_PP(\Delta_n>\varepsilon_n)\to0$. The maximum map is one-Lipschitz, so the feasible maximum $\widehat M_T=\max_t|\widehat Z_t|$ obeys
\[
\sup_{P\in\mathcal P_n}\sup_x
|P(\widehat M_T\le x)-F_T(x)|\to0.
\]

\paragraph{5. Calibration.}
By definition of $c_{1-\alpha,T}$,
\[
F_T(c_{1-\alpha,T})
=\{2\Phi(c_{1-\alpha,T})-1\}^T
=1-\alpha.
\]
The event $\{\beta_t\in C_t\text{ for all }t\}$ equals $\{\widehat M_T\le c_{1-\alpha,T}\}$. Evaluating the preceding uniform cdf approximation at this critical value proves the result.
\end{proof}

\begin{remark}[What supplies exact calibration]
The sequential conditional-Gaussian law makes the standardized leading scores iid standard normal even when the predictable information matrices vary with $t$. Cross-fitting is used for nuisance estimation, but it does not by itself establish this sequential law. Under non-Gaussian or temporally dependent leading scores, a different high-dimensional approximation and critical-value construction is required.
\end{remark}

\section{The band under non-Gaussian scores}
\label{sec:si-robust}
This section removes the exact conditional Gaussian law from the simultaneous
band. The leading standardized scores are martingale differences with unit
conditional variance whatever the innovation law; what Gaussianity added was
the exact conditional $N(0,1)$ distribution. We show that the same
\v{S}id\'{a}k band remains uniformly valid when the held-out innovations are
conditionally centered, standardized, independent across coordinates, and
sub-exponential, provided a computable leverage quantity is small relative to
$\log T_n$. We also give a finite-sample conservative alternative that
requires no distributional approximation at all.

\subsection{Conditions}

Recall the oracle leading term $Z_t=e_1^\trans\mathcal I_t^{-1}S_t/\sqrt{v_{\beta,t}}$.
Because $S_t=\sum_k J_{tk}^\trans R_{tk}\varepsilon_{tk}$ is linear in the
held-out innovation vector $\varepsilon_t$ (the stacked
$(\xi_{tk},\upsilon_{tk})$ over folds), we may write
\[
Z_t=a_t^\trans\varepsilon_t,
\qquad
a_t\ \text{is}\ \mathcal F_{n,t-1}\text{-measurable},
\qquad
\norm{a_t}_2=1 ,
\]
where the unit norm is the definition of the standardization when the
conditional covariance of $\varepsilon_t$ is the identity.

\begin{description}[style=nextline]
\item[(B1) Standardized sub-exponential innovations.]
Conditionally on $\mathcal F_{n,t-1}$, the coordinates of $\varepsilon_t$ are
independent, centered, have unit variance, and satisfy
$\sup_{t,i}\norm{\varepsilon_{t,i}}_{\psi_1}\le\bar b$ for a fixed $\bar b$.
The conditional law may vary with $t$, $i$, and $P$, and may be skewed.

\item[(B2) Sequential structure.]
As in Theorem~\ref*{thm:bands-main}: $S_t$ is $\mathcal F_{n,t}$-measurable,
$\mathcal I_t$ and the weight vectors $a_t$ are
$\mathcal F_{n,t-1}$-measurable, and conditionally on $\mathcal F_{n,t-1}$
the date-$t$ held-out innovations are independent of the past. Under (E1)
this holds by construction of the folds.

\item[(B3) Leverage.]
$\displaystyle\mathrm{lev}_n:=\max_{t\le T_n}\norm{a_t}_\infty$ satisfies
$\mathrm{lev}_n\{1+(\log T_n)^{3/2}\}\to0$.
\end{description}

The quantity $\mathrm{lev}_n$ is the maximal contribution of a single
held-out coordinate to a standardized score and is computable from the
design; it is the exact analogue of a leverage diagnostic. Under the
balanced-design bounds in (E2)--(E4), with rowwise Jacobian norms of order
one, $\mathrm{lev}_n=O(n^{-1/2})$, and (B3) reduces to
$(\log T_n)^{3}=o(n)$, the standard moderate-deviations regime.

\subsection{A Cram\'er-type conditional coverage lemma}

\begin{lemma*}[Conditional relative moderate deviations]
Let $\zeta=\sum_ia_i\varepsilon_i$ with independent centered unit-variance
$\varepsilon_i$, $\norm{\varepsilon_i}_{\psi_1}\le\bar b$, $\norm a_2=1$,
$\norm a_\infty\le\kappa$. There are constants $c_0,C_0$ depending only on
$\bar b$ such that for all $0\le x\le c_0\kappa^{-1/3}$,
\[
\left|\frac{P(\zeta>x)}{1-\Phi(x)}-1\right|\le C_0(1+x^3)\kappa ,
\qquad
\left|\frac{P(\zeta<-x)}{1-\Phi(x)}-1\right|\le C_0(1+x^3)\kappa .
\]
\end{lemma*}

\begin{proof}[Proof sketch and source]
This is the classical Cram\'er--Petrov relative moderate-deviation expansion
for sums of independent, non-identically distributed random variables under a
Bernstein--Statulevi\v{c}ius moment condition, which the $\psi_1$ bound and
the weight normalization imply with cumulant growth constant proportional to
$\kappa$; see \citet[Ch.~VIII]{Petrov1975} and
\citet[Ch.~3]{SaulisStatulevicius1991}. The one-line reduction: the $j$th
cumulant of $\zeta$ is bounded by $\sum_i|a_i|^j\,j!\,(C\bar b)^j\le
j!\,(C\bar b)^j\kappa^{j-2}$ for $j\ge3$, which is the Statulevi\v{c}ius
condition with parameter $\Delta\asymp\kappa^{-1}$, and the cited expansion
applies on $0\le x\le c_0\Delta^{1/3}$.
\end{proof}

\subsection{Robust validity of the \v{S}id\'{a}k band}

\begin{theorem*}[Theorem~\ref*{thm:robust-band-main}, restated]
Assume (E1)--(E6) with the conditional Gaussian specification in
(\ref*{eq:outcome-model})--(\ref*{eq:report-model}) replaced by (B1), and
assume (B2), (B3), and
\[
b_n^{\mathrm{rob}}
:=\left(\sqrt n\,a_n^2+(a_n+\delta_n)\sqrt{\ell_n}\right)\sqrt{\ell_n}\to0 .
\]
Then the band of Theorem~\ref*{thm:bands-main}, with the same critical value
$c_{1-\alpha,T_n}$, satisfies
\[
\sup_{P\in\mathcal P_n}
\left|P\{\beta_t\in C_t\ \text{for all}\ t\le T_n\}-(1-\alpha)\right|
\longrightarrow0 .
\]
\end{theorem*}

\begin{proof}
Write $T=T_n$, $c=c_{1-\alpha,T}$, $\pi=2\Phi(c)-1=(1-\alpha)^{1/T}$, and
$A_t=\{|Z_t|\le c\}$.

\paragraph{1. Backward induction to per-date conditional coverage.}
For any events adapted as above,
\[
P\Big(\bigcap_{s\le t}A_s\Big)
=\pi\,P\Big(\bigcap_{s\le t-1}A_s\Big)
+E\Big[\ind\Big\{\bigcap_{s\le t-1}A_s\Big\}
\big\{P(A_t\mid\mathcal F_{n,t-1})-\pi\big\}\Big].
\]
Iterating and using $\pi\le1$,
\begin{equation}
\Big|P\Big(\bigcap_{t\le T}A_t\Big)-\pi^T\Big|
\le\sum_{t\le T}\operatorname*{ess\,sup}
\big|P(A_t\mid\mathcal F_{n,t-1})-\pi\big| .
\label{eq:si-rb-telescope}
\end{equation}

\paragraph{2. Per-date error via the lemma.}
Conditionally on $\mathcal F_{n,t-1}$, $Z_t=a_t^\trans\varepsilon_t$
satisfies the lemma's hypotheses with $\kappa\le\mathrm{lev}_n$. Since
$c\le\sqrt{2\log(2T/\alpha)}$ and (B3) keeps $c\le c_0\,\mathrm{lev}_n^{-1/3}$
eventually,
\[
\big|P(A_t\mid\mathcal F_{n,t-1})-\pi\big|
\le2C_0(1+c^3)\,\mathrm{lev}_n\{1-\Phi(c)\}
\le C_1(1+c^3)\,\mathrm{lev}_n\,\frac{1}{T},
\]
where we used $1-\Phi(c)=(1-\pi)/2=\{1-(1-\alpha)^{1/T}\}/2\le
-\log(1-\alpha)/(2T)$, and $-\log(1-\alpha)$ is absorbed into $C_1$
(it depends only on the fixed level $\alpha$).
Summing over $t\le T$ in \eqref{eq:si-rb-telescope},
\begin{equation}
\Big|P\Big(\max_{t\le T}|Z_t|\le c\Big)-(1-\alpha)\Big|
\le C_1(1+c^3)\,\mathrm{lev}_n
\le C_2\,\mathrm{lev}_n\{1+(\log T)^{3/2}\}\to0 .
\label{eq:si-rb-oracle}
\end{equation}

\paragraph{3. Sub-exponential replacements for the Gaussian maximal bounds.}
In the proof of Theorem~\ref*{thm:path-main}, the only distributional inputs
were conditional Gaussian tail bounds for linear forms of held-out
innovations. Under (B1), every such linear form with
$\mathcal T_{tk}$-measurable weight vector $w$ satisfies the Bernstein
inequality
\[
P\big(|w^\trans\varepsilon|>\norm w_2\sqrt{2x}+C\bar b\norm w_\infty x
\,\big|\,\mathcal T_{tk}\big)\le2e^{-x},
\]
so all maximal bounds hold with $\sqrt{n\ell_n}$ replaced by
$\sqrt{n\ell_n}+\bar b\,\ell_n\lesssim\sqrt{n\ell_n}$ under $\ell_n=o(n)$.
The expansion (\ref*{eq:path-expansion-main}) and the variance consistency
\eqref{eq:si-band-var} therefore continue to hold, and
\[
\Delta_n:=\max_t|\widehat Z_t-Z_t|
=O_{\PP}^{\mathrm{unif}}\!\left(\sqrt n\,a_n^2+(a_n+\delta_n)\sqrt{\ell_n}\right).
\]

\paragraph{4. Ring bound and conclusion.}
Choose deterministic $\varepsilon_n\to0$ with
$\sup_PP(\Delta_n>\varepsilon_n)\to0$ and
$\varepsilon_n\sqrt{\ell_n}\to0$, possible by $b_n^{\mathrm{rob}}\to0$. On
$\{\Delta_n\le\varepsilon_n\}$, the feasible and oracle coverage events differ
only inside the ring $\{c-\varepsilon_n<\max_t|Z_t|\le c+\varepsilon_n\}$.
By the same telescoping argument, the ring probability is at most
\[
\sum_{t\le T}\operatorname*{ess\,sup}
P\big(c-\varepsilon_n<|Z_t|\le c+\varepsilon_n\mid\mathcal F_{n,t-1}\big)
\le T\left\{2\varepsilon_n\,\phi(c-\varepsilon_n)
+C_1(1+c^3)\,\mathrm{lev}_n\frac1T\right\},
\]
using the lemma at both radii. Since $T\phi(c)\le C(1+c^2)\,T\{1-\Phi(c)\}/c\le C'\sqrt{\log T}$
by the Mills-ratio bound and $T\{1-\Phi(c)\}\le-\log(1-\alpha)/2$, the first term is $O(\varepsilon_n\sqrt{\ell_n})\to0$
and the second repeats \eqref{eq:si-rb-oracle}. Combining with
\eqref{eq:si-rb-oracle} completes the proof.
\end{proof}

\subsection{A finite-sample conservative alternative}

\begin{proposition*}[Leverage-adjusted Bernstein band]
Assume (B1)--(B2) and oracle scores. For $L_T=\log(2T_n/\alpha)$ define
\[
c^{\mathrm B}_t=\sqrt{2L_T}+C_{\mathrm B}\bar b\,\norm{a_t}_\infty L_T ,
\]
with the absolute constant $C_{\mathrm B}$ from Bernstein's inequality for
$\psi_1$ variables. Then
\[
P\big(|Z_t|\le c^{\mathrm B}_t\ \text{for all}\ t\le T_n\big)\ge1-\alpha
\qquad\text{for every }n,\ \text{every }P .
\]
Under the negligibility condition of the theorem, the feasible version holds
with asymptotic level at least $1-\alpha$ uniformly.
\end{proposition*}

\begin{proof}
Bernstein's inequality for weighted sums of independent $\psi_1$ variables
gives $P(|Z_t|>c^{\mathrm B}_t\mid\mathcal F_{n,t-1})\le2e^{-L_T}=\alpha/T_n$
almost surely; see, e.g., \citet[Thm.~2.8.1]{Vershynin2018}. Taking
expectations and a union bound over $t$ proves the display. The feasible
statement follows from Step 3--4 of the theorem's proof with the ring bound
replaced by the same union bound at radius $c^{\mathrm B}_t\pm\varepsilon_n$.
\end{proof}

\begin{remark}[What is and is not being claimed]
The theorem keeps the \v{S}id\'{a}k critical value and shows its asymptotic
validity without Gaussianity; the price is the moderate-deviations condition
(B3), which fails when single coordinates dominate a date's score or when
$T_n$ grows almost exponentially in $n$. The proposition is exact in finite
samples but conservative, with width inflated by the leverage term. A
multiplier-bootstrap calibration in the regime of
\citet{ChernozhukovChetverikovKato2013} is a third option; we do not develop
it because the two results above already bracket the practical cases; the
simulation section reports where the plain \v{S}id\'{a}k band degrades
(small per-date information, skewed innovations), which locates the
boundary of (B3) empirically.
\end{remark}

\section{Identification-robust confidence sets}
\label{sec:si-weakid}
The Wald-type path inference of Sections~\ref*{sec:si-estimation} and
\ref*{sec:si-bands} requires the uniform information floor (E4) and the honest
pilot (E6). Both fail precisely where the identification analysis says the
problem is hard: $r_t$ near zero (strength weakly identified) or
$\sigma_{\min}(Q_t)$ near zero (composition weakly identified). This section
provides confidence sets that remain valid in those regimes. The construction
is the score (Anderson--Rubin-type) inversion: at a hypothesized target value
the realized score is exactly centered and exactly variance-standardized,
regardless of how weak the information is, so no estimator and no pilot is
required.

\subsection{The statistic}

Fix a date $t$ and a hypothesized $\theta^0=(b,h)$. With oracle covariances,
define
\[
S_t(\theta^0)=\sum_kJ_{tk}(\theta^0)^\trans e_{tk}(\theta^0),
\qquad
\mathcal I_t(\theta^0)=\sum_kJ_{tk}(\theta^0)^\trans J_{tk}(\theta^0),
\]
exactly as in (\ref*{eq:oracle-score-information}) but evaluated at $\theta^0$,
and the statistic
\[
\mathrm{AR}_t(\theta^0)
=S_t(\theta^0)^\trans\,\mathcal I_t(\theta^0)^{\dagger}\,S_t(\theta^0),
\]
with $\dagger$ the Moore--Penrose inverse. The feasible version
$\widehat{\mathrm{AR}}_t$ replaces oracle covariances by the training
estimates of Section~\ref*{sec:estimation} and adds a ridge:
$\widehat{\mathcal I}_t(\theta^0)^{\dagger}$ is replaced by
$\{\widehat{\mathcal I}_t(\theta^0)+\varrho_n nI_d\}^{-1}$ for a deterministic
$\varrho_n\ge0$.

\subsection{Exact validity without identification}

\begin{theorem*}[Theorem~\ref*{thm:weakid-main}, restated]
Assume the conditional Gaussian model, (E1)--(E3), and oracle covariances.
Fix $\alpha\in(0,1)$ and let $\chi^2_{d,1-\alpha}$ be the $\chi^2_d$ quantile.

\begin{enumerate}[label=\textup{(\alph*)}]
\item At the true value, for every $n$, every $t$, and every
$P\in\mathcal P_n$, including $P$ under which $\mathcal I_t$ is singular or
nearly singular,
\[
P\{\mathrm{AR}_t(\theta_t)>\chi^2_{d,1-\alpha}\mid\mathcal F_{t-1}\}\le\alpha
\quad\text{a.s.},
\]
with equality if and only if $\rank\{\mathcal I_t(\theta_t)\}=d$. The same
holds for the ridge version for every $\varrho_n\ge0$.

\item Under the sequential condition of Theorem~\ref*{thm:bands-main} (with
Gaussianity), the sets
$\mathcal C_t^{\mathrm{AR}}
=\{\theta^0:\mathrm{AR}_t(\theta^0)\le\chi^2_{d,1-\alpha_T}\}$ with
$1-\alpha_T=(1-\alpha)^{1/T_n}$ satisfy the exact simultaneous guarantee
\[
P\{\theta_t\in\mathcal C_t^{\mathrm{AR}}\ \text{for all}\ t\le T_n\}
\ge1-\alpha ,
\]
with equality under full rank at every date.

\item (Projection.) For any subvector or functional $\psi(\theta)$, the set
$\psi(\mathcal C_t^{\mathrm{AR}})$ covers $\psi(\theta_t)$ with at least the
same probability. In particular
$\{b:\min_h\mathrm{AR}_t(b,h)\le\chi^2_{d,1-\alpha}\}$ is a level-$(1-\alpha)$
confidence set for $\beta_t$ that requires neither (E4) nor (E6).

\item (Feasible version.) Assume in addition (E5),
$\varrho_n\ge C(\delta_n+\varrho_n')$ for the Jacobian perturbation rate
$\varrho_n'$ of Section~\ref*{sec:si-estimation}, and
$\delta_n^2\ell_n/\varrho_n\to0$. Then
\[
\limsup_n\ \sup_{P\in\mathcal P_n}\ \max_{t\le T_n}
\left[P\{\widehat{\mathrm{AR}}_t(\theta_t)>\chi^2_{d,1-\alpha}\}-\alpha\right]\le0 ,
\]
and when additionally $\liminf_n\lambda_{\min}(\mathcal I_t)/n>0$ and
$\varrho_n\to0$, the asymptotic size is exactly $\alpha$.
\end{enumerate}
\end{theorem*}

\begin{proof}
(a) At $\theta_t$, residualization annihilates the mean nuisances exactly, so
conditionally on $\mathcal F_{t-1}$ (and the training fields),
$e_{tk}(\theta_t)=R_{tk}\varepsilon_{tk}$ with
$\Var\{e_{tk}(\theta_t)\}=R_{tk}\,\mathcal V_{tk}R_{tk}^\trans$ an orthogonal
projection in the whitened coordinates; hence
$S_t(\theta_t)\mid\mathcal F_{t-1}\sim N_d\{0,\mathcal I_t(\theta_t)\}$
exactly, for any rank. Write
$S_t=\mathcal I_t^{1/2}\zeta$ with $\zeta\sim N_d(0,P)$, $P$ the projector
onto $\mathrm{range}(\mathcal I_t)$. Then
$\mathrm{AR}_t(\theta_t)=\zeta^\trans P\zeta\sim\chi^2_{\rank(\mathcal I_t)}$,
which is stochastically dominated by $\chi^2_d$, with equality of laws iff
the rank is $d$. For the ridge version, with spectral decomposition
$\mathcal I_t=\sum_j\lambda_ju_ju_j^\trans$,
\[
S_t^\trans(\mathcal I_t+\varrho nI)^{-1}S_t
=\sum_j\frac{\lambda_j}{\lambda_j+\varrho n}\,(u_j^\trans\zeta)^2
\le\zeta^\trans P\zeta ,
\]
a pointwise bound, so domination is preserved.

(b) Under the sequential condition, the conditional law in (a) is
$\mathcal F_{n,t-1}$-measurable and the standardized quadratic form is
conditionally pivotal; the backward-induction identity
\eqref{eq:si-rb-telescope} with per-date conditional coverage exactly
$\ge1-\alpha_T$ (equality under full rank) yields the product bound
$(1-\alpha_T)^{T_n}=1-\alpha$.

(c) Immediate: $\theta_t\in\mathcal C_t^{\mathrm{AR}}$ implies
$\psi(\theta_t)\in\psi(\mathcal C_t^{\mathrm{AR}})$.

(d) On the covariance-good event of (E5), the perturbation bounds
\eqref{eq:si-residualizer-perturbation}--\eqref{eq:si-jacobian-perturbation}
evaluated at the fixed point $\theta_t$ (no pilot enters: the statistic is
evaluated at the hypothesized value) give
$\widehat S_t(\theta_t)=S_t(\theta_t)+\Delta_S$ and
$\widehat{\mathcal I}_t(\theta_t)=\mathcal I_t(\theta_t)+\Delta_I$ with
$\norm{\Delta_S}=O_{\PP}^{\mathrm{unif}}(\delta_n\sqrt{n\ell_n})$,
$\norm{\Delta_I}_{\op}=O_{\PP}^{\mathrm{unif}}(\delta_nn)$. For the ridge
inverse with $\varrho_nn$ dominating $\norm{\Delta_I}_{\op}$,
\[
\widehat{\mathrm{AR}}_t(\theta_t)
\le(1+o_{\PP}(1))\,S_t^\trans(\mathcal I_t+\tfrac12\varrho_nnI)^{-1}S_t
+o_{\PP}(1)
\le(1+o_{\PP}(1))\,\chi^2\text{-dominated}+o_{\PP}(1),
\]
where the first inequality uses
$(x+\delta)^\trans A(x+\delta)\le(1+\epsilon)x^\trans Ax+(1+\epsilon^{-1})\delta^\trans A\delta$
and $\delta^\trans A\delta\le\norm{\Delta_S}^2/(\varrho_nn)
=O_{\PP}(\delta_n^2\ell_n/\varrho_n)=o_{\PP}(1)$ by the assumed rate
condition. Conservative asymptotic size follows from (a) and Slutsky
uniformly over $\mathcal P_n$; exact size under an information floor follows
because then $\varrho_n\to0$ makes the ridge negligible and the dominated
variable is full-rank $\chi^2_d$.
\end{proof}

\begin{remark}[Power and use]
Under strong identification and local alternatives
$\theta^0=\theta_t+\mathcal I_t^{-1/2}\mu$, $\mathrm{AR}_t$ is noncentral
$\chi^2_d(\norm\mu^2)$: the score set incurs a fixed $d$-degrees-of-freedom
widening relative to the Wald interval but never undercovers. The intended
use is exactly the one the identification theorem motivates: report Wald
bands when the diagnostic $\lambda_{\min}(\widehat{\mathcal I}_t)/n$ clears a
prespecified floor, and the projected score set
$\{b:\min_h\widehat{\mathrm{AR}}_t(b,h)\le\chi^2_{d,1-\alpha}\}$ at all dates
where it does not. The profiled minimum over $h$ is a $q$-dimensional smooth
minimization; incomplete minimization overstates the minimum, shrinks the
projected set, and can therefore undercover, so implementations must verify
convergence of the inner minimization or refine it on a grid. The simulation section quantifies both the Wald size distortion under
weak reports and the cost in length of the score set under strong
identification.
\end{remark}

\section{Pseudo-true targets and bounded-bias sensitivity}
\label{sec:si-pseudo}
This section formalizes what the estimator targets when the two declared
modeling commitments fail in bounded ways: the chart may not span the true
composition perturbation (off-chart error), and a common dyad bias of
bounded size may contaminate both mirror reports. Both failures act on the
report-channel mean and propagate to the target through one fixed linear
map, so pseudo-true targets, sensitivity intervals, and breakdown values are
exact linear algebra. Throughout, the chart is linear ($\dot m=\Psi$),
covariances are the working ones of Section~\ref*{sec:model}, and we write
$R_z=M_{L_zU}L_z$, $Q=R_zA\Psi$, and $\mathcal I_c$ for the joint
information \eqref{eq:joint-information-main} at the working point.

\subsection{Pseudo-true targets under off-chart composition}

Suppose the true dyad log-composition is $m^\dagger=\Psi\eta^\dagger+w$ with
$C^\trans w=0$ and $w$ outside the chart span. The report-channel population
score at $\eta$ is $Q^\trans\{R_zA(m^\dagger)-Q\eta\}$, and the outcome
channel is evaluated along the true network path.

\begin{proposition*}[Information-projection target]
Assume $\norm w$ small enough that the population joint score has a root in
the $r_0$-ball of Theorem~\ref*{thm:identification-main}. Then the root
$\theta^*=(\beta^*,\eta^*)$ satisfies
\[
\theta^*=\theta^\dagger
+\mathcal I_c^{-1}
\begin{pmatrix}0\\ Q^\trans R_zAw\end{pmatrix}
+O(\norm w^2),
\]
with the remainder constant controlled by the softmax Hessian bound of
Section~\ref*{sec:si-instance}; for the report-only subproblem the formula
is exact: $\eta^*_{\mathrm{rep}}=\eta^\dagger+(Q^\trans Q)^{-1}Q^\trans
R_zAw$, the information projection of the true composition onto the chart.
Read conditions (E1)--(E6) at the pseudo-true path (the $r_0$-ball, the
information floor, and the pilot separation taken around $\theta^*_t$).
Then, with $\theta^*_t$ as the target:
(i) the score-inversion statements of Section~\ref*{sec:si-weakid} hold
\emph{exactly}: $S_t(\theta^*_t)$ is exactly conditionally centered with
conditional covariance $\mathcal I_t(\theta^*_t)$, so
$\mathrm{AR}_t(\theta^*_t)$ retains its exact $\chi^2$ domination at every
$n$;
(ii) the Wald-type statements of
Sections~\ref*{sec:si-estimation}--\ref*{sec:si-bands} (path expansion,
studentized intervals, simultaneous bands) hold with coverage and level
errors of order $O(\norm w)$, uniformly over dates, in addition to their
stated asymptotic errors; the error is not asymptotically negligible at
fixed $w$ and vanishes as $\norm w\to0$.
The declared-chart analysis estimates its chart's best information-metric
approximation; its Wald uncertainty statements for that object are correct
up to $O(\norm w)$, and its score-inversion statements are exact.
\end{proposition*}

\begin{proof}
Existence and the expansion are the implicit function theorem applied to the
population score map $\theta\mapsto E S(\theta)$, whose value at
$\theta^\dagger$ is $(0;\,Q^\trans R_zAw)$ (the outcome channel is
centered at the truth and the report channel's residualized mean error is
$R_zAw$) and whose derivative at $\theta^\dagger$ is
$-\mathcal I_c+O(\norm w)$: the exact derivative is
$-J^\trans J+(D_\theta J)^\trans E e$, the first term is
$-\mathcal I_c$ by Theorem~\ref*{thm:identification-main}, and the second
is $O(\norm w)$ because $Ee(\theta^\dagger)=(0;R_zAw)$; the correction
contributes only to the stated quadratic remainder. Second derivatives of
the score in $\theta$ are bounded by the constants of
Section~\ref*{sec:si-instance}, giving the quadratic remainder. For the
report-only subproblem the score is exactly linear in $\eta$, so the root
is exact.

(i) The pseudo-true path is defined date by date as the root of the
conditional population score, so
$E\{S_t(\theta^*_t)\mid\mathcal F_{t-1}\}=0$ exactly; the innovations are
unchanged, so $\Var\{S_t(\theta^*_t)\mid\cdot\}$ is the Gram matrix
$\mathcal I_t(\theta^*_t)$ of the same whitened residualized derivatives;
the proofs in Section~\ref*{sec:si-weakid} use only this conditional
centering and covariance, so they apply unchanged.

(ii) The Wald statements additionally use the linearization of the score
between the pilot and the target. At $\theta^*_t$ the residual
$e_t(\theta^*_t)$ has a deterministic conditional-mean component
$r_w=R\{\mu^\dagger-\mu(\theta^*_t)\}$ of norm $O(\sqrt n\norm w)$ which is
orthogonal to the score columns but not to their derivatives: the effective
curvature of the one-step is
$J^\trans J-(D_\theta J)^\trans e
=\mathcal I_t\{1+O(\norm w)\}+(D_\theta J)^\trans R\varepsilon$,
where the stochastic term is controlled exactly as in the exact-model proof
and the deterministic perturbation contributes a relative $O(\norm w)$ to
the linearization and hence to the studentized variance. Coverage of a
nominal-$(1-\alpha)$ Wald interval is therefore $1-\alpha+O(\norm w)$; no
other step of the band proof is affected, since the leading score term is
exactly centered with unchanged covariance by (i).
\end{proof}

\begin{remark}[Specification diagnostic]
The whitened report residual after projecting off $[L_zU,\ L_zA\Psi]$ has,
under the working covariances, an exact quadratic-form law with
noncentrality $\norm{(I-P_{[U,A\Psi]})L_zAw}^2$: zero on-chart, positive
off-chart. This is the report-channel analogue of the reporter-cycle test
and is reported alongside it in the case study; its noncentrality gives the
power curve against off-chart alternatives directly.
\end{remark}

\subsection{Sensitivity regions under bounded common dyad bias}

Theorem~\ref*{thm:identification-main}(d) shows an \emph{unrestricted}
common dyad bias $c$ (the same in both mirror reports) destroys
identification. The empirically relevant question is bounded failure:
$\norm c_\infty\le\delta$ for a declared $\delta$. Let
$c_2=(c;c)$ denote the stacked bias and define the \emph{passthrough
matrix} $\Lambda\in\R^{d\times|\mathcal E|}$ columnwise by
\[
\Lambda_{\cdot e}=\mathcal I_c^{-1}
\begin{pmatrix}0\\ Q^\trans R_z\,(\mathbf e_e;\mathbf e_e)\end{pmatrix},
\qquad e\in\mathcal E,
\]
so that the first-order target shift under bias $c$ is $\Lambda c$.

\begin{proposition*}[Sensitivity intervals and breakdown]
Under $\norm c_\infty\le\delta$ and the conditions above, the pseudo-true
target satisfies, coordinate by coordinate and to first order in $\delta$
(exactly for the report-channel subproblem),
\[
\theta^*_l\in
\big[\theta^\dagger_l-\delta\norm{\Lambda_{l\cdot}}_1,\
\theta^\dagger_l+\delta\norm{\Lambda_{l\cdot}}_1\big],
\qquad l=1,\ldots,d,
\]
and the bounds are attained at sign-pattern extreme points of the box. Any
conclusion drawn from a contrast $a^\trans\theta$ is robust to common bias
of size $\delta$ whenever the estimated contrast exceeds its critical value
by more than $\delta\norm{a^\trans\Lambda}_1$ (plus the window-doubling
factor when the contrast compares pre- and post-change windows and the bias
may differ across them); the \emph{breakdown value} $\delta^*$ is the
smallest $\delta$ at which this margin is exhausted. As
$\delta\to\infty$ the intervals cover the whole parameter space in the
directions $\mathrm{range}(\Lambda)$, recovering the impossibility theorem
as the limit of the sensitivity analysis.
\end{proposition*}

\begin{proof}
By the same implicit-function expansion, $\theta^*(c)-\theta^\dagger=\Lambda
c+O(\norm c^2)$, linear in $c$. A linear functional of a box
$\{\norm c_\infty\le\delta\}$ has range $\pm\delta$ times the $\ell_1$ norm
of its coefficient vector, attained at vertices. The contrast and breakdown
statements are the same computation applied to $a^\trans\Lambda$; when the
pre- and post-window biases may take different values in the box, the
worst-case contrast shift doubles. The limit statement follows because the
first-order image of the growing box is $\mathrm{range}(\Lambda)$, and
this range contains every composition direction: by
Theorem~\ref*{thm:identification-main}(d), for every $h\in\R^q$ there are
$(\lambda,c)$ with $L_zA\Psi h=L_zU\lambda+L_z(c;c)$; applying
$M_{L_zU}$ gives $Q\Psi\text{-direction }=R_z(c;c)$ projected, hence
$Q^\trans Q\,h=Q^\trans R_z(c;c)$, and since $Q^\trans Q\succ0$ the map
$c\mapsto Q^\trans R_z(c;c)$ is surjective onto $\R^q$.
\end{proof}

\begin{remark}[Use and verification]
The passthrough matrix is computable from the design before outcomes are
seen, like every other diagnostic in this paper. The replication code
verifies the expansion against direct recomputation of the population root
(agreement to $10^{-10}$ at machine-checkable instances) and verifies the
MC centering of the estimator on $\eta^*$ under an off-chart perturbation.
The Wald coverage statement of the information-projection proposition is
measured with controls. At the suite's off-chart scale the
nominal-$95\%$ Wald ellipse for $\eta^*$ covers in $\senCovFull\%$ of
replications; the zero-misspecification control at the same design covers
in $\senCovZero\%$ ($R=\senCovR$ per row, MCSE $\approx\senCovMCSEv$
points). The gap is therefore the finite-sample Wald error of the
uncalibrated one-date fit at this deliberately small design, the
phenomenon the calibration of Section~\ref*{sec:si-estimation} and the
simulation ladder address, and not a misspecification effect: across an
eightfold range of scales coverage is
$\senCovQuarter$, $\senCovHalf$, $\senCovFull$, and $\senCovDouble\%$,
flat within Monte Carlo error at the smaller scales and rising at the
largest, the conservative direction (off-chart signal inflates the
estimated score covariance). Score-inversion sets are exact at every
scale by part (i). The case study reports $\delta^*$ for its headline attribution:
the size of common reporting bias that would have to be present before the
composition-only classification could be an artifact.
\end{remark}

\section{Score-inversion change inference on the observational path}
\label{sec:si-obsdetect}
The change guarantees of Section~\ref*{sec:si-change} hold in a replicated
Gaussian benchmark; the transfer instance of Section~\ref*{sec:si-instance}
is a designed experiment. This section gives change inference directly in
the observational experiment (realized lags, one dependent path, no
pilot, no information floor) by inverting the identification-robust score
sets of Section~\ref*{sec:si-weakid}. The construction trades power for
validity: it is exactly sized under the conditions of that section, its
confidence statements follow from inversion without asymptotic
approximation, and its power is evaluated by simulation and is not matched
to the benchmark lower bounds. Observational attribution optimality
(procedures attaining benchmark-order thresholds for a single dependent
path) remains open, and we state it as such.

\subsection{The score-inversion construction}

For a hypothesized common value $\theta^0$ over a horizon
$\mathcal T=\{1,\ldots,T\}$, each date supplies the statistic
$\mathrm{AR}_t(\theta^0)$ of Section~\ref*{sec:si-weakid}. Let
$\chi^2_{d,\pi_T}$ with $\pi_T=(1-\alpha)^{1/T}$ be the per-date critical
value.

\begin{description}[style=nextline]
\item[Constancy test.] Reject ``the target path is constant on
$\mathcal T$'' if and only if
\[
\min_{\theta^0}\ \max_{t\le T}\ \mathrm{AR}_t(\theta^0)\ >\ \chi^2_{d,\pi_T}.
\]
\item[Split confidence set.] For a candidate change date $s$, accept $s$ if
both segments $\{t\le s\}$ and $\{t>s\}$ pass their own constancy tests at
level determined by the same per-date $\chi^2_{d,\pi_T}$. The accepted set
$\widehat{\mathcal S}$ is a confidence set for the true change date.
\item[Attribution regions.] At each accepted split $s$, the set of
accepted pairs $(\theta^{(0)},\theta^{(1)})$ induces a confidence region
for the jump $\Delta=\theta^{(1)}-\theta^{(0)}$; the region reported for
the change is the union over accepted splits. A change is \emph{consistent
with composition-only} if some accepted pair at some accepted split has
$\Delta_\beta=0$, i.e.\ if the constrained minimax problem with a common
$\beta$ and free segment compositions is feasible at the critical value;
it is \emph{certified inconsistent} only when no accepted split admits
such a pair.
\end{description}

\begin{theorem*}[Theorem~\ref*{thm:obsdetect-main}, restated]
Assume the conditional Gaussian model, (E1)--(E3), oracle covariances, and
the sequential condition of Theorem~\ref*{thm:bands-main}.
\begin{enumerate}[label=\textup{(\alph*)}]
\item Under any constant path $\theta_t\equiv\theta^\star$, the constancy
test has size at most $\alpha$, exactly: with probability at least
$1-\alpha$, $\theta^\star$ itself satisfies
$\max_t\mathrm{AR}_t(\theta^\star)\le\chi^2_{d,\pi_T}$, so the minimum
cannot exceed the critical value.
\item Under a single change at $\tau$ with segment values
$(\theta^{(0)},\theta^{(1)})$, with probability at least $1-\alpha$ the true
split is accepted, $\tau\in\widehat{\mathcal S}$, and the true jump lies in
the attribution region. Consequently the composition-only classification
errs (rejects a true composition-only change, or fails to flag that a
strength change is inconsistent with $\Delta_\beta=0$ when every accepted
pair excludes it) with probability at most $\alpha$.
\item With estimated covariances and the ridge of
Theorem~\ref*{thm:weakid-main}(d), the same statements hold with asymptotic
size and coverage, uniformly over $\mathcal P_n$, without (E4) or (E6).
\end{enumerate}
\end{theorem*}

\begin{proof}
(a) By Theorem~\ref*{thm:weakid-main}(b), the event
$\mathcal A=\{\mathrm{AR}_t(\theta^\star)\le\chi^2_{d,\pi_T}\ \forall
t\le T\}$ has probability at least $(1-\alpha_T)^T=1-\alpha$ exactly, by the
backward-induction product argument with the exact conditional
$\chi^2$-domination at every date. On $\mathcal A$ the minimax value is at
most the critical value and the test accepts. (b) Apply the same argument
segmentwise at the true $(\tau,\theta^{(0)},\theta^{(1)})$: the $T$ per-date
events again have joint probability at least $1-\alpha$ (the product
telescopes across the two segments in date order because the sequential
conditioning does not care where the segment boundary lies). On that event
the true split, the true pair, and hence the true jump are accepted, which
is simultaneously the coverage of $\widehat{\mathcal S}$, of the attribution
region, and the error bound for the classification rule. (c) follows by
replacing exact conditional domination with the asymptotic domination of
Theorem~\ref*{thm:weakid-main}(d) uniformly over dates, exactly as there.
\end{proof}

\begin{remark}[Structure of the accepted-split set]
For fixed data, the minimax value
$v(\mathcal D)=\min_{\theta}\max_{t\in\mathcal D}\mathrm{AR}_t(\theta)$ is
nondecreasing in the date set $\mathcal D$: for every $\theta$ the inner
maximum over a superset is at least the maximum over a subset, and taking
the minimum over $\theta$ preserves the inequality. Prefix feasibility
$\{v(\{1,\ldots,s-1\})\le\chi^2_{d,\pi_T}\}$ is therefore nonincreasing in
$s$ and suffix feasibility nondecreasing, so the accepted-split set is an
interval whenever it is nonempty. An implementation that scans every
split needs only to certify the interval's two endpoints; the replication
code does exactly this, warm-starting each segment's search from the
neighboring segment's feasible point.
\end{remark}

\begin{remark}[Computation, three-state verdicts, and power]
Acceptance requires exhibiting a feasible $\theta^0$; rejection is
licensed only when a candidate search (pooled and segmentwise starting
values, an information-scaled grid of $\pm3$ standardized units around
them, and smoothed Gauss--Newton refinement from every start) fails to
exhibit one. Rejection is therefore a certified-search statement, not a
global optimality certificate; its error rate under the null is part of
the realized size measured in Section~\ref*{sec:sim-obsdetect}. A feasible
point exists in the true value's neighborhood with probability $1-\alpha$,
and the search explores that neighborhood at resolution finer than the
acceptance region, so the implemented size tracks the nominal level, and
search failure can move the constancy test only toward rejection. The
attribution step uses the same discipline symmetrically: the constrained
(common-$\beta$) search is multistart, is run at every accepted split, and
the verdict takes three states: consistent with composition-only (a
feasible constrained pair was exhibited at some accepted split),
inconsistent by search (the multistart search was exhausted at every
accepted split), or undetermined (the accepted-split set is empty although
constancy was rejected). Undetermined outcomes are reported, never
converted into certificates. Power, the cardinality of
$\widehat{\mathcal S}$, and the verdict rates are evaluated by simulation
in the observational design itself (Section~\ref*{sec:sim-obsdetect}); no
optimality claim is made, and the benchmark lower bounds of
Section~\ref*{sec:si-minimax} bound the possible improvement. Closing that
gap for dependent observational paths is the main open problem this paper
defines.
\end{remark}

\section{The change benchmark against its own constants}
\label{sec:si-simchange}
This section evaluates the three-copy Gaussian benchmark procedure against
its own constants. It is a benchmark evaluation, not evidence about the
observational experiment; it is placed in the Supplement for that reason,
and its role is to measure how conservative the finite-sample guarantee of
Theorem~\ref*{thm:change-main} is and where the minimax obstructions sit
relative to practice.

Figure~\ref{fig:change-benchmark} runs the exact three-copy procedure of
Section~\ref*{sec:si-change} ($d=3$, $T=\chgT$, $h=\chgH$, two changes,
$R=\chgR$ per jump size). The procedure has a sharp empirical phase
transition: detection with correct attribution moves from failure to
$95\%$ success near $\varsigma\approx\chgAttKninefive$, about half the
finite-sample guarantee threshold $\varsigma\ge\chgKappaUB$ of
Theorem~\ref*{thm:change-main}, so the guarantee is conservative by
roughly a factor of two at this design; the attribution lower-bound scale
$\varsigma\approx\chgKappaAttLB$ lies an order of magnitude lower. These two
values bracket the performance attainable by any observational procedure,
including Theorem~\ref*{thm:obsdetect-main}. Localization is sharp once
detection holds: at $\varsigma=5$ the median (90th percentile) refined error
is $\chgLocMedAtFive$ ($\chgLocQnineAtFive$) dates against a guaranteed
bound of $\chgLocBoundAtFive$. Under $K=0$ the screening stage
false-alarms in $\chgNullFP\%$ of replications at $\alpha=0.05$.

\begin{figure}[t]
\centering
\includegraphics[width=0.55\textwidth]{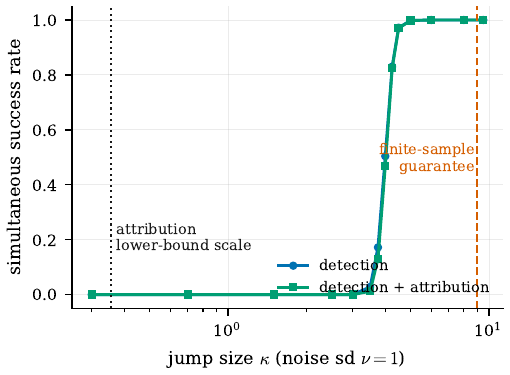}
\caption{Three-copy Gaussian benchmark: simultaneous success rates across
jump sizes, with the finite-sample guarantee threshold (dashed) and the
minimax attribution scale (dotted).}
\label{fig:change-benchmark}
\end{figure}

\section{Synthetic end-to-end validation of the deployment pipeline}
\label{sec:si-vignette}
This section validates the full deployment pipeline end to end under
known truth, complementing the real application of
Section~\ref*{sec:application} of the main text with the steps whose
evaluation requires a verifiable target (attribution labeling and
forecasting). Bilateral trade is the canonical instance of the paper's
measurement model:
the same flow from exporter $j$ to importer $i$ is reported twice, free on
board by the exporter and cost-insurance-freight by the importer, with
systematic reporter-side wedges, correlated errors, missing returns, and
well-documented two-sided discrepancies
\citep{Bhagwati1964,FismanWei2004,JavorcikNarciso2008}. The paired design
of Theorem~\ref*{thm:identification-main}(d) was written for exactly this
structure.

The panel is simulated from the declared experiment, so ground truth is
known and every step of the protocol has a verifiable target; it is
labeled synthetic in every figure and table and is evidence about the
pipeline, not about trade.

\subsection{Design and prespecified diagnostics}

The panel has $N=\appN$ economies over $T=\appT$ quarters with
$n_y=\appny$ within-quarter outcome observations, complete eligible
support, and a $q=3$ gravity chart (distance decay, two-bloc affinity,
agreement pair, row-centered). True strength is constant at
$\beta_t=\appBeta$; the distance elasticity drifts slowly; bloc affinity
drops at quarter $\appTau$, a reallocation of the type produced by a trade
conflict. Importer reports carry a CIF-type wedge, reporter biases follow
the declared receiver/supplier design, mirror errors are correlated, and
$\appMissPct\%$ of pairs are missing at random each quarter under the
prespecified ignorable rule; all descriptive statistics, plug-in
comparators, and the estimator itself use only available pairs.

Diagnostics come first, as the protocol requires. Mirror discrepancies on
available pairs average $\appDiscMean$ log points (dispersion
$\appDiscSD$), consistent with the declared additive wedge. The
reporter-cycle test, run quarterly with the declared covariance, is
consistent with the additive receiver/supplier design in
$\appCyclePassPct\%$ of quarters at the $5\%$ level
(Figure~\ref{fig:application}(d)); by
Theorem~\ref*{thm:identification-main}(d), passing certifies only the
absence of non-additive discrepancy structure; the common-bias direction
is invisible to this test, so its effect is quantified by the sensitivity
analysis below. The identification diagnostic
$\{\lambda_{\min}(\widehat{\mathcal I}_t)/n\}^{1/2}$ averages
$\appSigMinMean$ (minimum $\appSigMinMin$), clearing the prespecified floor
at every quarter, so Wald-type path inference is licensed and the score
sets of Theorem~\ref*{thm:weakid-main} remain available.

\subsection{Paths, verdicts, attribution, sensitivity, forecasting}

Figure~\ref{fig:application}(a)--(b) shows the estimates from the
calibrated estimator ($\widehat\gamma=\appGamma$). The joint
$\widehat\beta_t$ tracks the constant truth; its $95\%$ simultaneous band
covers the full constant path, and the band-based constancy verdict, which
rejects only if no horizontal line fits inside the band, correctly does
not reject ($\appConstReject$ rejections). The static plug-in on the same
data moves from $\appPlugPre$ to $\appPlugPost$ across the composition
break, a standardized spurious shift of $\appPlugZ$ under the naive rule
and the single-dataset counterpart of Theorem~\ref*{thm:plugin-main}. The
report-only two-step also stays centered (its disadvantage is precision,
not bias, as in Section~\ref*{sec:sim-false-attr}). The composition path
recovers all three chart coordinates including the bloc-affinity drop.

Attribution is then reported descriptively and stress-tested. The
eight-quarter window contrast separates coordinates by an order of
magnitude: bloc affinity moves $\appJumpZetaTwo$ standard errors against
$\appJumpZbeta$ (strength), $\appJumpZetaOne$ (distance, whose true path
drifts), and $\appJumpZetaThree$ (agreement); under the
$6\varepsilon$-type labeling rule only the bloc coordinate is active, so
the episode is classified composition-only, in agreement with the truth.
Two caveats are quantified. First, the classification is descriptive;
exact observational change inference is available via
Theorem~\ref*{thm:obsdetect-main} and its simulation profile in
Section~\ref*{sec:sim-obsdetect}. Second, the common-bias sensitivity of
Proposition~\ref*{prop:pseudo-main}: the design-computed passthrough gives
per-unit-$\delta$ sensitivities of $\appSensBeta$ (strength) and
$\appSensEtaTwo$ (bloc affinity), so the composition-only classification
survives any common dyad bias up to the breakdown value
$\delta^*=\appDeltaStar$ on the log-report scale, a number that can be
compared with institutional knowledge of reporting practices.

The safe-inverse fallback never triggered ($\appSafePct\%$ of quarters).
A pseudo-out-of-sample forecast comparison is part of the prespecified
deployment protocol below; on this synthetic panel it is not reported as
evidence, since forecast rankings among estimators that share the
generating model carry no information about practical value. The
information grid of Section~\ref*{sec:sim-jointgain} is the controlled
comparison of the joint estimator with the two-step alternative.

\begin{figure}[t]
\centering
\includegraphics[width=\textwidth]{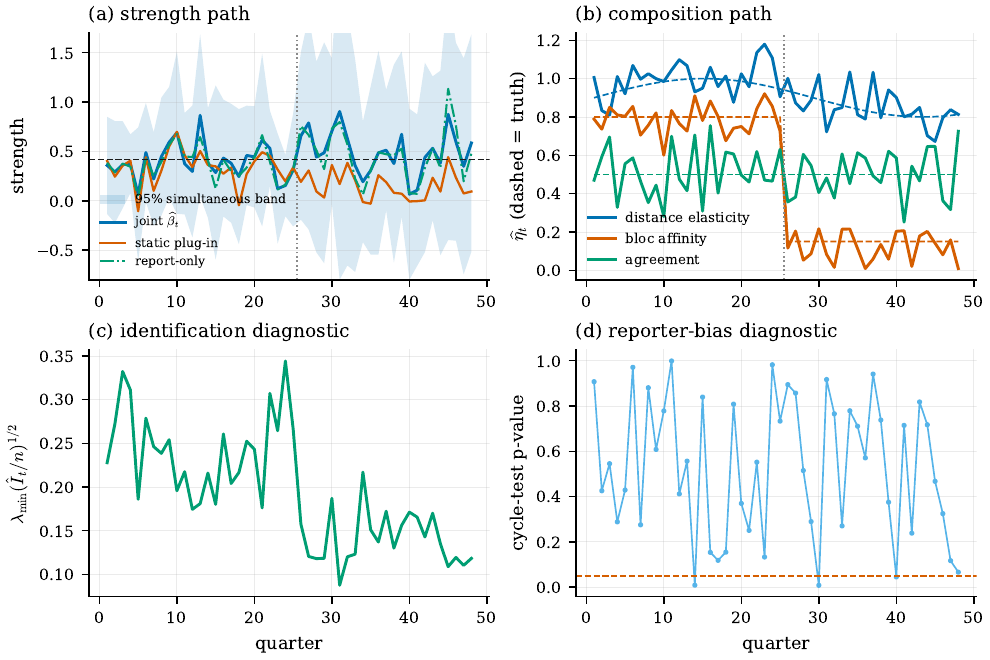}
\caption{Synthetic vignette, mirror-trade design (composition break at the
dotted line). (a) Strength: calibrated joint estimate with $95\%$
simultaneous band, static plug-in, and report-only two-step. (b)
Composition coordinates (dashed: truth). (c) Identification diagnostic.
(d) Reporter-cycle specification test.}
\label{fig:application}
\end{figure}

\section{Complete Gaussian change procedure and proof}
\label{sec:si-change}
\begin{lemma*}[Gaussian positive-drift bound]
Let $Y_1,\ldots,Y_N$ be iid $N(m,v^2)$ with $m,v>0$. For $1\le r\le N$,
\[
P\left\{\min_{r\le\ell\le N}\sum_{i=1}^{\ell}Y_i\le0\right\}
\le2\exp\left(-\frac{m^2r}{8v^2}\right).
\]
\end{lemma*}

\begin{proof}
Let $S_\ell=\sum_{i=1}^\ell Y_i$ and $B_r=\{S_r\le rm/2\}$. Gaussian tails give $P(B_r)\le\exp\{-m^2r/(8v^2)\}$. For $n\le N-r$, put $T_n=\sum_{i=r+1}^{r+n}Y_i$ and $M_n=\exp\{-2mT_n/v^2\}$. The normal moment-generating function shows that $(M_n)$ is a nonnegative mean-one martingale. On $B_r^c$, if some $S_{r+n}\le0$, then $T_n<-rm/2$ and $M_n>\exp(m^2r/v^2)$. Ville's inequality bounds this event by $\exp(-m^2r/v^2)$. Adding the two probabilities proves the claim.
\end{proof}

\subsection{Procedure}

Fix $d,h,T$ with $T\ge4h$, a known $\nu>0$, and piecewise-constant means $\mu_t$ with change points $0<\tau_1<\cdots<\tau_K<T$. Observe three independent copies
\begin{equation}
X_t^{(s)}=\mu_t+\xi_t^{(s)},
\qquad
\xi_t^{(s)}\stackrel{\mathrm{iid}}\sim N_d(0,\nu^2I_d),
\qquad s=1,2,3.
\label{eq:si-three-copy}
\end{equation}
Let $\mathcal I_h=\{2h,\ldots,T-2h\}$, $M_h=T-4h+1$, and
\[
C_h(k)=\sqrt{h/2}\left\{
h^{-1}\sum_{t=k+1}^{k+h}X_t^{(1)}
-h^{-1}\sum_{t=k-h+1}^{k}X_t^{(1)}
\right\}.
\]
With
\[
\lambda_\alpha=\nu\left\{\sqrt d+\sqrt{2\log(3M_h/\alpha)}\right\},
\]
retain $\mathcal A_h=\{k:\norm{C_h(k)}>2\lambda_\alpha\}$ and connect retained indices within distance $2h$. Order the connected components by their smallest elements. Each component gives one preliminary change $q_j$, chosen as the smallest index attaining the component maximum of $\norm{C_h(k)}$.

For each preliminary location define
\[
L_j=\{q_j-2h+1,\ldots,q_j-h\},
\qquad
R_j=\{q_j+h+1,\ldots,q_j+2h\}.
\]
Using copy 2, set $\widehat\mu_{j,-}=h^{-1}\sum_{t\in L_j}X_t^{(2)}$, $\widehat\mu_{j,+}=h^{-1}\sum_{t\in R_j}X_t^{(2)}$, and $\widehat\delta_j=\widehat\mu_{j,+}-\widehat\mu_{j,-}$. For a fixed orthogonal decomposition $\R^d=\mathcal V_{\rm str}\oplus\mathcal V_{\rm cmp}$, label block $b$ active when
\[
\norm{P_b\widehat\delta_j}>3\varepsilon_{\alpha,\widehat K},
\quad
\varepsilon_{\alpha,m}=\frac\nu{\sqrt h}
\left\{\sqrt d+\sqrt{2\log(6(m\vee1)/\alpha)}\right\}.
\]
Using copy 3, let $\mathcal R_j=\{q_j-h+1,\ldots,q_j+h-1\}$ and
\[
Q_j(k)=\sum_{t=q_j-h+1}^{k}\norm{X_t^{(3)}-\widehat\mu_{j,-}}^2
+\sum_{t=k+1}^{q_j+h}\norm{X_t^{(3)}-\widehat\mu_{j,+}}^2.
\]
The refined location is the smallest minimizer of $Q_j$ over $\mathcal R_j$. When $j>\widehat K$, set $q_j=2h$ and define all component-indexed quantities by the same formulas; this measurable padding is used only off $\{\widehat K=K\}$.

\begin{theorem*}[Upper-bound part of Theorem~\ref*{thm:change-main}, restated]
Let the boundary-inclusive minimum spacing be $\Delta$ and the minimum jump norm be $\varsigma_{\min}$. If
\[
\Delta\ge5h,
\quad
\sqrt{h/2}\,\varsigma_{\min}>3\lambda_\alpha,
\quad
\varsigma_{\min}\ge8\varepsilon_{\alpha,K},
\]
then, with probability at least $1-\alpha$,
\begin{enumerate}[label=\textup{(\roman*)}]
\item $\widehat K=K$ and every preliminary location is within $h$ of its true change;
\item every refined location has error below
\[
r_\alpha=\frac{800\nu^2}{9\varsigma_{\min}^2}\log(12K/\alpha);
\]
\item every block whose projected jump is either zero or at least $6\varepsilon_{\alpha,K}$ is classified correctly, so eligible changes are labelled strength-only, composition-only, or joint.
\end{enumerate}
If $K=0$, the screening rule returns no changes with probability at least $1-\alpha/3$.
\end{theorem*}

\begin{proof}
For each $k$, $C_h(k)-EC_h(k)\sim N_d(0,\nu^2I_d)$. Gaussian norm concentration and a union bound imply, with probability at least $1-\alpha/3$,
\[
\max_{k\in\mathcal I_h}\norm{C_h(k)-EC_h(k)}\le\lambda_\alpha.
\]
Away from all $h$-neighborhoods of changes the expectation is zero, so no retained index lies there. At each true change, the centered window contrast has norm $\sqrt{h/2}\,\varsigma_j>3\lambda_\alpha$, hence its observed norm exceeds $2\lambda_\alpha$. The spacing condition makes the retained neighborhoods disjoint under the $2h$ connectivity rule. Thus there is exactly one component per change and each preliminary location is within $h$.

On this event, the copy-2 anchor blocks are wholly within their adjacent constant segments. Gaussian norm concentration and a union bound over $2K$ blocks give, with probability at least $1-\alpha/3$,
\[
\max_j\{\norm{\widehat\mu_{j,-}-\mu^{(j)}},
\norm{\widehat\mu_{j,+}-\mu^{(j+1)}}\}
\le\varepsilon_{\alpha,K}.
\]
Therefore $\norm{\widehat\delta_j-\delta_j}\le2\varepsilon_{\alpha,K}$. Projection is contractive. A zero block remains below $3\varepsilon_{\alpha,K}$, while a block of norm at least $6\varepsilon_{\alpha,K}$ remains above that threshold. This proves attribution.

It remains to refine a fixed change. Let $\mathcal F_{12}=\sigma\{X_t^{(1)},X_t^{(2)}:1\le t\le T\}$ and condition on the joint screening-and-anchor event. The anchors and search intervals are then fixed while copy 3 remains independent. Write $a_-=\widehat\mu_{j,-}-\mu^{(j)}$, $a_+=\widehat\mu_{j,+}-\mu^{(j+1)}$, and $\delta_j=\mu^{(j+1)}-\mu^{(j)}$. For $k>\tau_j$, the criterion difference is a sum of iid Gaussian increments with
\[
m_+=\norm{\delta_j-a_-}^2-\norm{a_+}^2,
\qquad
v_+^2=4\nu^2\norm{\delta_j+a_+-a_-}^2;
\]
the reverse-order increments for $k<\tau_j$ have
\[
m_-=\norm{\delta_j+a_+}^2-\norm{a_-}^2,
\qquad
v_-^2=4\nu^2\norm{\delta_j+a_+-a_-}^2.
\]
On the anchor event and under $\varsigma_{\min}\ge8\varepsilon_{\alpha,K}$,
\[
m_\pm\ge\varsigma_j^2-2\varsigma_j\varepsilon_{\alpha,K}
\ge\frac34\varsigma_j^2,
\qquad
v_\pm^2\le\frac{25}{4}\nu^2\varsigma_j^2,
\]
so $m_\pm^2/(8v_\pm^2)\ge9\varsigma_j^2/(800\nu^2)$. The Gaussian positive-drift bound above, applied on each side of the change, yields
\[
P\{|\widehat\tau_j-\tau_j|\ge r\mid\text{anchors}\}
\le4\exp\left(-\frac{9\varsigma_j^2r}{800\nu^2}\right).
\]
With the displayed $r_\alpha$, a union bound over $K$ changes is at most $\alpha/3$. Intersecting screening, anchor, and refinement events proves the joint probability. Under $K=0$, only the screening bound is needed.
\end{proof}

\begin{remark}[Benchmark status]
The proof uses genuinely independent copies for screening, attribution, and refinement. Dividing one dependent path computationally does not establish \eqref{eq:si-three-copy}. Section~\ref{sec:si-transfer} states the stronger experiment-comparison conditions under which a transfer is valid.
\end{remark}

\section{Gaussian minimax obstructions}
\label{sec:si-minimax}
Let $\mathcal I_T=\{1,\ldots,T\}$ and $\mathcal B_T=\{0,1,\ldots,T\}$. For $a=1,\ldots,n_{\mathrm c}$, observe
\begin{equation}
X_{a,t}=x_t+\nu\xi_{a,t},
\qquad
\xi_{a,t}\stackrel{\mathrm{iid}}\sim N_d(0,I_d).
\label{eq:si-minimax-model}
\end{equation}
Let $\mathbb P_x$ be the joint law, fix $\varsigma>0$ and $\alpha_0\in(0,1/2)$, and define
\[
\mathfrak q_{J,\alpha_0}
=2\left[\Phi^{-1}\{(1-\alpha_0)^{1/J}\}\right]^2.
\]

\begin{theorem*}[Lower-bound part of Theorem~\ref*{thm:change-main}, restated]
The following finite-sample statements hold.

\begin{enumerate}[label=\textup{\arabic*.}]
\item \textbf{Detection.} Suppose there are $M\ge2$ mutually disjoint interior blocks $B_m$, each of length $\Delta$, with gaps at least $\Delta$. Let $P_0=\mathbb P_0$ and $P_m$ have mean $u\mathbf1\{t\in B_m\}$ for a fixed $\norm u=\varsigma$. Every test $\phi$ satisfies
\[
P_0\phi+\max_mP_m(1-\phi)
\ge1-\frac12\left\{
\frac{\exp(n_{\mathrm c}\Delta\varsigma^2/\nu^2)-1}{M}
\right\}^{1/2}.
\]
Thus simultaneous type-I and worst-case type-II errors at most $\alpha_0$ require
\[
\frac{n_{\mathrm c}\Delta\varsigma^2}{\nu^2}
\ge\log\{1+4M(1-2\alpha_0)^2\}.
\]

\item \textbf{Attribution.} Let $d\ge2$, $u_s=\varsigma e_1$, $u_c=\varsigma e_2$, and suppose $J$ disjoint informative blocks have known change times and $h$ post-change observations per copy. If all remaining observations arise through a label-independent Markov kernel from these blocks, then every classifier $\widehat v\in\{s,c\}^J$ satisfies
\[
\sup_vP_v(\widehat v\ne v)
\ge1-(1-p_h)^J,
\qquad
p_h=\Phi\left(-\sqrt{n_{\mathrm c}h\varsigma^2/(2\nu^2)}\right).
\]
Uniform complete-vector error at most $\alpha_0$ therefore requires
\[
\frac{n_{\mathrm c}h\varsigma^2}{\nu^2}\ge\mathfrak q_{J,\alpha_0}.
\]

\item \textbf{Localization.} Let $J,L,r\in\N$ with $1\le r<L/4$, and choose base boundaries $s_1,\ldots,s_J\in\mathcal B_T$. Define
\[
I_j(q)=\{s_j+qr+1,\ldots,s_j+qr+L\},\qquad q\in\{0,1\},
\]
and assume all these intervals lie in $\mathcal I_T$. Let $\mathcal E_j=\{s_j,s_j+r,s_j+L,s_j+r+L\}$ and require every endpoint from different macro-blocks to be separated by at least $L-r$. For $b\in\{0,1\}^J$, put
\[
x_t^{(b)}=u\ind\!\left\{t\in\bigcup_{j=1}^JI_j(b_j)\right\},
\quad \norm u=\varsigma,
\qquad
T_b=\bigcup_{j=1}^J\{s_j+b_jr,s_j+b_jr+L\}.
\]
For nonempty finite sets use the usual Hausdorff distance, extended by $d_H(\varnothing,T)=d_H(T,\varnothing)=\infty$ for $T\ne\varnothing$ and $d_H(\varnothing,\varnothing)=0$. Every measurable, possibly randomized estimator $\widehat T$ taking values in the finite subsets of $\mathcal B_T$ satisfies
\[
\sup_bP_b\{d_H(\widehat T,T_b)\ge r/2\}
\ge1-(1-p_r)^J,
\qquad
p_r=\Phi\left(-\sqrt{n_{\mathrm c}r\varsigma^2/(2\nu^2)}\right).
\]
Uniform error at most $\alpha_0$ requires
\[
\frac{n_{\mathrm c}r\varsigma^2}{\nu^2}\ge\mathfrak q_{J,\alpha_0}.
\]
The construction is nonvacuous when $\lceil\nu^2\mathfrak q_{J,\alpha_0}/(n_{\mathrm c}\varsigma^2)\rceil<L/4$ and the horizon contains the separated blocks.
\end{enumerate}

As $J\to\infty$,
\[
\mathfrak q_{J,\alpha_0}
=4\log J-2\log\log J-4\log c_{\alpha_0}-2\log(4\pi)+o(1),
\quad
c_{\alpha_0}=-\log(1-\alpha_0).
\]
\end{theorem*}

\begin{proof}
\medskip\noindent\textit{Detection.}
Let $L_m=dP_m/dP_0$ and $\bar P=M^{-1}\sum_mP_m$. Disjoint supports imply $E_0(L_mL_{m'})=1$ for $m\ne m'$, while the Gaussian likelihood-ratio identity gives $E_0L_m^2=\exp(n_{\mathrm c}\Delta\varsigma^2/\nu^2)$. Hence
\[
\chi^2(\bar P,P_0)=
\frac{\exp(n_{\mathrm c}\Delta\varsigma^2/\nu^2)-1}{M}.
\]
Total variation is at most one half the square root of chi-square. The standard testing inequality
\[
P_0\phi+\bar P(1-\phi)\ge1-\norm{P_0-\bar P}_{\rm TV}
\]
and $\bar P(1-\phi)\le\max_mP_m(1-\phi)$ prove the risk bound. Algebraic inversion gives the necessary condition.

\medskip\noindent\textit{Attribution.}
For one block, the sample mean of all $Ah$ observations has covariance $\nu^2(Ah)^{-1}I_d$. The likelihood-ratio rule between $u_s$ and $u_c$ has equal error
\[
p_h=\Phi\left(-\frac{\norm{u_s-u_c}\sqrt{Ah}}{2\nu}\right)
=\Phi\left(-\sqrt{n_{\mathrm c}h\varsigma^2/(2\nu^2)}\right).
\]
Under the uniform prior on $\{s,c\}^J$, independence across blocks makes the Bayes probability of recovering every label $(1-p_h)^J$. A parameter-free kernel for remaining observations cannot improve this Bayes risk. The maximum risk is at least the Bayes risk, proving the bound. Inverting the monotone normal cdf gives $n_{\mathrm c}h\varsigma^2/\nu^2\ge\mathfrak q_{J,\alpha_0}$.

\medskip\noindent\textit{Localization.}
In macro-block $j$, the two shifted pulse hypotheses differ on exactly $2r$ observations. The optimal binary error is
\[
p_r=\Phi\left(-\sqrt{n_{\mathrm c}r\varsigma^2/(2\nu^2)}\right).
\]
Uniform separation implies that any estimated endpoint set within Hausdorff distance $r/2$ uniquely decodes every shift bit. Under the uniform prior, exact decoding probability is at most $(1-p_r)^J$. Therefore the maximum localization error is at least $1-(1-p_r)^J$, and inversion gives the threshold.

\medskip\noindent\textit{Asymptotic threshold.}
Let $u_J=1-(1-\alpha_0)^{1/J}=c_{\alpha_0}/J+O(J^{-2})$. The standard lower-tail normal-quantile expansion
\[
\{\Phi^{-1}(u)\}^2
=2\log(1/u)-\log\log(1/u)-\log(4\pi)+o(1)
\]
as $u\downarrow0$ gives the displayed formula after multiplication by two.
\end{proof}

\begin{remark}[Scope]
This is a lower bound for \eqref{eq:si-minimax-model}. It applies to a raw outcome--network experiment only after a comparison establishes that the raw experiment is no more informative, up to a controlled deficiency error.
\end{remark}

\section{Conditional transfer from a local outcome--report experiment}
\label{sec:si-transfer}
\subsection{Local equivalence theorem}

Let $d$ be fixed, $T=T_n$, $A=A_n$, and let the deterministic parameter path $\boldsymbol\theta=(\theta_1,\ldots,\theta_T)$ belong to
\[
\Theta_n(b_n)=\{\boldsymbol\theta:\max_t\norm{\theta_t-\theta_t^0}\le b_n\}
\]
around a known deterministic baseline. Let $D_n$ be a parameter-free design-and-ancillary variable. The residual experiment observes $D_n$ and $Z_{ta}$ whose conditional law is
\[
\bigotimes_{a=1}^A
N_{m_{ta}}\{\mu_{ta}(\theta_t;\mathcal F_{t-1}^Z),I_{m_{ta}}\}.
\]
Put $B_{ta}=D_\theta\mu_{ta}(\theta_t^0;\mathcal F_{t-1}^Z)$.

\begin{theorem*}[Local Gaussian-sequence equivalence under common information]
Assume, uniformly over dates, copies, histories, and $\Theta_n(b_n)$:
\begin{enumerate}[label=\textup{(C\arabic*)}]
\item for a known $\nu_n>0$ and known block-diagonal $\mathcal I_0=\diag(i_\beta,\mathcal I_\eta)\succ0$,
\[
B_{ta}^\trans B_{ta}=\nu_n^{-2}\mathcal I_0;
\]
\item the residual mean has the quadratic remainder
\[
\norm{\mu_{ta}(\theta;h)-\mu_{ta}(\theta_t^0;h)-B_{ta}(h)(\theta-\theta_t^0)}
\le L\sqrt n\norm{\theta-\theta_t^0}^2;
\]
\item $A_nT_nnb_n^4\to0$.
\end{enumerate}
Let $\mathcal G_n$ observe
\[
X_t^{(a)}=\mathcal I_0^{1/2}\theta_t+\nu_n\xi_t^{(a)},
\qquad \xi_t^{(a)}\stackrel{\mathrm{iid}}\sim N_d(0,I_d).
\]
Then the Le Cam distance satisfies
\[
\Delta_{\rm LC}(\mathcal E_n^Z,\mathcal G_n)
\le\rho_n:=\frac L2\sqrt{A_nT_nnb_n^4}\to0.
\]
If the raw outcome--report experiment and the residual experiment admit parameter-free Markov kernels that reconstruct each other exactly, the same bound holds for the raw experiment. If the residual mean is exactly affine, residual and Gaussian experiments are exactly equivalent at finite sample.
\end{theorem*}

\begin{proof}
Linearize the conditional means at the baseline. In the linear experiment, the statistic
\[
\widetilde\theta_{ta}=\theta_t^0+
(B_{ta}^\trans B_{ta})^{-1}B_{ta}^\trans(Z_{ta}-\mu_{ta}^0)
\]
and $X_t^{(a)}=\mathcal I_0^{1/2}\widetilde\theta_{ta}$ produce the Gaussian sequence because the common-information identity makes the transformed conditional covariance $\nu_n^2I_d$. The conditional law is history independent, so iterated conditioning makes the transformed noises iid.

Conversely, starting from the Gaussian sequence, independently draw $D_n$ and reconstruct histories recursively. Let $P_{ta}=B_{ta}(B_{ta}^\trans B_{ta})^{-1}B_{ta}^\trans$, draw $U_{ta}\sim N(0,I-P_{ta})$, and set
\[
Z_{ta}^*=\mu_{ta}^0+B_{ta}(\mathcal I_0^{-1/2}X_t^{(a)}-\theta_t^0)+U_{ta}.
\]
This kernel is parameter free. Its conditional mean and covariance equal those of the linear experiment, so the two linear experiments are exactly equivalent.

For the nonlinear experiment, conditional Gaussian relative entropy between the true and linearized kernels is one half the squared remainder. The chain rule for relative entropy and (C2) give
\[
\operatorname{KL}(P^Z_{\boldsymbol\theta},P^{Z,\rm lin}_{\boldsymbol\theta})
\le\frac{L^2}{2}A_nT_nnb_n^4.
\]
Pinsker's inequality gives total variation at most $\rho_n$. The identity kernel bounds both deficiencies, and the triangle inequality with exact linear equivalence proves the result. Exact raw/residual reconstruction adds zero distance by another triangle inequality.
\end{proof}

\begin{corollary*}[Transfer of change procedures and lower bounds]
For either the residual experiment, or the raw experiment under exact parameter-free reconstruction, Gaussian upper guarantees transfer with success probability reduced by at most $\rho_n$. Gaussian lower bounds transfer with risk reduced by at most the reverse deficiency, also bounded by $\rho_n$. Strength and composition jumps are measured in the common information metric:
\[
s_j=\sqrt{i_\beta}|\Delta\beta_j|,
\qquad
c_j=\norm{\mathcal I_\eta^{1/2}\Delta\eta_j}.
\]
\end{corollary*}

\begin{proof}
Compose any Gaussian procedure with the kernel from the residual/raw experiment to the Gaussian experiment. Total variation changes the probability of its joint success event by at most the forward deficiency. For a lower bound, compose any raw procedure with a kernel from the Gaussian experiment to the raw experiment; a $[0,1]$ loss changes by at most the reverse deficiency. Taking infima over kernels proves both claims.
\end{proof}

\begin{remark}[Why this is not automatic]
Condition (E4) gives only eigenvalue bounds; it does not imply exactly common, block-diagonal information across dates and copies. Unknown whitening, estimated covariance, cross-fitted pilots, reporter nuisances, and a random support may also prevent parameter-free reconstruction. If a separate comparison controls these steps by $\zeta_n$, the transfer error becomes $\rho_n+\zeta_n$. Without such verification, the Gaussian upper and lower results in Sections~\ref{sec:si-change} and \ref{sec:si-minimax} remain benchmarks.
\end{remark}

\section{A verified transfer instance: the pinned-design experiment}
\label{sec:si-instance}
The transfer theorem of Section~\ref*{sec:si-transfer} was stated as a
conditional result: its hypotheses (C1)--(C3) were not verified for any
concrete outcome--report design. This section closes that gap for one
explicit design. The design is not universal: it is a designed replication
experiment, not an observational autoregression. Its value is that every
hypothesis of the transfer theorem becomes a finite, checkable
computation, so the Gaussian change guarantees and lower obstructions of
Section~\ref*{sec:si-change} and \ref*{sec:si-minimax} apply to a fully
specified data-generating mechanism with an explicit deficiency bound.

\subsection{The pinned-design replication experiment}

Fix a support $\mathcal E$ on $N$ nodes with partner sets as in the main
text, a linear chart $m(\eta)=\Psi\eta$ with $C^\trans\Psi=0$ and
$\rank\Psi=q$, and known covariances $\Sigma=\sigma_y^2I$,
$\Omega$ (mirror-pair structure allowed). Prespecify:
a fixed exposure basket $\bar y\in\R^N$ (a reference lag vector held fixed by
design, not the realized lagged outcome); fixed nuisance designs
$X$ (outcome) and $U=[AC,B]$ (reports); and a baseline
$\theta^0=(\beta_0,\eta_0)$ interior to the parameter space. At each date
$t\le T_n$ and replicate $a\le A_n$, independent draws
\begin{align*}
Y_t^{(a)}&=X\gamma+\beta_t\,g(\eta_t)+\sigma_y\xi_t^{(a)},
\qquad g(\eta)=W(\eta)\bar y,\\
z_t^{(a)}&=U\lambda+A\Psi\eta_t+\Omega^{1/2}\upsilon_t^{(a)},
\end{align*}
with iid standard Gaussian innovations. Let $r=M_{L_yX}L_yg(\eta_0)$,
$H=M_{L_yX}L_yG(\eta_0)$, $Q=M_{L_zU}L_zA\Psi$ as in
Theorem~\ref*{thm:identification-main}, all evaluated once at the baseline
(they are date- and replicate-independent because the design is pinned).

\begin{description}[style=nextline]
\item[(P1) Identification at baseline.] $r\ne0$ and
$K_c+\beta_0^2H^\trans P_r^\perp H\succ0$.
\item[(P2) Designed orthogonality.] $H^\trans r=0$.
\item[(P3) Local neighborhood.] The path lies in
$\Theta_n(b_n)=\{\max_t\norm{\theta_t-\theta^0}\le b_n\}$ with
$A_nT_n\bar nb_n^4\to0$, where $\bar n=N+2|\mathcal E|$ is the per-replicate
observation count.
\end{description}

Condition (P2) is a design equation: $r$ and each column of $H$ are linear in
$\bar y$, so $H^\trans r=0$ imposes $q$ quadratic equations on the $N$
coordinates of $\bar y$; generically these have solutions for $N>q+2$, and a
solution can be found and verified numerically before any data are collected.
The replication code accompanying the paper exhibits such a basket for the
worked design of Section~\ref*{sec:charts-main} and verifies (P1)--(P2) to
machine precision.

\subsection{Verification of the transfer hypotheses}

\begin{proposition*}[Proposition~\ref*{prop:transfer-instance-main}, restated]
Assume (P1)--(P3). Let
$C_\Psi=\max_{e}\norm{\Psi_{e\cdot}}_2$,
$\bar\beta=|\beta_0|+b_n$, and
\[
L_\star
=\frac{1}{\sigma_y}
\left(2\,\norm{G(\eta_0)}_{\op}
+8\,\bar\beta\,C_\Psi^2\,\norm{\bar y}_\infty\sqrt N\right)
\Big/\sqrt{\bar n} .
\]
Then the pinned-design experiment satisfies (C1)--(C3) of
Section~\ref*{sec:si-transfer} with
\[
\nu_n=1,
\qquad
\mathcal I_0=
\begin{pmatrix}
r^\trans r&0\\
0&\beta_0^2H^\trans H+K_c
\end{pmatrix}\succ0,
\qquad
L=L_\star ,
\]
and admits an exact parameter-free reconstruction kernel. Consequently
\[
\Delta_{\rm LC}\big(\mathcal E_n^{\rm pinned},\mathcal G_n\big)
\le\rho_n=\frac{L_\star}2\sqrt{A_nT_n\bar nb_n^4}\to0,
\]
and the three-copy change guarantees and the minimax lower bounds of
Theorem~\ref*{thm:change-main} hold for the pinned-design experiment with all
jumps measured in the $\mathcal I_0$ metric, with success probabilities and
risks shifted by at most $\rho_n$.
\end{proposition*}

\begin{proof}
\textit{(C1).} After whitening and residualizing the fixed nuisance
designs, a date- and replicate-independent linear operation, the residual
mean of one replicate is
\[
\mu(\theta)=
\begin{pmatrix}
\sigma_y^{-1}M_{L_yX}L_y\,\beta g(\eta)\\
M_{L_zU}L_z A\Psi\eta
\end{pmatrix},
\qquad
B:=D_\theta\mu(\theta^0)=
\begin{pmatrix}
r&\beta_0H\\
0&Q
\end{pmatrix}.
\]
$B$ is the same matrix for every $(t,a)$ because the design is pinned. Then
$B^\trans B$ equals $\mathcal I_0$ displayed above by (P2), and (P1) makes it
positive definite. This is (C1) with $\nu_n=1$.

\textit{(C2).} The report block of $\mu$ is exactly affine in $\theta$, so
its remainder vanishes. For the outcome block, write
$f(\theta)=\beta g(\eta)$. Second derivatives:
$D^2_{\beta\eta}f=G(\eta)$ and $D^2_{\eta\eta}f=\beta D^2_\eta g(\eta)$. For
the row-softmax chart, differentiating
$D_\eta W_{ij}=W_{ij}(\Psi_{ij}-\bar\Psi_i)^\trans$ once more gives, for unit
vectors $u,v$,
\[
u^\trans D^2_\eta W_{ij}(\eta)v
=W_{ij}\Big\{(\Psi_{ij}-\bar\Psi_i)^\trans u\,(\Psi_{ij}-\bar\Psi_i)^\trans v
-\sum_kW_{ik}\,(\Psi_{ik}-\bar\Psi_i)^\trans u\;\Psi_{ik}^\trans v\Big\},
\]
so, using $\norm{\Psi_{ij}-\bar\Psi_i}\le2C_\Psi$,
$|\Psi_{ik}^\trans v|\le C_\Psi\norm v$, and $\sum_kW_{ik}=1$,
\[
\big|u^\trans D^2_\eta W_{ij}(\eta)v\big|
\le W_{ij}\left\{4C_\Psi^2+2C_\Psi^2\right\}\norm u\norm v
=6C_\Psi^2W_{ij}\norm u\norm v
\le8C_\Psi^2W_{ij}\norm u\norm v,
\]
using $\norm{\Psi_{ij}-\bar\Psi_i}\le2C_\Psi$. Summing within a row against
$\bar y$ and stacking rows,
$\norm{D^2_\eta g(\eta)}\le8C_\Psi^2\norm{\bar y}_\infty\sqrt N$ uniformly on
the ball. The whitened residualizing operator has operator norm at most
$\sigma_y^{-1}$ on the outcome block. A second-order Taylor expansion with
integral remainder then gives, for $\theta\in\Theta_n(b_n)$,
\[
\norm{\mu(\theta)-\mu(\theta^0)-B(\theta-\theta^0)}
\le\frac{1}{2\sigma_y}
\left(2\norm{G(\eta_0)}_{\op}+8\bar\beta C_\Psi^2\norm{\bar y}_\infty\sqrt N\right)
\norm{\theta-\theta^0}^2
= \frac{L_\star}{2}\sqrt{\bar n}\,\norm{\theta-\theta^0}^2 ,
\]
where the extra $\norm{G(\eta)-G(\eta_0)}$ term along the segment is at
most $8C_\Psi^2\norm{\bar y}_\infty\sqrt N\,b_n$ by the same softmax
bound and is absorbed into the first constant for all $n$ large, since
$b_n\to0$. This is (C2) with
$L=L_\star$ after matching the normalization
$L\sqrt{\bar n}$ used there.

\textit{(C3).} Immediate from (P3).

\textit{Reconstruction.} The design variables are deterministic, the
covariances known, and the whitening--residualizing map is a fixed linear
surjection with known kernel; the Gaussian-sequence reconstruction kernel of
the transfer proof (draw the annihilated components independently from their
known Gaussian law and invert the fixed linear map) is parameter free and
exact. The transfer theorem and its corollary then give the displayed
deficiency bound and the transfer of upper and lower results.
\end{proof}

\begin{remark}[Scope of the instance]
Three hypotheses are substantive and delimit what has been verified.
First, the exposure is pinned at a prespecified $\bar y$: this is a designed
replication (repeated independent panels measured against a fixed reference
basket), not the observational autoregression in which $y_{t-1}$ is random
and serially dependent; for the latter, common information across dates fails
in general and the transfer question remains open, as stated in
Section~\ref*{sec:si-transfer}. Second, covariances are known; estimated
covariances add a deficiency term of order $\delta_n\sqrt{A_nT_n\bar n}\,b_n$
that must be tracked separately. Third, (P2) is a genuine design choice; without
it the strength/composition decomposition is still identified but the
information matrix is not block diagonal, and attribution statements must be
made in the oblique $\mathcal I_c$ metric. Within these limits, every
constant in the change guarantees of Theorem~\ref*{thm:change-main} is now a
number that can be computed before running the experiment; the simulation
section evaluates how conservative those numbers are.
\end{remark}

\section{Specialized observation boundaries and diagnostics}
\label{sec:si-boundaries}
\subsection{Independent Gaussian censoring}

Suppose independent latent reports satisfy $Z_m^*\sim N\{\mu_m(\zeta,\eta),\sigma_m^2\}$ with known finite thresholds $c_m$ and known scales. We observe the exact value only above $c_m$ and otherwise a censored symbol. Put $a_m=(c_m-\mu_m)/\sigma_m$ and
\[
\omega_m=\sigma_m^{-2}\left\{
1-\Phi(a_m)+a_m\phi(a_m)+\frac{\phi(a_m)^2}{\Phi(a_m)}
\right\},
\qquad W_c=\diag(\omega_1,\ldots,\omega_M).
\]

\begin{proposition*}[Information under independent Gaussian censoring]
Assume an interior DQM point and locally constant nuisance-score rank.
Let $D=D_\eta\mu$ and $U=D_\zeta\mu$. After profiling $\zeta$, the efficient information is
\[
K_{\rm cen}=D^\trans W_c^{1/2}M_{W_c^{1/2}U}W_c^{1/2}D.
\]
For finite $a_m$, every $\omega_m>0$, so composition is regularly identified exactly when
\[
Dh\notin\col(U)\quad\text{for every }h\ne0.
\]
Finite censoring preserves exact first-order rank. Holding $\sigma_m$ fixed, $\omega_m\to0$ as $a_m\to\infty$, so indispensable directions can become arbitrarily weak.
\end{proposition*}

\begin{proof}
The mean score for coordinate $m$ is
\[
s_m=\Delta_m\frac{Z_m-\mu_m}{\sigma_m^2}
-(1-\Delta_m)\frac{\phi(a_m)}{\sigma_m\Phi(a_m)}.
\]
Direct truncated-normal integration gives $E(s_m)=0$ and $E(s_m^2)=\omega_m$. Independence makes $\Var(s)=W_c$. Projecting the whitened composition derivative $W_c^{1/2}D$ off $W_c^{1/2}U$ gives the displayed Gram matrix. Since $W_c^{1/2}$ is invertible at finite thresholds, its null directions are exactly those with $Dh\in\col(U)$. Finally,
\[
\sigma_m^2\omega_m
=\int_{a_m}^{\infty}x^2\phi(x)\,dx
+\frac{\phi(a_m)^2}{\Phi(a_m)}\longrightarrow0
\]
as $a_m\to\infty$, proving weak information under severe censoring.
\end{proof}

Unknown scales and correlated censored reports require projection off their complete observed-data scores. One cannot replace $W_c$ by the inverse latent covariance after censoring.

\subsection{Structural zeros}

\begin{proposition*}[Sharp sets under pooled zeros]
Let $P_F=\pi\delta_0+(1-\pi)H$, with $H$ supported on $(0,\infty)$, and record $O=0$ for $F\le c$ and $O=F$ for $F>c$, where $c>0$ is known. Put $p_0=P(O=0)$ and let $Q$ be the observed positive-tail subprobability measure. Then
\[
\mathcal I_\pi=[0,p_0].
\]
If $T=\int_{(c,\infty)}xQ(dx)<\infty$, then
\[
\mathcal I_{E(F)}=[T,T+cp_0].
\]
Both sets are sharp.
If $T=\infty$, every compatible latent distribution has infinite mean.
\end{proposition*}

\begin{proof}
Every compatible law has $\pi\le p_0$. Conversely, for any $\widetilde\pi\in[0,p_0]$, assign mass $p_0-\widetilde\pi$ arbitrarily on $(0,c]$ and use $Q$ above $c$; this reproduces the complete observed law and has structural-zero mass $\widetilde\pi$. For the mean, the unobserved below-limit contribution lies between zero and $cp_0$. Every point is attained by placing the needed mass at $c$ and assigning the remainder to the zero atom.
If $T=\infty$, nonnegativity implies $E(F)\ge T=\infty$ for every compatible law.
\end{proof}

\subsection{Nonignorable selection}

\begin{proposition*}[Unrestricted selection destroys identification]
For any prescribed reporting rate $\rho\in(0,1)$, there are two distinct Gaussian means and two latent-report-dependent selection functions that induce the same observed law of $(R,RZ^*)$ and satisfy $P(R=1)=\rho$.
\end{proposition*}

\begin{proof}
Choose two distinct, sufficiently close everywhere-positive Gaussian densities $f_0,f_1$ so that $I=\int\min(f_0,f_1)>\rho$. Put
\[
q(z)=\frac\rho I\min\{f_0(z),f_1(z)\},
\qquad
\pi_k(z)=q(z)/f_k(z).
\]
Then $0\le\pi_k<1$, the reported-value subdistribution is $q(z)dz$ under both models, and the nonreporting atom is $1-\rho$. Hence the observed laws coincide although the Gaussian means differ.
\end{proof}

\begin{corollary*}[Specified-selection score criterion]
Under a DQM observed-data selection model, let $S_\eta^o$ be the composition score and $\mathcal T_\nu$ the closed nuisance tangent space including selection, flow levels, biases, and covariance parameters. Then composition is regularly first-order identified exactly when
\[
K_{\rm sel}:=E\left[
\{S_\eta^o-\Pi_{\mathcal T_\nu}S_\eta^o\}
\{S_\eta^o-\Pi_{\mathcal T_\nu}S_\eta^o\}^\trans
\right]\succ0.
\]
If a finite nuisance score spans the complete tangent space, this equals
\[
J_{\eta\eta}-J_{\eta\nu}J_{\nu\nu}^\dagger J_{\nu\eta}.
\]
\end{corollary*}

\begin{proof}
For every $h$, $h^\trans K_{\rm sel}h$ is the squared $L_2$ distance from $h^\trans S_\eta^o$ to the closed nuisance tangent space. It vanishes exactly when the target direction can be cancelled by regular nuisance perturbations. The finite-dimensional formula is the covariance of the Moore--Penrose least-squares residual; redundant nuisance directions have zero score and therefore zero cross-information.
\end{proof}

\subsection{Support change}

\begin{proposition*}[No intrinsic support--intensity attribution]
If dyad $e$ is absent at $t-1$ and present at $t$, an unrestricted positive counterfactual pre-entry weight is not identified, nor is the corresponding log intensive change. More generally, endpoint row-normalized networks on different supports do not determine a unique decomposition into support and intensive-weight effects.
\end{proposition*}

\begin{proof}
Changing an off-support counterfactual pre-entry weight leaves both datewise observed laws unchanged but changes the proposed intensive target. Extend both active weight vectors positively to the union support. The support-first identity is
\[
W(\bar q_1,S_1)-W(\bar q_0,S_0)
=\{W(\bar q_1,S_1)-W(\bar q_0,S_1)\}
+\{W(\bar q_0,S_1)-W(\bar q_0,S_0)\},
\]
while the weights-first identity is
\[
W(\bar q_1,S_1)-W(\bar q_0,S_0)
=\{W(\bar q_1,S_1)-W(\bar q_1,S_0)\}
+\{W(\bar q_1,S_0)-W(\bar q_0,S_0)\}.
\]
For one row with $S_0=\{a\}$, $S_1=\{a,b\}$, $q_{0a}=q_{1a}=q_{1b}=1$, and arbitrary unobserved extension $\bar q_{0b}=x>0$, the support-first component is $(-x/(1+x),x/(1+x))$, whereas the weights-first component is $(-1/2,1/2)$. These differ whenever $x\ne1$, although every $x>0$ gives the same date-0 observed law. Hence neither the extension nor the decomposition order is data determined.
\end{proof}

The common-support construction in Section~\ref{sec:si-observation} is a well-defined alternative; it is not the same estimand as full-support composition.

\subsection{Reporter-cycle diagnostics and common bias}

Let paired-report discrepancies satisfy $d=Gb+e$, where $G$ is an oriented edge-by-vertex incidence matrix and $e\sim N(0,\Omega_d)$ with known $\Omega_d\succ0$. Put $L_d=\Omega_d^{-1/2}$, $A=L_dG$, and $M_A=I-A(A^\trans A)^\dagger A^\trans$.

\begin{proposition*}[Exact cycle test]
\[
T_{\rm bias}=\norm{M_AL_dd}^2\sim\chi^2_{M-\rank(G)}.
\]
Moreover, $E(d)\in\col(G)$ if and only if every signed population cycle sum is zero. If the reporter graph is a forest, there is no overidentifying cycle restriction.
\end{proposition*}

\begin{proof}
Whitening gives $L_dd=Ab+\varepsilon$ with $\varepsilon\sim N(0,I)$. Thus $M_AL_dd=M_A\varepsilon$. The projector $M_A$ is symmetric idempotent of rank $M-\rank(G)$, proving the chi-square law. The cycle space is $\ker(G^\trans)$ and $\col(G)=\ker(G^\trans)^\perp$, proving the population restriction.
\end{proof}

\begin{proposition*}[Common dyad bias remains invisible]
Suppose the joint law of $(u,v)$ is fixed and parameter-independent under the transformations below, $E(u)=E(v)=0$, and
\[
z^E=\ell+c+G_Eb+u,
\qquad
z^I=\ell+c+G_Ib+v,
\]
where $c$ is an unrestricted dyad-specific bias common to both reports and $G_E-G_I=G$ is the incidence design used in the cycle test. Then the discrepancy law depends on neither $\ell$ nor $c$, and $E(z^E-z^I)=Gb$, so every population cycle restriction may hold while composition is unidentified. Any admissible perturbation $a$ can be absorbed through $(\ell,c)\mapsto(\ell+a,c-a)$. Row-normalized composition changes exactly when $a$ is nonconstant within some eligible row.
\end{proposition*}

\begin{proof}
Subtraction gives $z^E-z^I=(G_E-G_I)b+(u-v)$, so $\ell$ and $c$ cancel. The displayed transformation leaves both report means unchanged. Within a row,
\[
\frac{W_{ij}(\ell+a)}{W_{ik}(\ell+a)}
=\frac{W_{ij}(\ell)}{W_{ik}(\ell)}e^{a_{ij}-a_{ik}},
\]
so the softmax is unchanged if and only if the perturbation is row constant.
\end{proof}

Passing the cycle test can reject an additive discrepancy model, but it cannot validate the absence of common dyad bias. Estimated covariance also removes the exact finite-sample chi-square law unless an independent conditional argument or justified calibration is supplied.

\end{document}